\documentclass[12pt]{article}
\usepackage{amsmath}
\usepackage{amsfonts}
\usepackage{makeidx}
\usepackage{amssymb}
\usepackage{times}
\usepackage{anysize}

\marginsize{3cm}{3cm}{2cm}{2cm}
\newtheorem{satz}{Theorem}[section]
\newtheorem{definition}[satz]{Definition}
\newtheorem{lemma}[satz]{Lemma}
\newtheorem{koro}[satz]{Corollary}
\newtheorem{bemerkung}[satz]{Remark}

\newtheorem{notation}[satz]{Notation}
\newenvironment{proof}{\par\noindent {\it Proof:} \hspace{7pt}}{\hfill\hbox{\vrule width 7pt depth 0pt height 7pt}
\par\vspace{10pt}}

\begin{document}

\title{Microscopic Conductivity of Lattice Fermions at Equilibrium -- Part
I: Non--Interacting Particles}
\author{J.-B. Bru \and W. de Siqueira Pedra \and C. Hertling}
\date{\today }
\maketitle

\begin{abstract}
We consider free lattice fermions subjected to a static bounded potential
and a time-- and space--dependent electric field. For any bounded convex
region $\mathcal{R}\subset \mathbb{R}^{d}$ ($d\geq 1$) of space, electric
fields $\mathcal{E}$ within $\mathcal{R}$ drive currents. At leading order,
uniformly with respect to the volume $\left\vert \mathcal{R}\right\vert $ of
$\mathcal{R}$ and the particular choice of the static potential, the
dependency on $\mathcal{E}$ of the current is linear and described by a
conductivity (tempered, operator--valued) distribution. Because of the
positivity of the heat production, the real part of its Fourier transform is
a positive measure, named here (microscopic) conductivity measure of $%
\mathcal{R}$, in accordance with Ohm's law in Fourier space. This finite
measure is the Fourier transform of a time--correlation function of current
fluctuations, i.e., the conductivity distribution satisfies Green--Kubo
relations. We additionally show that this measure can also be seen as the
boundary value of the Laplace--Fourier transform of a so--called quantum
current viscosity. The real and imaginary parts of conductivity
distributions are related to each other via the Hilbert transform, i.e.,
they satisfy Kramers--Kronig relations. At leading order, uniformly with
respect to parameters, the heat production is the classical work performed
by electric fields on the system in presence of currents. The conductivity
measure is uniformly bounded with respect to parameters of the system and it
is never the trivial measure $0\,\mathrm{d}\nu $. Therefore, electric fields
generally produce heat in such systems. In fact, the conductivity measure
defines a quadratic form in the space of Schwartz functions, the
Legendre--Fenchel transform of which describes the resistivity of the
system. This leads to Joule's law, i.e., the heat produced by currents is
proportional to the resistivity and the square of currents.
\end{abstract}

%TCIMACRO{\TeXButton{\tableofcontents }{\tableofcontents}}%
%BeginExpansion
\tableofcontents%
%EndExpansion

\section{Introduction}

The present paper belongs to a succession of works on Ohm and Joule's laws
starting with \cite{OhmI}, where heat production of free lattice fermions
subjected to a static bounded potential and a time-- and space--dependent
electric field has been analyzed in detail.

Note that there are mathematical results, previous to \cite{OhmI}, on
transport properties of different models that yield Ohm's law in some form.
The closest results to ours are \cite{JMP-autre,JMP-autre2,Annale}, where
the concept of a \textquotedblleft conduc%
%TCIMACRO{\TeXButton{\-}{\-}}%
%BeginExpansion
\-%
%EndExpansion
tivity measure\textquotedblright\ is introduced for a system of
non--interacting fermions subjected to a random potential. \cite{Cornean}
proves Ohm's law for free fermions in graphene--like materials subjected to
space--homogeneous time--periodic electric fields. In \cite%
{FroehlichMerkliUeltschi}, Ohm's law in the DC--regime is stated for contact
interactions between two quasi--free reservoirs with the steady current
being a function of the chemical potential difference between the
reservoirs. This corresponds to an open quantum system approach to transport
properties as in \cite{JOP1,JOP2,JOP3,CMP}. In particular, in contrast to
our approach, the conductivity derived in \cite{FroehlichMerkliUeltschi} is
not a bulk property. We rather consider the current response of a closed
infinite system of fermions to time--dependent electric fields so that
properties of bulk coefficients can be studied in the AC--regime \cite%
{OhmIII}. For previous results on heat production in infinite
non--autonomous (closed) quantum systems, see, e.g., \cite{FMSU}.

Ohm's law is also valid at microscopic scales. Indeed, in a
recent work \cite{Ohm-exp} the authors experimentally verified the validity
of Ohm's law at the atomic scale for a purely quantum system. Such a
behavior was unexpected \cite{Ohm-exp2}:\bigskip

\noindent \textit{...In the 1920s and 1930s, it was expected that classical
behavior would operate at macroscopic scales but would break down at the
microscopic scale, where it would be replaced by the new quantum mechanics.
The pointlike electron motion of the classical world would be replaced by
the spread out quantum waves. These quantum waves would lead to very
different behavior. ... Ohm's law remains valid, even at very low
temperatures, a surprising result that reveals classical behavior in the
quantum regime. }\smallskip

\hfill \lbrack D. K. Ferry, 2012]\bigskip\

One aim of the present paper is to establish a form of Ohm and Joule's laws
at \emph{microscopic} scales, by introducing the concept of microscopic\emph{%
\ conductivity distributions} for bounded regions $\mathcal{R}\subset
\mathbb{R}^{d}$ of space, whose existence and basic properties follow from
rather general properties of fermion systems at equilibrium.

More precisely, consider any arbitrary smooth compactly supported function $%
\mathcal{E}:\mathbb{R\rightarrow R}$ which yields a space--homogeneous
electric field $\mathbf{1}[x\in \mathcal{R}]\,\mathcal{E}_{t}\,\vec{w}$ at
time $t\in \mathbb{R}$ oriented along the normalized vector $\vec{w}%
:=(w_{1},\ldots ,w_{d})\in \mathbb{R}^{d}$ in some open convex domain $%
\mathcal{R}\subset \mathbb{R}^{d}$. For free lattice fermions at thermal
equilibrium subjected to a static bounded potential, we show the existence
of finite symmetric measures $\{\mu _{\mathcal{R}}\}_{\mathcal{R}\subset
\mathbb{R}^{d}}$ on $\mathbb{R}$ taking values in the set $\mathcal{B}_{+}(%
\mathbb{R}^{d})$ of positive linear operators on $\mathbb{R}^{d}$ such that,
uniformly with respect to (w.r.t.) the volume $|\mathcal{R}|$ and the choice
of the static potential, the induced mean current response $\mathbb{J}_{%
\mathcal{R}}^{(\mathcal{E})}\left( t\right) $ at time $t$ within $\mathcal{R}
$ obeys:%
\begin{equation*}
\mathbb{J}_{\mathcal{R}}^{(\mathcal{E})}\left( t\right) =\frac{1}{2}\int_{%
\mathbb{R}}\mathcal{\hat{E}}_{\nu }^{(t)}\mu _{\mathcal{R}}\left( \mathrm{d}%
\nu \right) \vec{w}+\frac{i}{2}\int_{\mathbb{R}}\mathbb{H}(\mathcal{\hat{E}}%
^{(t)})\left( \nu \right) \mu _{\mathcal{R}}\left( \mathrm{d}\nu \right)
\vec{w}+\mathcal{O}\left( \left\Vert \mathcal{E}\right\Vert _{\infty
}^{2}\right) \ ,
\end{equation*}%
with $\mathcal{\hat{E}}$ being the Fourier transform of $\mathcal{E}$, $%
\mathcal{\hat{E}}_{\nu }^{(t)}:=\mathrm{e}^{i\nu t}\mathcal{\hat{E}}_{\nu }$%
, and where $\mathbb{H}$ is the Hilbert transform. This expression allows us
to define $\mathcal{B}(\mathbb{R}^{d})$--valued tempered distributions $\mu
_{\mathcal{R}}^{\Vert },\mu _{\mathcal{R}}^{\bot }$ satisfying
Kramers--Kronig relations and such that%
\begin{equation*}
\mathbb{J}_{\mathcal{R}}^{(\mathcal{E})}\left( t\right) =\left( \mu _{%
\mathcal{R}}^{\Vert }(\mathcal{\hat{E}}^{(t)})+i\mu _{\mathcal{R}}^{\bot }(%
\mathcal{\hat{E}}^{(t)})\right) \vec{w}+\mathcal{O}\left( \left\Vert
\mathcal{E}\right\Vert _{\infty }^{2}\right) \ ,
\end{equation*}%
see Equations (\ref{Kramers--Kronig relations})--(\ref{current distribution}%
). By $\mathcal{B}(\mathbb{R}^{d})$--valued tempered distributions, we mean
a map from the space $\mathcal{S}\left( \mathbb{R};\mathbb{C}\right) $ of
Schwartz functions to the space $\mathcal{B}(\mathbb{R}^{d})$ of linear
operators on $\mathbb{R}^{d}$ where each entry w.r.t. the canonical
orthonormal basis of $\mathbb{R}^{d}$ is a (tempered) distribution. $\mu _{%
\mathcal{R}}^{\Vert }$ is the linear response in--phase component of the
total conductivity in Fourier space and $\mu _{\mathcal{R}}^{\Vert }+i\mu _{%
\mathcal{R}}^{\bot }$ is named the (microscopic, $\mathcal{B}(\mathbb{R}%
^{d}) $--valued) \emph{conductivity distribution} of the region $\mathcal{R}$%
, while $\mu _{\mathcal{R}}$ is the (in--phase) conductivity measure,
similar to \cite{Annale}.

We show four important properties of $\mu _{\mathcal{R}}$:

\begin{itemize}
\item It is the Fourier transform of a time--correlation function of current
fluctuations, i.e., the microscopic conductivity measures satisfy \emph{%
Green--Kubo relations}. See Theorem \ref{lemma sigma pos type copy(4)} and
Equation (\ref{green kubo}).

\item $\left\Vert \mu _{\mathcal{R}}\left( \mathbb{R}\right) \right\Vert _{%
\mathrm{op}}$ is uniformly bounded w.r.t. $\mathcal{R}$ and $\mu _{\mathcal{R%
}}\left( \mathbb{R}\backslash \{0\}\right) >0$. See Theorem \ref{lemma sigma
pos type copy(4)}.

\item If a cyclic representation of the equilibrium state of the system is
denoted by $(\mathcal{H},\pi ,\Psi )$, then $\mu _{\mathcal{R}}$ is the
spectral measure of the Liouvillean $\mathcal{L}$ of the system w.r.t. a
vector $\Psi _{\mathcal{R}}\in \mathcal{H}$. We show that $\mu _{\mathcal{R}%
}\left( \mathbb{R}\backslash \{0\}\right) =0$ if and only if $\Psi _{%
\mathcal{R}}\in \ker \mathcal{L}$. Thus, $\mu _{\mathcal{R}}\left( \mathbb{R}%
\backslash \{0\}\right) >0$ is equivalent to the geometric condition $\Psi _{%
\mathcal{R}}\notin \ker \mathcal{L}$ which is easily verified in the present
case. See Equation (\ref{von braun}), Theorem \ref{lemma sigma pos type
copy(2)} and Corollary \ref{Corollary Stationarity copy(1)}.

\item $\mu _{\mathcal{R}}$ can also be constructed on $\mathbb{R}\backslash
\{0\}$ as the boundary value of the Laplace--Fourier transform of a
so--called quantum current viscosity. See Equations (\ref{quantum viscosity}%
) and (\ref{quantum viscosity bis bis}) as well as Theorem \ref{lemma sigma
pos type copy(6)}.
\end{itemize}

If the first law of thermodynamics\ holds true for the system under
consideration, then the existence and basic properties of the microscopic
conductivity measures are, roughly speaking, consequences of very general
properties of KMS (Kubo--Martin--Schwinger) states and decay bounds of space--time correlation
functions of the equilibrium state.

Indeed, the existence of the (in--phase) conductivity measure is related to
the positivity of the heat production induced by the electric field on the
fermion system at thermal equilibrium. When the so--called AC--condition%
\begin{equation}
\int_{\mathbb{R}}\mathcal{E}_{t}\mathrm{d}t=0  \label{AC condition}
\end{equation}%
holds, the total heat production per unit of volume (of $\mathcal{R}$) as
the electric field is switched off turns out to be equal to
\begin{equation*}
\int_{\mathbb{R}}\mathcal{\hat{E}}_{\nu }\left\langle \vec{w},\mu _{\mathcal{%
R}}\left( \mathrm{d}\nu \right) \vec{w}\right\rangle +\mathcal{O}\left(
\left\Vert \mathcal{E}\right\Vert _{\infty }^{3}\right) =\int_{\mathbb{R}%
}\left\langle \mathcal{E}_{t}\vec{w},\mu _{\mathcal{R}}^{\Vert }(\mathcal{%
\hat{E}}^{(t)})\vec{w}\right\rangle \mathrm{d}t+\mathcal{O}\left( \left\Vert
\mathcal{E}\right\Vert _{\infty }^{3}\right) \ ,
\end{equation*}%
uniformly w.r.t. $|\mathcal{R}|$ and the choice of the static potential.
Since
\begin{equation*}
\int_{\mathbb{R}}\left\langle \mathcal{E}_{t}\vec{w},\mu _{\mathcal{R}%
}^{\bot }(\mathcal{\hat{E}}^{(t)})\vec{w}\right\rangle \mathrm{d}t=0\ ,
\end{equation*}%
this expression is the classical work performed by the electric field on the
fermion system in the presence of currents $\mathbb{J}_{\mathcal{R}}^{(%
\mathcal{E})}$:
\begin{equation}
\int_{\mathbb{R}}\left\langle \mathcal{E}_{t}\vec{w},\mathbb{J}_{\mathcal{R}%
}^{(\mathcal{E})}\left( t\right) \right\rangle \mathrm{d}t+\mathcal{O}\left(
\left\Vert \mathcal{E}\right\Vert _{\infty }^{3}\right) \ .
\label{joule law}
\end{equation}%
As $\mu _{\mathcal{R}}\left( \mathbb{R}\backslash \{0\}\right) >0$, this
implies that electric fields generally produce heat in such systems and heat
production is directly related to the electric conductivity.

Note that the elements of the dual $\mathcal{S}_{0}^{\ast }$ of the space $%
\mathcal{S}_{0}$ of Schwartz functions $\mathbb{R\rightarrow R}$ satisfying
the AC--condition (\ref{AC condition}) are restrictions to $\mathcal{S}_{0}$
of tempered distributions. $\mathcal{S}_{0}^{\ast }$ is interpreted here as
a space of AC--currents and $(\mathcal{S}_{0},\mathcal{S}_{0}^{\ast })$ is a
dual pair. To obtain Joule's law in its original formulation, which relates
the heat production with currents rather than with electric fields, we
consider the Legendre--Fenchel transform $\mathbf{Q}_{\mathcal{R}}^{\ast }$
of the positive quadratic form
\begin{equation*}
\mathbf{Q}_{\mathcal{R}}\left( \mathcal{E}\right) :=\int_{\mathbb{R}%
}\left\langle \mathcal{E}_{t}\vec{w},\mu _{\mathcal{R}}^{\Vert }(\mathcal{%
\hat{E}}^{(t)})\right\rangle \mathrm{d}t\ .
\end{equation*}%
Let $\partial \mathbf{Q}_{\mathcal{R}}\left( \mathcal{E}\right) \subset
\mathcal{S}_{0}^{\ast }$ be the subdifferential of $\mathbf{Q}_{\mathcal{R}}$
at the point $\mathcal{E}\in \mathcal{S}_{0}$.The multifunction
\begin{equation*}
\mathcal{E}\mapsto \mathbf{\sigma }_{\mathcal{R}}\left( \mathcal{E}\right) =%
\frac{1}{2}\partial \mathbf{Q}_{\mathcal{R}}\left( \mathcal{E}\right)
\mathbf{\ }
\end{equation*}%
from $\mathcal{S}_{0}$ to $\mathcal{S}_{0}^{\ast }$ (i.e., the set--valued
map from $\mathcal{S}_{0}$ to $2^{\mathcal{S}_{0}^{\ast }}$) is
single--valued with domain $\mathrm{Dom}(\mathbf{\sigma }_{\mathcal{R}})=%
\mathcal{S}_{0}$. It is interpreted as the \emph{AC--conductivity} of the
region $\mathcal{R}$. Similarly, the multifunction
\begin{equation*}
\mathcal{J}\mapsto \mathbf{\rho }_{\mathcal{R}}\left( \mathcal{J}\right) =%
\frac{1}{2}\partial \mathbf{Q}_{\mathcal{R}}^{\ast }\left( \mathcal{J}\right)
\end{equation*}
from $\mathcal{S}_{0}^{\ast }$ to $\mathcal{S}_{0}$ (i.e., the set--valued
map from $\mathcal{S}_{0}^{\ast }$ to $2^{\mathcal{S}_{0}}$) is the \emph{%
AC--resistivity} of the region $\mathcal{R}$. Indeed, for all $\mathcal{J}%
\in \mathrm{Dom}(\mathbf{\rho }_{\mathcal{R}})\neq \emptyset $ and $\mathcal{%
E}\in \mathrm{Dom}(\mathbf{\sigma }_{\mathcal{R}})=\mathcal{S}_{0}$,
\begin{equation*}
\mathbf{\sigma }_{\mathcal{R}}\left( \mathbf{\rho }_{\mathcal{R}}\left(
\mathcal{J}\right) \right) =\{\mathcal{J}\}\qquad \text{and}\qquad \mathbf{%
\rho }_{\mathcal{R}}\left( \mathbf{\sigma }_{\mathcal{R}}\left( \mathcal{E}%
\right) \right) \supset \{\mathcal{E}\}\ .
\end{equation*}%
Moreover, the multifunction $\mathbf{\rho }_{\mathcal{R}}$ is \emph{linear},
in the sense described in Section \ref{Resistivity and Joule's Law sect},
and, for any $\mathcal{J}\in \mathrm{Dom}(\mathbf{\rho }_{\mathcal{R}})$,
\begin{equation}
\{\mathbf{Q}_{\mathcal{R}}^{\ast }\left( \mathcal{J}\right) \}=\left\langle
\mathcal{J},\mathbf{\rho }_{\mathcal{R}}\left( \mathcal{J}\right)
\right\rangle =\mathbf{Q}_{\mathcal{R}}\left( \mathbf{\rho }_{\mathcal{R}%
}\left( \mathcal{J}\right) \right) \ .  \label{joule classical}
\end{equation}%
Thus, $\left\langle \mathcal{J},\mathbf{\rho }_{\mathcal{R}}\left( \mathcal{J%
}\right) \right\rangle $ is the heat production (per unit of volume) in
presence of the current $\mathcal{J}\in \mathrm{Dom}(\mathbf{\rho }_{%
\mathcal{R}})$. In other words, (\ref{joule classical}) is an expression of
Joule's law in its original formulation, that is, the heat produced by
currents is proportional to the resistivity and the square of currents.

Remark that we use the Weyl gauge for which $\mathcal{E}$ is minus the time
derivative of the potential $\mathcal{A}$. Thus, the quantity $\int_{\mathbb{%
R}}\mathcal{E}_{t}\mathrm{d}t$ is the total shift of the electromagnetic
potential $\mathcal{A}$ between the times where the field $\mathcal{E}$ is
turned on and off. For this reason, we impose the AC--condition (\ref{AC
condition}) to identify the total electromagnetic work with the total \emph{%
internal} energy change of the system, which turns out to be the heat
production, by \cite[Theorem 3.2]{OhmI}. This condition is however not used
in our proofs and a general expression of the heat production as a function
of the applied electric field at any time is obtained.

Indeed, based on Araki's notion of relative entropy, \cite{OhmI} proves for
the fermion system under consideration that the first law of thermodynamics
holds at any time: We identify the heat production with an \emph{internal}
energy increment and define an electromagnetic \emph{potential} energy as
being the difference between the total and the internal energy increments.
Both energies are studied in detail here to get the heat production at
microscopic scales for all times.

Besides the internal energy increment we introduce the \emph{paramagnetic}
and \emph{diamagnetic} energy increments. The first one is the part of
electromagnetic work implying a change of the internal state of the system,
whereas the diamagnetic energy is the raw\emph{\ }electromagnetic energy
given to the system at thermal equilibrium. The paramagnetic energy
increment is associated to the presence of paramagnetic currents, whereas
the second one is caused by thermal and diamagnetic currents. We show that
these currents have different physical origins:

\begin{itemize}
\item \emph{Thermal} currents are currents coming from the space
inhomogeneity of the system. They exist, in general, even at thermal
equilibrium.

\item \emph{Diamagnetic} currents correspond to the raw ballistic flow of
charged particles due to the electric field, starting at thermal equilibrium.

\item Diamagnetic currents produced by the electric field create a kind of
\textquotedblleft propagating wave front\textquotedblright\ that
destabilizes the whole system by changing its internal state. In presence of
inhomogeneities the system opposes itself to the propagation of that front
by progressively creating so--called \emph{paramagnetic} currents. Such
induced currents act as a sort of friction (cf. current viscosity) to the
diamagnetic current and produce heat as well as a modification of the
electromagnetic potential energy.
\end{itemize}

We thus analyze the linear response in terms of diamagnetic and paramagnetic
currents, which form altogether the total current of the system and yield
the conductivity distribution. For more details on the features of such
currents, see Sections \ref{Sec Ohm law linear response} and \ref{Section
Joule effect}.

For the sake of technical simplicity and without loss of generality, note
that we only consider in the sequel an increasing sequence $\{\Lambda
_{l}\}_{l=1}^{\infty }$ of boxes instead of general convex regions $\mathcal{%
R}$ where the electric field is non--vanishing. We obtain \emph{uniform}
bounds permitting to control the behavior of $\mu _{\Lambda _{l}}$ at large
size $l\gg 1$ of the boxes $\{\Lambda _{l}\}_{l=1}^{\infty }$. The
uniformity of our results w.r.t. $l$ and the choice of the static potential
is a consequence of tree--decay bounds of the $n$--point, $n\in 2\mathbb{N}$%
, correlations of the many--fermion system \cite[Section 4]{OhmI}. Such
uniform bounds are crucial in our next paper \cite{OhmIII} on Ohm's law to
construct the macroscopic conductivity distribution in the case of free
fermions subjected to random static potentials (i.e., in the presence of
disorder).

The validity of Ohm's law at atomic scales mentioned in \cite%
{Ohm-exp,Ohm-exp2} suggests a fast convergence of $\mu _{\Lambda _{l}}$, as $%
l\rightarrow \infty $. Hence, we expect that the family $\{\mu _{\Lambda
_{l}}\}_{l=1}^{\infty }$ of measures on $\mathbb{R}$ obeys a large deviation
principle, for some relevant class of interactions between lattice fermions.
This question is, however, not addressed here.

To conclude, our main assertions are Theorems \ref{lemma sigma pos type
copy(4)} (existence of the conductivity measure), \ref{thm Local Ohm's law}
(cf. Ohm's law) and \ref{Local Ohm's law thm copy(2)}, \ref{Local Ohm's law
thm copy(1)} (cf. Joule's law). This paper is organized as follows:

\begin{itemize}
\item In Section \ref{Section main results} we briefly describe the
non--autonomous $C^{\ast }$--dynamical system for (free) fermions associated
to a discrete Schr\"{o}%
%TCIMACRO{\TeXButton{\-}{\-}}%
%BeginExpansion
\-%
%EndExpansion
dinger operator with bounded static potential in presence of an electric
field that is time-- and space--dependent. For more details, see also \cite[%
Section 2]{OhmI}.

\item Section \ref{Sect Local Ohm law} introduces Ohm's law at microscopic
scales via paramagnetic and diamagnetic currents. Mathematical properties of
the corresponding conductivities are explained in detail and a notion of
current viscosity is discussed.

\item Section \ref{Sect Joule} is devoted to the derivation of Joule's law
at microscopic scales. In particular, we introduce there four kinds of
energy increments: the internal energy increment or heat production, the
electromagnetic potential energy, the paramagnetic energy increment and the
diamagnetic energy. The AC--resistivity is also described.

\item All technical proofs are postponed to Section \ref{sect technical
proofs}. Additional properties on the conductivity measure are also proven,
see Section \ref{Section Local AC--Conductivity0 copy(2)}.

\item Finally, Section \ref{Section Duhamel Two--Point Functions} is an
appendix on the Duhamel two--point function. It is indeed an important
mathematical tool used here which frequently appears in the context of
linear response theory.
\end{itemize}

\begin{notation}[Generic constants]
\label{remark constant}\mbox{
}\newline
To simplify notation, we denote by $D$ any generic positive and finite
constant. These constants do not need to be the same from one statement to
another.
\end{notation}

\section{Setup of the Problem\label{Section main results}}

The aim of this section is to describe the non--autonomous $C^{\ast }$%
--dynamical system under consideration. Since almost everything is already
described in detail in \cite[Section 2]{OhmI}, we only focus on the specific
concepts or definitions that are important in the sequel.

\subsection{Free Fermion Systems on Lattices}

\subsubsection{Algebraic Formulation of Fermion Systems on Lattices}

The $d$--dimensional lattice $\mathfrak{L}:=\mathbb{Z}^{d}$ ($d\in \mathbb{N}
$) represents the (cubic) crystal and we define $\mathcal{P}_{f}(\mathfrak{L}%
)\subset 2^{\mathfrak{L}}$ to be the set of all \emph{finite} subsets of $%
\mathfrak{L}$. We denote by $\mathcal{U}$ the CAR $C^{\ast }$--algebra of
the infinite system and define annihilation and creation operators of
(spinless) fermions with wave functions $\psi \in \ell ^{2}(\mathfrak{L})$
by
\begin{equation*}
a(\psi ):=\sum\limits_{x\in \mathfrak{L}}\overline{\psi (x)}a_{x}\in
\mathcal{U}\ ,\quad a^{\ast }(\psi ):=\sum\limits_{x\in \mathfrak{L}}\psi
(x)a_{x}^{\ast }\in \mathcal{U}\ .
\end{equation*}%
Here, $a_{x},a_{x}^{\ast }$, $x\in \mathfrak{L}$, and the identity $\mathbf{1%
}$ are generators of $\mathcal{U}$ and satisfy the canonical
anti--commutation relations: For any $x,y\in \mathfrak{L}$,
\begin{equation}
a_{x}a_{y}+a_{y}a_{x}=0\ ,\qquad a_{x}a_{y}^{\ast }+a_{y}^{\ast
}a_{x}=\delta _{x,y}\mathbf{1}\ .  \label{CAR}
\end{equation}

\subsubsection{Static External Potentials}

Let $\Omega :=[-1,1]^{\mathfrak{L}}$. For any $\omega \in \Omega $, $%
V_{\omega }\in \mathcal{B}(\ell ^{2}(\mathfrak{L}))$ is defined to be the
self--adjoint multiplication operator with the function $\omega :\mathfrak{L}%
\rightarrow \lbrack -1,1]$. The static external potential $V_{\omega }$ is
of order $\mathcal{O}(1)$ and we rescale below its strength by an additional
parameter $\lambda \in \mathbb{R}_{0}^{+}$ (i.e., $\lambda \geq 0$).

\subsubsection{Dynamics on the One--Particle Hilbert Space}

Let $\Delta _{\mathrm{d}}\in \mathcal{B}(\ell ^{2}(\mathfrak{L}))$ be (up to
a minus sign) the usual $d$--dimensional discrete Laplacian defined by%
\begin{equation}
\lbrack \Delta _{\mathrm{d}}(\psi )](x):=2d\psi (x)-\sum\limits_{z\in
\mathfrak{L},\text{ }|z|=1}\psi (x+z)\ ,\text{\qquad }x\in \mathfrak{L},\
\psi \in \ell ^{2}(\mathfrak{L})\ .  \label{discrete laplacian}
\end{equation}%
Then, for $\omega \in \Omega $ and $\lambda \in \mathbb{R}_{0}^{+}$, the
dynamics in the one--particle Hilbert space $\ell ^{2}(\mathfrak{L})$ is
implemented by the unitary group $\{\mathrm{U}_{t}^{(\omega ,\lambda
)}\}_{t\in \mathbb{R}}$ generated by the (anti--self--adjoint) operator $%
-i(\Delta _{\mathrm{d}}+\lambda V_{\omega })$:%
\begin{equation}
\mathrm{U}_{t}^{(\omega ,\lambda )}:=\exp (-it(\Delta _{\mathrm{d}}+\lambda
V_{\omega }))\in \mathcal{B}(\ell ^{2}(\mathfrak{L}))\ ,\text{\qquad }t\in
\mathbb{R}\ .  \label{rescaled}
\end{equation}

\subsubsection{Dynamics on the CAR $C^{\ast }$--Algebra\label{Dynamics free
sect}}

For all $\omega \in \Omega $ and $\lambda \in \mathbb{R}_{0}^{+}$, the
condition%
\begin{equation}
\tau _{t}^{(\omega ,\lambda )}(a(\psi ))=a((\mathrm{U}_{t}^{(\omega ,\lambda
)})^{\ast }(\psi ))\ ,\text{\qquad }t\in \mathbb{R}\ ,\ \psi \in \ell ^{2}(%
\mathfrak{L})\ ,  \label{rescaledbis}
\end{equation}%
uniquely defines a family $\tau ^{(\omega ,\lambda )}:=\{\tau _{t}^{(\omega
,\lambda )}\}_{t\in {\mathbb{R}}}$ of (Bogoliubov) $\ast $--automorphisms of
$\mathcal{U}$, see \cite[Theorem 5.2.5]{BratteliRobinson}. The
one--parameter group $\tau ^{(\omega ,\lambda )}$ is strongly continuous and
we denote its generator by $\delta ^{(\omega ,\lambda )}$. Clearly,%
\begin{equation}
\tau _{t}^{(\omega ,\lambda )}(B_{1}B_{2})=\tau _{t}^{(\omega ,\lambda
)}(B_{1})\tau _{t}^{(\omega ,\lambda )}(B_{2})\ ,\qquad B_{1},B_{2}\in
\mathcal{U}\ ,\ t\in \mathbb{R}\ .  \label{inequality idiote}
\end{equation}%
In the following, we will need the \emph{time--reversal} operation $\Theta $%
. It is the unique map $\Theta :\mathcal{U}\rightarrow \mathcal{U}$
satisfying the following properties:

\begin{itemize}
\item $\Theta $ is antilinear and continuous.

\item $\Theta \left( \mathbf{1}\right) =\mathbf{1}$ and $\Theta \left(
a_{x}\right) =a_{x}$ for all $x\in \mathfrak{L}$.

\item $\Theta \left( B_{1}B_{2}\right) =\Theta \left( B_{1}\right) \Theta
\left( B_{2}\right) $ for all $B_{1},B_{2}\in \mathcal{U}$.

\item $\Theta \left( B^{\ast }\right) =\Theta \left( B\right) ^{\ast }$ for
all $B\in \mathcal{U}$.
\end{itemize}

\noindent In particular, $\Theta $ is involutive, i.e., $\Theta \circ \Theta
=\mathrm{Id}_{\mathcal{U}}$. This operation can be explicitly defined by
using the Fock representation of $\mathcal{U}$. It is called \emph{%
time--reversal} of the dynamics $\tau _{t}^{(\omega ,\lambda )}$ because of
the following identity%
\begin{equation*}
\Theta \circ \tau _{t}^{(\omega ,\lambda )}=\tau _{-t}^{(\omega ,\lambda
)}\circ \Theta \ ,
\end{equation*}%
which is valid for all $\omega \in \Omega $, $\lambda \in \mathbb{R}_{0}^{+}$
and $t\in \mathbb{R}$, see Lemma \ref{lemma time reversal}. This feature is
important to obtain a symmetric conductivity measure.

\subsubsection{Thermal Equilibrium State}

For any realization $\omega \in \Omega $ and strength $\lambda \in \mathbb{R}%
_{0}^{+}$ of the static external potential, the thermal equilibrium state of
the system at inverse temperature $\beta \in \mathbb{R}^{+}$ (i.e., $\beta
>0 $) is by definition the unique $(\tau ^{(\omega ,\lambda )},\beta )$--KMS
state $\varrho ^{(\beta ,\omega ,\lambda )}$, see \cite[Example 5.3.2.]%
{BratteliRobinson} or \cite[Theorem 5.9]{AttalJoyePillet2006a}. It is
well--known that such a state is stationary with respect to (w.r.t.) the
dynamics, that is,
\begin{equation}
\varrho ^{(\beta ,\omega ,\lambda )}\circ \tau _{t}^{(\omega ,\lambda
)}=\varrho ^{(\beta ,\omega ,\lambda )}\ ,\qquad \beta \in \mathbb{R}^{+},\
\omega \in \Omega ,\ \lambda \in \mathbb{R}_{0}^{+},\ t\in \mathbb{R}\ .
\label{stationary}
\end{equation}%
The state $\varrho ^{(\beta ,\omega ,\lambda )}$ is \emph{gauge--invariant
and quasi--free}. Such states are uniquely characterized by bounded positive
operators $\mathbf{d}\in \mathcal{B}(\ell ^{2}(\mathfrak{L}))$ obeying $%
0\leq \mathbf{d}\leq \mathbf{1}$. These operators are named \emph{symbols}
of the corresponding states. The symbol of $\varrho ^{(\beta ,\omega
,\lambda )}$ is given by%
\begin{equation}
\mathbf{d}_{\mathrm{fermi}}^{(\beta ,\omega ,\lambda )}:=\frac{1}{1+\mathrm{e%
}^{\beta \left( \Delta _{\mathrm{d}}+\lambda V_{\omega }\right) }}\in
\mathcal{B}(\ell ^{2}(\mathfrak{L}))\ .  \label{Fermi statistic}
\end{equation}%
Let us remark here that $\varrho ^{(\beta ,\omega ,\lambda )}$ is
time--reversal invariant, i.e., for all parameters $\beta \in \mathbb{R}^{+}$%
,\ $\omega \in \Omega $,\ $\lambda \in \mathbb{R}_{0}^{+}$,
\begin{equation*}
\varrho ^{(\beta ,\omega ,\lambda )}\circ \Theta \left( B\right) =\overline{%
\varrho ^{(\beta ,\omega ,\lambda )}\left( B\right) }\ ,\text{\qquad }B\in
\mathcal{U}\ .
\end{equation*}%
See Lemma \ref{lemma time reversal}.

\subsection{Fermion Systems in Presence of Electromagnetic Fields}

\subsubsection{Electric Fields}

Using the Weyl gauge (also named temporal gauge), the electric field is
defined from a compactly supported potential%
\begin{equation*}
\mathbf{A}\in \mathbf{C}_{0}^{\infty }=\underset{l\in \mathbb{R}^{+}}{%
\mathop{\displaystyle \bigcup }}C_{0}^{\infty }(\mathbb{R}\times \left[ -l,l%
\right] ^{d};({\mathbb{R}}^{d})^{\ast })
\end{equation*}%
by
\begin{equation}
E_{\mathbf{A}}(t,x):=-\partial _{t}\mathbf{A}(t,x)\ ,\quad t\in \mathbb{R},\
x\in \mathbb{R}^{d}\ .  \label{V bar 0}
\end{equation}%
Here, $({\mathbb{R}}^{d})^{\ast }$ is the set of one--forms\footnote{%
In a strict sense, one should take the dual space of the tangent spaces $T({%
\mathbb{R}}^{d})_{x}$, $x\in {\mathbb{R}}^{d}$.} on ${\mathbb{R}}^{d}$ that
take values in $\mathbb{R}$ and $\mathbf{A}(t,x)\equiv 0$ whenever $x\notin
\lbrack -l,l]^{d}$ and $\mathbf{A}\in C_{0}^{\infty }(\mathbb{R}\times \left[
-l,l\right] ^{d};({\mathbb{R}}^{d})^{\ast })$. Since $\mathbf{A}\in \mathbf{C%
}_{0}^{\infty }$, $\mathbf{A}(t,x)=0$ for all $t\leq t_{0}$, where $t_{0}\in
\mathbb{R}$ is some initial time. We also define the integrated electric
field between $x^{(2)}\in \mathfrak{L}$ and $x^{(1)}\in \mathfrak{L}$ at
time $t\in \mathbb{R}$ by
\begin{equation}
\mathbf{E}_{t}^{\mathbf{A}}\left( \mathbf{x}\right) :=\int\nolimits_{0}^{1}%
\left[ E_{\mathbf{A}}(t,\alpha x^{(2)}+(1-\alpha )x^{(1)})\right]
(x^{(2)}-x^{(1)})\mathrm{d}\alpha \ ,  \label{V bar 0bis}
\end{equation}%
where $\mathbf{x}:=(x^{(1)},x^{(2)})\in \mathfrak{L}^{2}$.

\subsubsection{Discrete Magnetic Laplacian}

We consider without loss of generality \emph{negatively} charged fermions.
Thus, using the (minimal) coupling of $\mathbf{A}\in \mathbf{C}_{0}^{\infty
} $ to the discrete Laplacian $-\Delta _{\mathrm{d}}$, the discrete \emph{%
time--dependent} magnetic Laplacian is (up to a minus sign) the
self--adjoint operator
\begin{equation*}
\Delta _{\mathrm{d}}^{(\mathbf{A})}\equiv \Delta _{\mathrm{d}}^{(\mathbf{A}%
(t,\cdot ))}\in \mathcal{B}(\ell ^{2}(\mathfrak{L}))\ ,\qquad t\in \mathbb{R}%
\ ,
\end{equation*}%
defined\footnote{%
Observe that the sign of the coupling between the electromagnetic potential $%
\mathbf{A}\in \mathbf{C}_{0}^{\infty }$ and the laplacian is wrong in \cite[%
Eq. (2.8)]{OhmI}.} by%
\begin{equation}
\langle \mathfrak{e}_{x},\Delta _{\mathrm{d}}^{(\mathbf{A})}\mathfrak{e}%
_{y}\rangle =\exp \left( i\int\nolimits_{0}^{1}\left[ \mathbf{A}(t,\alpha
y+(1-\alpha )x)\right] (y-x)\mathrm{d}\alpha \right) \langle \mathfrak{e}%
_{x},\Delta _{\mathrm{d}}\mathfrak{e}_{y}\rangle  \label{eq discrete lapla A}
\end{equation}%
for all $t\in \mathbb{R}$ and $x,y\in \mathfrak{L}$. Here, $\langle \cdot
,\cdot \rangle $ is the scalar product in $\ell ^{2}(\mathfrak{L})$ and $%
\left\{ \mathfrak{e}_{x}\right\} _{x\in \mathfrak{L}}$ is the canonical
orthonormal basis $\mathfrak{e}_{x}(y)\equiv \delta _{x,y}$ of $\ell ^{2}(%
\mathfrak{L})$. In (\ref{eq discrete lapla A}), $\alpha y+(1-\alpha )x$ and $%
y-x$ are seen as vectors in ${\mathbb{R}}^{d}$.

\subsubsection{Perturbed Dynamics on the One--Particle Hilbert Space}

The dynamics of the system under the influence of an electromagnetic
potential is defined via the two--parameter group $\{\mathrm{U}%
_{t,s}^{(\omega ,\lambda ,\mathbf{A})}\}_{t\geq s}$ of unitary operators on $%
\ell ^{2}(\mathfrak{L})$ generated by the (time--dependent
anti--self--adjoint) operator $-i(\Delta _{\mathrm{d}}^{(\mathbf{A}%
)}+\lambda V_{\omega })$ for any $\omega \in \Omega $, $\lambda \in \mathbb{R%
}_{0}^{+}$ and $\mathbf{A}\in \mathbf{C}_{0}^{\infty }$:
\begin{equation}
\forall s,t\in {\mathbb{R}},\ t\geq s:\quad \partial _{t}\mathrm{U}%
_{t,s}^{(\omega ,\lambda ,\mathbf{A})}=-i(\Delta _{\mathrm{d}}^{(\mathbf{A}%
(t,\cdot ))}+\lambda V_{\omega })\mathrm{U}_{t,s}^{(\omega ,\lambda ,\mathbf{%
A})}\ ,\quad \mathrm{U}_{s,s}^{(\omega ,\lambda )}:=\mathbf{1}\ .
\label{time evolution one-particle}
\end{equation}%
The dynamics is well--defined because the map
\begin{equation*}
t\mapsto (\Delta _{\mathrm{d}}^{(\mathbf{A}(t,\cdot ))}+\lambda V_{\omega
})\in \mathcal{B}(\ell ^{2}(\mathfrak{L}))
\end{equation*}%
from $\mathbb{R}$ to the set $\mathcal{B}(\ell ^{2}(\mathfrak{L}))$ of
bounded operators acting on $\ell ^{2}(\mathfrak{L})$ is continuously
differentiable for every $\mathbf{A}\in \mathbf{C}_{0}^{\infty }$.

Note that, as explained in \cite[Section 2.3]{OhmI}, the interaction between
magnetic fields and electron spins is here neglected because such a term
will become negligible for electromagnetic potentials slowly varying in
space, see Section \ref{section Macroscopic Electromagnetic Fields}. This
justifies the assumption of fermions with zero--spin.

\subsubsection{Perturbed Dynamics on the CAR $C^{\ast }$--Algebra}

For all $\omega \in \Omega $, $\lambda \in \mathbb{R}_{0}^{+}$ and $\mathbf{A%
}\in \mathbf{C}_{0}^{\infty }$, the condition%
\begin{equation}
\tau _{t,s}^{(\omega ,\lambda ,\mathbf{A})}(a(\psi ))=a((\mathrm{U}%
_{t,s}^{(\omega ,\lambda ,\mathbf{A})})^{\ast }(\psi ))\ ,\text{\qquad }%
t\geq s,\ \psi \in \ell ^{2}(\mathfrak{L})\ ,  \label{Cauchy problem 0}
\end{equation}%
uniquely defines a family of Bogoliubov automorphisms of the $C^{\ast }$%
--algebra $\mathcal{U}$, see \cite[Theorem 5.2.5]{BratteliRobinson}. The
family $\{\tau _{t,s}^{(\omega ,\lambda ,\mathbf{A})}\}_{t\geq s}$ is itself
the solution of a non--autonomous evolution equation, see \cite[Sections
5.2-5.3]{OhmI}.

\subsubsection{Time--Dependent State}

Since $\varrho ^{(\beta ,\omega ,\lambda )}$ is stationary (cf. (\ref%
{stationary})) and $\mathbf{A}(t,x)=0$ for all $t\leq t_{0}$, the time
evolution of the state of the system equals%
\begin{equation}
\rho _{t}^{(\beta ,\omega ,\lambda ,\mathbf{A})}:=\left\{
\begin{array}{lll}
\varrho ^{(\beta ,\omega ,\lambda )} & , & \qquad t\leq t_{0}\ , \\
\varrho ^{(\beta ,\omega ,\lambda )}\circ \tau _{t,t_{0}}^{(\omega ,\lambda ,%
\mathbf{A})} & , & \qquad t\geq t_{0}\ .%
\end{array}%
\right.  \label{time dependent state}
\end{equation}%
This state is gauge--invariant and quasi--free for all times, by
construction.

\subsection{Space--Scale of Fields, Linear Response Theory and Scanning Gate Microscopy}

\subsubsection{From Microscopic to Macroscopic Electromagnetic Fields\label%
{section Macroscopic Electromagnetic Fields}}

For space scales large compared to $10^{-14}$ m, electron and nuclei are
usually treated as point systems and electromagnetic phenomena are governed
by \emph{microscopic} Maxwell equations. However, the electromagnetic fields
produced by these point charges fluctuate very much in space and time and
macroscopic devices generally measure averages over intervals in space and
time much larger than the scale of these fluctuations. This implies
relatively smooth and slowly varying macroscopic quantities. As explained in
\cite[Section 6.6]{Jackson}, \textit{\textquotedblleft only a spatial
averaging is necessary.}\textquotedblright\ The \emph{macroscopic}
electromagnetic fields are thus coarse--grainings of microscopic ones and
satisfy the so--called macroscopic Maxwell equations. In particular, their
spacial variations become negligible on the atomic scale.

Similarly, we consider that the infinite bulk containing conducting fermions
only experiences mesoscopic electromagnetic fields, which are produced by
mesoscopic devices. In other words, the heat production or the conductivity
is measured in a local region which is very small w.r.t. the size of the
bulk, but very large w.r.t. the lattice spacing of the crystal. We implement
this hierarchy of space scales by rescaling vector potentials. That means,
for any $l\in \mathbb{R}^{+}$ and $\mathbf{A}\in \mathbf{C}_{0}^{\infty }$,
we consider the space--rescaled vector potential $\mathbf{A}_{l}$ defined by
\begin{equation}
\mathbf{A}_{l}(t,x):=\mathbf{A}(t,l^{-1}x)\ ,\quad t\in \mathbb{R},\ x\in
\mathbb{R}^{d}\ .  \label{rescaled vector potential}
\end{equation}%
Then, to ensure that an infinite number of lattice sites is involved, we
eventually perform the limit $l\rightarrow \infty $. See \cite{OhmIII} for
more details.

Indeed, the scaling factor $l^{-1}$ used in (\ref{rescaled vector potential}%
) means, at fixed $l$, that the space scale of the electric field (\ref{V
bar 0}) is infinitesimal w.r.t. the macroscopic bulk (which is the whole
space), whereas the lattice spacing gets infinitesimal w.r.t. the space
scale of the electric field when $l\rightarrow \infty $.

\subsubsection{Linear Response Theory}

Linear response theory refers here to linearized non--equilibrium
statistical mechanics and has been initiated by Kubo \cite{kubo} and Mori
\cite{mori}. Ohm's law is one of the first and certainly one of the most
important examples thereof. It is indeed a linear\emph{\ }response to
electric fields. Therefore, we also rescale the strength of the
electromagnetic potential $\mathbf{A}_{l}$ by a real parameter $\eta \in
\mathbb{R}$ and eventually take the limit $\eta \rightarrow 0$.

When $|\eta |\ll 1$ and $l\gg 1$, it turns out that, uniformly w.r.t. $l$,
the mean currents $\mathbb{J}_{\mathrm{p}}^{(\omega ,\eta \mathbf{\bar{A}}%
_{l})}$ and $\mathbb{J}_{\mathrm{d}}^{(\omega ,\eta \mathbf{\bar{A}}_{l})}$,
defined below by (\ref{finite volume current density})--(\ref{finite volume
current density2}), are of order $\mathcal{O}\left( \eta \right) $.
Similarly, the energy increments $\mathbf{S}^{(\omega ,\eta \mathbf{A}_{l})}$%
, $\mathbf{P}^{(\omega ,\eta \mathbf{A}_{l})}$, $\mathfrak{I}_{\mathrm{p}%
}^{(\omega ,\eta \mathbf{A}_{l})}$ and $\mathfrak{I}_{\mathrm{d}}^{(\omega
,\eta \mathbf{A}_{l})}$, respectively defined by (\ref{entropic energy
increment}), (\ref{electro free energy}), (\ref{lim_en_incr}) and (\ref%
{lim_en_incr dia}), are all of order $\mathcal{O}\left( \eta
^{2}l^{d}\right) $. Such results are derived in the next sections by using
tree--decay bounds of the $n$--point, $n\in 2\mathbb{N}$, correlations of
the many--fermion system \cite[Section 4]{OhmI}.

\subsubsection{Experimental Setting of Scanning Gate Microscopy}

Our setting is reminiscent of the so--called scanning gate microscopy used
to perform imaging of electron transport in two--dimensional semiconductor
quantum structures. See, e.g., \cite{SGM1}. In this experimental situation,
the two--dimensional electron system on a lattice experiences a
time--periodic space--homogeneous electromagnetic potential perturbed by a
mesoscopic or microscopic \emph{time--independent} electric potential.
Physically speaking, this situation is, mutatis mutandis, analogous to the
one considered here. Therefore, we expect that our setting can also be
implemented in experiments by similar technics combined with calorimetry to
measure the heat production.

\section{Microscopic Ohm's Law\label{Sect Local Ohm law}}

In his original work \cite{thermo-ohm} G.S. Ohm states that the current in
the steady regime is proportional to the voltage applied to the conducting
material. The proportionality coefficient is the conductivity of the
physical system. Ohm's laws is among the most resilient laws of (classical)
electricity theory and is usually justified from a microscopic point of view
by the Drude model or some of its improvements that take into account
quantum corrections. [Cf. the Landau theory of Fermi liquids.] As in the
Drude model we do not consider here interactions between charge carriers,
but our approach will be also applied to interacting fermions in subsequent
papers.

In this section, we study, among other things, (microscopic) Ohm's law in
Fourier space for the system of free fermions described in\ Section \ref%
{Section main results}. Without loss of generality, we only consider
space--homogeneous (though time--dependent) electric fields in the box%
\begin{equation}
\Lambda _{l}:=\{(x_{1},\ldots ,x_{d})\in \mathfrak{L}\,:\,|x_{1}|,\ldots
,|x_{d}|\leq l\}\in \mathcal{P}_{f}(\mathfrak{L})  \label{eq:def lambda n}
\end{equation}%
with $l\in \mathbb{R}^{+}$. More precisely, let $\vec{w}:=(w_{1},\ldots
,w_{d})\in \mathbb{R}^{d}$ be any (normalized) vector, $\mathcal{A}\in
C_{0}^{\infty }\left( \mathbb{R};\mathbb{R}\right) $ and set $\mathcal{E}%
_{t}:=-\partial _{t}\mathcal{A}_{t}$ for all $t\in \mathbb{R}$. Then, $%
\mathbf{\bar{A}}\in \mathbf{C}_{0}^{\infty }$ is defined to be the
electromagnetic potential such that the value of the electric field equals $%
\mathcal{E}_{t}\vec{w}$ at time $t\in \mathbb{R}$ for all $x\in \left[ -1,1%
\right] ^{d}$ and $(0,0,\ldots ,0)$ for $t\in \mathbb{R}$ and $x\notin \left[
-1,1\right] ^{d}$. This choice yields rescaled electromagnetic potentials $%
\eta \mathbf{\bar{A}}_{l}$ as defined by (\ref{rescaled vector potential})
for $l\in \mathbb{R}^{+}$ and $\eta \in \mathbb{R}$.

Before stating Ohm's law for the system under consideration we first need
some definitions.

\subsection{Current Observables\label{Section Current Observables}}

For any pair $\mathbf{x}:=(x^{(1)},x^{(2)})\in \mathfrak{L}^{2}$, we define
the \emph{paramagnetic}\ and \emph{diamagnetic} current observables $I_{%
\mathbf{x}}=I_{\mathbf{x}}^{\ast }$ and $\mathrm{I}_{\mathbf{x}}^{\mathbf{A}%
}=(\mathrm{I}_{\mathbf{x}}^{\mathbf{A}})^{\ast }$ for $\mathbf{A}\in \mathbf{%
C}_{0}^{\infty }$ at time $t\in \mathbb{R}$ by%
\begin{equation}
I_{\mathbf{x}}:=-2\mathrm{Im}(a_{x^{(2)}}^{\ast
}a_{x^{(1)}})=i(a_{x^{(2)}}^{\ast }a_{x^{(1)}}-a_{x^{(1)}}^{\ast
}a_{x^{(2)}})  \label{current observable}
\end{equation}%
and%
\begin{equation}
\mathrm{I}_{\mathbf{x}}^{\mathbf{A}}:=-2\mathrm{Im}\left( \left( \mathrm{e}%
^{i\int\nolimits_{0}^{1}[\mathbf{A}(t,\alpha x^{(2)}+(1-\alpha
)x^{(1)})](x^{(2)}-x^{(1)})\mathrm{d}\alpha }-1\right) a_{x^{(2)}}^{\ast
}a_{x^{(1)}}\right) \ .  \label{current observable new}
\end{equation}%
These are seen as currents because, by (\ref{time evolution one-particle})--(%
\ref{Cauchy problem 0}), they satisfy the discrete continuity equation%
\begin{equation}
\partial _{t}n_{x}(t)=-\tau _{t,t_{0}}^{(\omega ,\lambda ,\mathbf{A})}\left(
\sum\limits_{z\in \mathfrak{L}}\mathbf{1}\left[ |z|=1\right] \left(
I_{(x,x+z)}+\mathrm{I}_{(x,x+z)}^{\mathbf{A}}\right) \right)
\label{current observable2}
\end{equation}%
for $x\in \mathfrak{L}$ and $t\geq t_{0}$, where%
\begin{equation}
n_{x}(t):=\tau _{t,t_{0}}^{(\omega ,\lambda ,\mathbf{A})}(a_{x}^{\ast }a_{x})
\label{current observable3}
\end{equation}%
is the density observable at lattice site $x\in \mathfrak{L}$ and time $%
t\geq t_{0}$. The notions of paramagnetic and diamagnetic current
observables come from the physics literature, see, e.g., \cite[Eq. (A2.14)]%
{dia-current}. The paramagnetic current observable $\mathbf{1}\left[ |z|=1%
\right] I_{(x,x+z)}$ is intrinsic to the system whereas the diamagnetic one $%
\mathrm{I}_{\mathbf{x}}^{\mathbf{A}}$ is only non--vanishing in presence of
electromagnetic potentials.

Observe that the minus sign in the right hand side of (\ref{current
observable2}) comes from the fact that the particles are negatively charged,
$I_{(x,y)}$ being the observable related to the flow of particles from the
lattice site $x$ to the lattice site $y$ or the current from $y$ to $x$
without external electromagnetic potential. [Positively charged particles
can of course be treated in the same way.] As one can see from (\ref{current
observable2}), current observables on bonds of nearest neighbors are
especially important. Thus, we define the subset
\begin{equation}
\mathfrak{K}:=\left\{ \mathbf{x}:=(x^{(1)},x^{(2)})\in \mathfrak{L}^{2}\ :\
|x^{(1)}-x^{(2)}|=1\right\}  \label{proche voisins0}
\end{equation}%
of bonds of nearest neighbors.

In fact, by using the canonical orthonormal basis $\{e_{k}\}_{k=1}^{d}$ of
the Euclidian space $\mathbb{R}^{d}$, we define the current sums in the box $%
\Lambda _{l}$ (\ref{eq:def lambda n}) for any $l\in \mathbb{R}^{+}$, $%
\mathbf{A}\in \mathbf{C}_{0}^{\infty }$, $t\in \mathbb{R}$ and $k\in
\{1,\ldots ,d\}$ by
\begin{equation}
\mathbb{I}_{k,l}:=\underset{x\in \Lambda _{l}}{\sum }I_{(x+e_{k},x)}-\varrho
^{(\beta ,\omega ,\lambda )}\left( I_{(x+e_{k},x)}\right) \mathbf{1}\mathcal{%
\qquad }\text{and}\mathcal{\qquad }\mathbf{I}_{k,l}^{\mathbf{A}}:=\underset{%
x\in \Lambda _{l}}{\sum }\mathrm{I}_{(x+e_{k},x)}^{\mathbf{A}}\ .
\label{current density=}
\end{equation}%
In particular, $\varrho ^{(\beta ,\omega ,\lambda )}\left( \mathbb{I}%
_{k,l}\right) =0$, while $\mathbf{I}_{k,l}^{\mathbf{A}}=0$ when $\mathbf{A}%
(t,\cdot )=0$.

\subsection{Adjacency Observables}

Let $P_{\mathbf{x}}$, $\mathbf{x}=(x^{(1)},x^{(2)})$, be the
second--quantization of the \emph{adjacency matrix} of the oriented graph
containing exactly the pairs $(x^{(2)},x^{(1)})$ and $(x^{(1)},x^{(2)})$,
i.e.,
\begin{equation}
P_{\mathbf{x}}:=-a_{x^{(2)}}^{\ast }a_{x^{(1)}}-a_{x^{(1)}}^{\ast
}a_{x^{(2)}}\ ,\qquad \mathbf{x}:=(x^{(1)},x^{(2)})\in \mathfrak{L}^{2}\ .
\label{R x}
\end{equation}%
The observable $P_{\mathbf{x}}$ is related to the current observable $I_{%
\mathbf{x}}$ in the following way: For any $\mathbf{x}:=(x^{(1)},x^{(2)})\in
\mathfrak{L}^{2}$,
\begin{equation}
2a_{x^{(1)}}^{\ast }a_{x^{(2)}}=-P_{\mathbf{x}}+iI_{\mathbf{x}}\ ,\qquad %
\left[ P_{\mathbf{x}},I_{\mathbf{x}}\right] =2i\left( a_{x^{(2)}}^{\ast
}a_{x^{(2)}}-a_{x^{(1)}}^{\ast }a_{x^{(1)}}\right) \ .
\label{cool complementaire}
\end{equation}%
The importance of the adjacency observable $P_{\mathbf{x}}$ in the linear
response regime results from the fact that
\begin{equation}
\mathrm{I}_{\mathbf{x}}^{\eta \mathbf{A}}=\eta P_{\mathbf{x}%
}\int\nolimits_{0}^{1}[\mathbf{A}(t,\alpha x^{(2)}+(1-\alpha
)x^{(1)})](x^{(2)}-x^{(1)})\mathrm{d}\alpha +\mathcal{O}\left( \eta
^{2}\right) \ .  \label{dia-adjency}
\end{equation}%
Then, similar to the \emph{diamagnetic} current sum $\mathbf{I}_{k,l}^{%
\mathbf{A}}$ (\ref{current density=}), we define the observables
\begin{equation}
\mathbb{P}_{k,l}:=\underset{x\in \Lambda _{l}}{\sum }P_{(x+e_{k},x)}\in
\mathcal{U}\ ,\qquad l\in \mathbb{R}^{+}\ ,\ k\in \{1,\ldots ,d\}\ .
\label{macroscopic fermion field}
\end{equation}

\subsection{Microscopic Transport Coefficients\label{Sect Trans coeef ddef}}

Now, for any $\beta \in \mathbb{R}^{+}$, $\omega \in \Omega $ and $\lambda
\in \mathbb{R}_{0}^{+}$ we define two important functions associated with
the observables $I_{\mathbf{x}}$ and $P_{\mathbf{x}}$:

\begin{itemize}
\item[(p)] The \emph{paramagnetic} transport coefficient $\sigma _{\mathrm{p}%
}^{(\omega )}\equiv \sigma _{\mathrm{p}}^{(\beta ,\omega ,\lambda )}$ is
defined by
\begin{equation}
\sigma _{\mathrm{p}}^{(\omega )}\left( \mathbf{x},\mathbf{y},t\right)
:=\int\nolimits_{0}^{t}\varrho ^{(\beta ,\omega ,\lambda )}\left( i[I_{%
\mathbf{y}},\tau _{s}^{(\omega ,\lambda )}(I_{\mathbf{x}})]\right) \mathrm{d}%
s\ ,\quad \mathbf{x},\mathbf{y}\in \mathfrak{L}^{2}\ ,\ t\in \mathbb{R}\ .
\label{backwards -1bis}
\end{equation}

\item[(d)] The \emph{diamagnetic} transport coefficient $\sigma _{\mathrm{d}%
}^{(\omega )}\equiv \sigma _{\mathrm{d}}^{(\beta ,\omega ,\lambda )}$ is
defined by%
\begin{equation}
\sigma _{\mathrm{d}}^{(\omega )}\left( \mathbf{x}\right) :=\varrho ^{(\beta
,\omega ,\lambda )}\left( P_{\mathbf{x}}\right) \ ,\qquad \mathbf{x}\in
\mathfrak{L}^{2}\ .  \label{backwards -1bispara}
\end{equation}
\end{itemize}

At $\mathbf{x}\in \mathfrak{L}^{2}$, $\sigma _{\mathrm{d}}^{(\omega )}\left(
\mathbf{x}\right) $ is obviously the expectation value of the adjacency
observable $P_{\mathbf{x}}$ in the thermal state $\varrho ^{(\beta ,\omega
,\lambda )}$ of the fermion system. This coefficient is diamagnetic because
of (\ref{dia-adjency}). For any bond $\mathbf{x}\in \mathfrak{K}$, it can be
interpreted as being the kinetic energy in $\mathbf{x}$: The total
kinetic energy observable in the box $\Lambda _{l}$ equals%
\begin{equation*}
2d\underset{x\in \Lambda _{l}}{\sum }a_{x}^{\ast }a_{x}-\underset{\mathbf{x}%
=(x^{(1)},x^{(2)})\in \mathfrak{K}\cap \Lambda _{l}^{2}}{\sum }%
a_{^{x^{(2)}}}^{\ast }a_{x^{(1)}}=2d\underset{x\in \Lambda _{l}}{\sum }%
a_{x}^{\ast }a_{x}+\frac{1}{2}\underset{\mathbf{x}\in \mathfrak{K}\cap
\Lambda _{l}^{2}}{\sum }P_{\mathbf{x}}\ .
\end{equation*}%
The particle number observables $a_{x}^{\ast }a_{x}$, $x\in \Lambda _{l}$,
are rather related to the (kinetic) energy in the lattice sites.

The physical meaning of $\sigma _{\mathrm{p}}^{(\omega )}$ is less obvious.
We motivate in the following that it is a linear coupling between the
diamagnetic current in the bond $\mathbf{y}$ and the paramagnetic current in
the bond $\mathbf{x}$: Indeed, define by $\delta ^{(\omega ,\lambda )}$ the
generator of the group $\tau ^{(\omega ,\lambda )}$, see (\ref{rescaledbis}%
). Then, for any fixed $\beta \in \mathbb{R}^{+}$, $\omega \in \Omega $, $%
\lambda \in \mathbb{R}_{0}^{+}$, $\eta \in \mathbb{R}$ and $\mathbf{y}\in
\mathfrak{K}$, let the symmetric derivation
\begin{equation}
\tilde{\delta}^{(\eta ,\mathbf{y})}:=\delta ^{(\omega ,\lambda )}+i\eta %
\left[ I_{\mathbf{y}},\ \cdot \ \right]
\label{pertubated derivation current}
\end{equation}%
be the generator of the (perturbed) group $\{\tilde{\tau}_{t}^{(\eta ,%
\mathbf{y})}\}_{t\in {\mathbb{R}}}$ of automorphisms of the $C^{\ast }$%
--algebra $\mathcal{U}$. Note that this perturbation corresponds at leading
order in $\eta $ to an electromagnetic potential $\eta \mathbf{A}^{(\mathbf{y%
})}$ of order $\eta $ along the bond $\mathbf{y}$. See, e.g., Lemma \ref%
{bound incr 1 Lemma copy(9)}. This small electromagnetic potential yields a
diamagnetic current observable of the order $\eta P_{\mathbf{y}}$ on the
same bond $\mathbf{y}$, cf. (\ref{dia-adjency}). Since $I_{\mathbf{y}}\in
\mathcal{U}$ (cf. (\ref{current observable})), we may use a Dyson--Phillips
series to obtain for small $|\eta |\ll 1$ that%
\begin{equation*}
\tilde{\tau}_{t}^{(\eta ,\mathbf{y})}(B)=\tau _{t}^{(\omega ,\lambda
)}(B)+\eta \int\nolimits_{0}^{t}\tau _{t-s}^{(\omega ,\lambda )}\left( i[I_{%
\mathbf{y}},\tau _{s}^{(\omega ,\lambda )}(B)]\right) \mathrm{d}s+\mathcal{O}%
\left( \eta ^{2}\right)
\end{equation*}%
for any $B\in \mathcal{U}$. If $|\eta |\ll 1$, then the diamagnetic current
behaves as
\begin{equation*}
\mathbb{J}_{\mathrm{d}}^{(\eta ,\mathbf{y})}:=\varrho ^{(\beta ,\omega
,\lambda )}(\tilde{\tau}_{t}^{(\eta ,\mathbf{y})}(\mathrm{I}_{\mathbf{y}%
}^{\eta \mathbf{A}^{(\mathbf{y})}}))=\eta \varrho ^{(\beta ,\omega ,\lambda
)}\left( P_{\mathbf{y}}\right) +\mathcal{O}\left( \eta ^{2}\left\vert
t\right\vert \right)
\end{equation*}%
with $\varrho ^{(\beta ,\omega ,\lambda )}\left( P_{\mathbf{y}}\right) =%
\mathcal{O}\left( 1\right) $, see (\ref{R x}) and (\ref{dia-adjency}). On
the other hand, by (\ref{stationary}) and (\ref{backwards -1bis}), the
so--called paramagnetic current
\begin{equation*}
\mathbb{J}_{\mathrm{p}}^{(\eta ,\mathbf{y})}\left( \mathbf{x},t\right)
:=\varrho ^{(\beta ,\omega ,\lambda )}(\tilde{\tau}_{t}^{(\eta ,\mathbf{y}%
)}(I_{\mathbf{x}}))-\varrho ^{(\beta ,\omega ,\lambda )}\left( I_{\mathbf{x}%
}\right)
\end{equation*}%
satisfies%
\begin{equation*}
\partial _{t}\mathbb{J}_{\mathrm{p}}^{(\eta ,\mathbf{y})}\left( \mathbf{x}%
,t\right) =\mathbb{J}_{\mathrm{d}}^{(\eta ,\mathbf{y})}\mathfrak{v}^{(%
\mathbf{y})}\left( \mathbf{x},t\right) +\mathcal{O}\left( |\mathbb{J}_{%
\mathrm{d}}^{(\eta ,\mathbf{y})}|^{2}\left\vert t\right\vert \right)
\end{equation*}%
for any $\mathbf{x,y}\in \mathfrak{K}$ and $t\in \mathbb{R}$, where
\begin{equation}
\mathfrak{v}^{(\mathbf{y})}\left( \mathbf{x},t\right) :=\frac{1}{\varrho
^{(\beta ,\omega ,\lambda )}\left( P_{\mathbf{y}}\right) }\varrho ^{(\beta
,\omega ,\lambda )}\left( i[I_{\mathbf{y}},\tau _{t}^{(\omega ,\lambda )}(I_{%
\mathbf{x}})]\right) =\frac{\partial _{t}\sigma _{\mathrm{p}}^{(\omega
)}\left( \mathbf{x},\mathbf{y},t\right) }{\sigma _{\mathrm{d}}^{(\omega
)}\left( \mathbf{y}\right) }\ .  \label{quantum viscosity}
\end{equation}%
In other words, $\mathfrak{v}$ can be interpreted as a (time--dependent)
\emph{quantum current viscosity}.

For any $l,\beta \in \mathbb{R}^{+}$, $\omega \in \Omega $ and $\lambda \in
\mathbb{R}_{0}^{+}$ we define two further important functions, the analogues
of $\sigma _{\mathrm{p}}^{(\omega )}$ and $\sigma _{\mathrm{d}}^{(\omega )}$%
, associated with the observables $\mathbb{I}_{k,l}$ and $\mathbb{P}_{k,l}$:

\begin{itemize}
\item[(p)] The space--averaged \emph{paramagnetic} transport coefficient
\begin{equation*}
t\mapsto \Xi _{\mathrm{p},l}^{(\omega )}\left( t\right) \equiv \Xi _{\mathrm{%
p},l}^{(\beta ,\omega ,\lambda )}\left( t\right) \in \mathcal{B}(\mathbb{R}%
^{d})
\end{equation*}%
is defined, w.r.t. the canonical orthonormal basis of $\mathbb{R}^{d}$, by%
\begin{equation}
\left\{ \Xi _{\mathrm{p},l}^{(\omega )}\left( t\right) \right\} _{k,q}:=%
\frac{1}{\left\vert \Lambda _{l}\right\vert }\int\nolimits_{0}^{t}\varrho
^{(\beta ,\omega ,\lambda )}\left( i[\mathbb{I}_{k,l},\tau _{s}^{(\omega
,\lambda )}(\mathbb{I}_{q,l})]\right) \mathrm{d}s
\label{average microscopic AC--conductivity}
\end{equation}%
for any $k,q\in \{1,\ldots ,d\}$ and $t\in \mathbb{R}$.

\item[(d)] The space--averaged \emph{diamagnetic} transport coefficient
\begin{equation*}
\Xi _{\mathrm{d},l}^{(\omega )}\equiv \Xi _{\mathrm{d},l}^{(\beta ,\omega
,\lambda )}\in \mathcal{B}(\mathbb{R}^{d})
\end{equation*}%
corresponds to the diagonal matrix defined by%
\begin{equation}
\left\{ \Xi _{\mathrm{d},l}^{(\omega )}\right\} _{k,q}:=\frac{\delta _{k,q}}{%
\left\vert \Lambda _{l}\right\vert }\varrho ^{(\beta ,\omega ,\lambda
)}\left( \mathbb{P}_{k,l}\right) \ ,\quad k,q\in \{1,\ldots ,d\}\ .
\label{average microscopic AC--conductivity dia}
\end{equation}
\end{itemize}

Of course, by (\ref{current density=}) and (\ref{backwards -1bis})--(\ref%
{backwards -1bispara}),
\begin{equation}
\left\{ \Xi _{\mathrm{p},l}^{(\omega )}\left( t\right) \right\} _{k,q}=\frac{%
1}{\left\vert \Lambda _{l}\right\vert }\underset{x,y\in \Lambda _{l}}{\sum }%
\sigma _{\mathrm{p}}^{(\omega )}\left( x+e_{q},x,y+e_{k},y,t\right)
\label{average conductivity}
\end{equation}%
for any $l,\beta \in \mathbb{R}^{+}$, $\omega \in \Omega $, $\lambda \in
\mathbb{R}_{0}^{+}$, $k,q\in \{1,\ldots ,d\}$ and $t\in \mathbb{R}$, while%
\begin{equation}
\left\{ \Xi _{\mathrm{d},l}^{(\omega )}\right\} _{k,k}=\frac{1}{\left\vert
\Lambda _{l}\right\vert }\underset{x\in \Lambda _{l}}{\sum }\sigma _{\mathrm{%
d}}^{(\omega )}\left( x+e_{k},x\right) \ .  \label{average conductivity +1}
\end{equation}%
Both coefficients are typically the paramagnetic and diamagnetic
conductivity one experimentally measures for large samples, i.e., large
enough boxes $\Lambda _{l}$. Indeed, we show in \cite{OhmIII} that the
limits $l\rightarrow \infty $ of $\Xi _{\mathrm{p},l}^{(\omega )}$ and $\Xi
_{\mathrm{d},l}^{(\omega )}$ generally exist and define so--called
macroscopic paramagnetic and diamagnetic conductivities. Before going
further, we first discuss some important mathematical properties of $\Xi _{%
\mathrm{p},l}^{(\omega )}$ and $\Xi _{\mathrm{d},l}^{(\omega )}$.

By using the scalar product $\langle \cdot ,\cdot \rangle $ in $\ell ^{2}(%
\mathfrak{L})$, the canonical orthonormal basis $\left\{ \mathfrak{e}%
_{x}\right\} _{x\in \mathfrak{L}}$ of $\ell ^{2}(\mathfrak{L})$ and the
symbol $\mathbf{d}_{\mathrm{fermi}}^{(\beta ,\omega ,\lambda )}$ defined by (%
\ref{Fermi statistic}), we observe from (\ref{average conductivity +1}) that%
\begin{equation}
\left\{ \Xi _{\mathrm{d},l}^{(\omega )}\right\} _{k,k}=\frac{2}{\left\vert
\Lambda _{l}\right\vert }\underset{x\in \Lambda _{l}}{\sum }\mathrm{Re}%
\left\{ \langle \mathfrak{e}_{x+e_{k}},\mathbf{d}_{\mathrm{fermi}}^{(\beta
,\omega ,\lambda )}\mathfrak{e}_{x}\rangle \right\} \in \left[ -2,2\right]
\label{auto-evident}
\end{equation}%
for any $l,\beta \in \mathbb{R}^{+}$, $\omega \in \Omega $, $\lambda \in
\mathbb{R}_{0}^{+}$ and $k\in \{1,\ldots ,d\}$.

The main property of the paramagnetic transport coefficient $\Xi _{\mathrm{p}%
,l}^{(\omega )}$ is proven in Section \ref{Section Local AC--Conductivity0
copy(2)} and given in the next theorem. To present it, we introduce the
notation $\mathcal{B}_{+}(\mathbb{R}^{d})\subset \mathcal{B}(\mathbb{R}^{d})$
for the set of positive linear operators on $\mathbb{R}^{d}$. For any $%
\mathcal{B}(\mathbb{R}^{d})$--valued measure $\mu $ on $\mathbb{R}$, we
additionally denote by $\Vert \mu \Vert _{\mathrm{op}}$ the measure on $%
\mathbb{R}$ taking values in $\mathbb{R}_{0}^{+}$ that is defined, for any
Borel set $\mathcal{X}$, by
\begin{equation}
\Vert \mu \Vert _{\mathrm{op}}\left( \mathcal{X}\right) :=\sup \left\{
\underset{i\in I}{\sum }\Vert \mu \left( \mathcal{X}_{i}\right) \Vert _{%
\mathrm{op}}:\{\mathcal{X}_{i}\}_{i\in I}\text{ is a finite Borel partition
of }\mathcal{X}\right\} \ .  \label{definion opera measure}
\end{equation}%
We, moreover, say that $\mu $ is symmetric if $\mu (\mathcal{X})=\mu (-%
\mathcal{X})$ for any Borel set $\mathcal{X}\subset \mathbb{R}$. With these
definitions we have the following assertion:

\begin{satz}[Microscopic paramagnetic conductivity measures]
\label{lemma sigma pos type copy(4)}\mbox{
}\newline
For any $l,\beta \in \mathbb{R}^{+}$, $\omega \in \Omega $ and $\lambda \in
\mathbb{R}_{0}^{+}$, there exists a non--zero symmetric $\mathcal{B}_{+}(%
\mathbb{R}^{d})$--valued measure $\mu _{\mathrm{p},l}^{(\omega )}\equiv \mu
_{\mathrm{p},l}^{(\beta ,\omega ,\lambda )}$ on $\mathbb{R}$ such that%
\begin{equation}
\int_{\mathbb{R}}\left( 1+\left\vert \nu \right\vert \right) \Vert \mu _{%
\mathrm{p},l}^{(\omega )}\Vert _{\mathrm{op}}(\mathrm{d}\nu )<\infty \ ,
\label{unifrom bound}
\end{equation}%
uniformly w.r.t. $l,\beta \in \mathbb{R}^{+}$, $\omega \in \Omega $, $%
\lambda \in \mathbb{R}_{0}^{+}$, and%
\begin{equation*}
\Xi _{\mathrm{p},l}^{(\omega )}(t)=\int_{\mathbb{R}}\left( \cos \left( t\nu
\right) -1\right) \mu _{\mathrm{p},l}^{(\omega )}(\mathrm{d}\nu )\ ,\qquad
t\in \mathbb{R}\ .
\end{equation*}
\end{satz}

\begin{proof}
The assertions follow from Theorems \ref{lemma sigma pos type copy(3)} and %
\ref{lemma sigma pos type copy(1)} combined with Corollary \ref{Corollary
Stationarity copy(1)} and Lemma \ref{bound incr 1 Lemma copy(3)}.
\end{proof}

\begin{koro}[Properties of the microscopic paramagnetic conductivity]
\label{lemma sigma pos type}\mbox{
}\newline
For $l,\beta \in \mathbb{R}^{+}$, $\omega \in \Omega $ and $\lambda \in
\mathbb{R}_{0}^{+}$, $\Xi _{\mathrm{p},l}^{(\omega )}$ has the following
properties:\newline
\emph{(i)} Time--reversal symmetry: $\Xi _{\mathrm{p},l}^{(\omega )}\left(
0\right) =0$ and
\begin{equation*}
\Xi _{\mathrm{p},l}^{(\omega )}\left( -t\right) =\Xi _{\mathrm{p}%
,l}^{(\omega )}\left( t\right) \ ,\qquad t\in \mathbb{R}\ .
\end{equation*}%
\emph{(ii)} Negativity of $\Xi _{\mathrm{p},l}^{(\omega )}$:
\begin{equation*}
\Xi _{\mathrm{p},l}^{(\omega )}\left( t\right) \leq 0\ ,\qquad t\in \mathbb{R%
}\ .
\end{equation*}%
\emph{(iii)} Ces\`{a}ro mean of $\Xi _{\mathrm{p},l}^{(\omega )}$:
\begin{equation*}
\underset{t\rightarrow \infty }{\lim }\ \frac{1}{t}\int_{0}^{t}\Xi _{\mathrm{%
p},l}^{(\omega )}\left( s\right) \mathrm{d}s=-\mu _{\mathrm{p},l}^{(\omega
)}\left( \mathbb{R}\backslash \left\{ 0\right\} \right) \leq 0\ .
\end{equation*}%
\emph{(iv)} Equicontinuity: The family $\{\Xi _{\mathrm{p},l}^{(\beta
,\omega ,\lambda )}\}_{l,\beta \in \mathbb{R}^{+},\omega \in \Omega ,\lambda
\in \mathbb{R}_{0}^{+}}$ of maps from $\mathbb{R}$ to $\mathcal{B}(\mathbb{R}%
^{d})$ is equicontinuous. \newline
\emph{(v)} Macroscopic paramagnetic conductivity measures: The family $\{\mu
_{\mathrm{p},l}^{(\omega )}\}_{l\in \mathbb{R}^{+}}$ has weak$^{\ast }$%
--accumulation points.
\end{koro}

\begin{proof}
(i)--(iii) are direct consequences of Theorem \ref{lemma sigma pos type
copy(4)} and Lebesgue's dominated convergence theorem. To prove (iv),
observe that the uniform bound (\ref{unifrom bound}) implies that, for any $%
\nu _{0}\in \mathbb{R}_{0}^{+}$,%
\begin{equation*}
\mu _{\mathrm{p},l}^{(\omega )}\left( \mathbb{R}\backslash \left[ -\nu
_{0},\nu _{0}\right] \right) =\mathcal{O}\left( \nu _{0}^{-1}\right)
\end{equation*}%
uniformly w.r.t. $l,\beta \in \mathbb{R}^{+}$, $\omega \in \Omega $, $%
\lambda \in \mathbb{R}_{0}^{+}$. (v) follows from Theorem \ref{lemma sigma
pos type copy(4)} and the weak$^{\ast }$--compactness of the unit ball in
the set of measures on $\mathbb{R}$ taking values in the set of positive
elements of $\mathcal{B}(\mathbb{R}^{d})$.
\end{proof}

The $\mathcal{B}_{+}(\mathbb{R}^{d})$--valued measures $\mu _{\mathrm{p}%
,l}^{(\omega )}$ can be represented in terms of the spectral measure of an
explicit self--adjoint operator w.r.t. explicitly given vectors, see
Equation (\ref{von braun}). From this representation, one concludes for
instance that, if the operator $(\Delta _{\mathrm{d}}+\lambda V_{\omega })$
has purely (absolutely) continuous spectrum (as for $\lambda =0$) then, for
any $k,q\in \{1,\ldots ,d\}$,%
\begin{equation*}
\left\{ \mu _{\mathrm{p},l}^{(\omega )}\left( \mathbb{R}\backslash \left\{
0\right\} \right) \right\} _{k,q}=\frac{1}{\left\vert \Lambda
_{l}\right\vert }(\mathbb{I}_{k,l},\mathbb{I}_{q,l})_{\sim }^{(\omega )}\ .
\end{equation*}%
Here, $(\cdot ,\cdot )_{\sim }^{(\omega )}$ is the Duhamel two--point
function $(\cdot ,\cdot )_{\sim }^{(\omega )}$, which is studied in detail
in Section \ref{Section Duhamel Two--Point Functions}. In fact, the constant
$\mu _{\mathrm{p},l}^{(\omega )}\left( \mathbb{R}\backslash \left\{
0\right\} \right) $ is the so--called static admittance of linear response
theory, see Theorem \ref{lemma sigma pos type copy(8)}. Moreover, Theorem %
\ref{lemma sigma pos type copy(6)} explains how $\mu _{\mathrm{p}%
,l}^{(\omega )}$ can also be constructed from the \emph{space--averaged }%
quantum current viscosity%
\begin{equation}
\mathbf{V}_{l}^{(\omega )}\left( t\right) :=\left( \Xi _{\mathrm{d}%
,l}^{(\omega )}\right) ^{-1}\partial _{t}\Xi _{\mathrm{p},l}^{(\omega
)}\left( t\right) \in \mathcal{B}(\mathbb{R}^{d})
\label{quantum viscosity bis bis}
\end{equation}%
for any $l,\beta \in \mathbb{R}^{+}$, $\omega \in \Omega $, $\lambda \in
\mathbb{R}_{0}^{+}$ and $t\in \mathbb{R}$. Compare with (\ref{quantum
viscosity}). More precisely, it is the boundary value of the (imaginary part
of the) Laplace--Fourier transform of $\Xi _{\mathrm{d},l}^{(\omega )}%
\mathbf{V}_{l}^{(\omega )}$.

Recall that, as asserted in Theorem \ref{lemma sigma pos type copy(4)}, the
measure $\mu _{\mathrm{p},l}^{(\omega )}$ is never the zero--measure.
Nevertheless, it is a priori not clear whether the weak$^{\ast }$%
--accumulation points of the family $\{\mu _{\mathrm{p},l}^{(\omega
)}\}_{l\in \mathbb{R}^{+}}$ also have this property. We show in a companion
paper that, as $l\rightarrow \infty $, the measure $\mu _{\mathrm{p}%
,l}^{(\omega )}$ converges to the zero--measure if $\lambda =0$ but, for $%
\lambda \in \mathbb{R}^{+}$, there is generally a unique weak$^{\ast }$%
--accumulation point of $\{\mu _{\mathrm{p},l}^{(\omega )}\}_{l\in \mathbb{R}%
^{+}}$, which is not the zero--measure.

\subsection{Paramagnetic and Diamagnetic Currents\label{Sect para dia
current}}

Recall that we assume in this section that the current results from a
space--homogeneous electric field $\eta \mathcal{E}_{t}\vec{w}$ at time $%
t\in \mathbb{R}$ in the box $\Lambda _{l}$, where $\vec{w}:=(w_{1},\ldots
,w_{d})\in \mathbb{R}^{d}$, $\mathcal{E}_{t}:=-\partial _{t}\mathcal{A}_{t}$
for all $t\in \mathbb{R}$, and $\mathcal{A}\in C_{0}^{\infty }\left( \mathbb{%
R};\mathbb{R}\right) $. This electric field corresponds to the (rescaled)
electromagnetic potential $\eta \mathbf{\bar{A}}_{l}$. We also remind that $%
\{e_{k}\}_{k=1}^{d}$ is the canonical orthonormal basis of the Euclidian
space $\mathbb{R}^{d}$.

Generally, even in the absence of electromagnetic fields, i.e., if $\eta =0$%
, there exist (thermal) currents coming from the inhomogeneity of the
fermion system for $\lambda \in \mathbb{R}^{+}$. For any $l,\beta \in
\mathbb{R}^{+}$, $\omega \in \Omega $, $\lambda \in \mathbb{R}_{0}^{+}$ and $%
k\in \{1,\ldots ,d\}$,%
\begin{equation}
\mathbb{J}_{k,l}^{(\omega )}\equiv \mathbb{J}_{k,l}^{(\beta ,\omega ,\lambda
)}:=\left\vert \Lambda _{l}\right\vert ^{-1}\underset{x\in \Lambda _{l}}{%
\sum }\varrho ^{(\beta ,\omega ,\lambda )}(I_{(x+e_{k},x)})
\label{free current}
\end{equation}%
is the density of current along the direction $e_{k}$ in the box $\Lambda
_{l}$. In the space--homogeneous case, by symmetry, $\mathbb{J}%
_{k,l}^{(\omega )}=0$ but in general, $\mathbb{J}_{k,l}^{(\omega )}\neq 0$.
We prove in \cite{OhmIII} that
\begin{equation*}
\underset{l\rightarrow \infty }{\lim }\,\mathbb{J}_{k,l}^{(\omega )}=0
\end{equation*}%
almost surely if $\omega \in \Omega $ is the realization of some ergodic
random potential.

Then, for any $l,\beta \in \mathbb{R}^{+}$, $\omega \in \Omega $, $\lambda
\in \mathbb{R}_{0}^{+}$, $\eta \in \mathbb{R}$, $\vec{w}\in \mathbb{R}^{d}$,
$\mathcal{A}\in C_{0}^{\infty }\left( \mathbb{R};\mathbb{R}\right) $ and $%
t\geq t_{0}$, the (increment of) current density resulting from the
space--homogeneous electric perturbation $\mathcal{E}$ in the box $\Lambda
_{l}$ is the sum of two current densities defined from (\ref{current
density=}):

\begin{itemize}
\item[(p)] The paramagnetic current density
\begin{equation*}
\mathbb{J}_{\mathrm{p}}^{(\omega ,\eta \mathbf{\bar{A}}_{l})}\left( t\right)
\equiv \mathbb{J}_{\mathrm{p}}^{(\beta ,\omega ,\lambda ,\eta \mathbf{\bar{A}%
}_{l})}\left( t\right) \in \mathbb{R}^{d}
\end{equation*}%
is defined by the space average of the current increment vector inside the
box $\Lambda _{l}$, that is for any $k\in \{1,\ldots ,d\}$,
\begin{equation}
\left\{ \mathbb{J}_{\mathrm{p}}^{(\omega ,\eta \mathbf{\bar{A}}_{l})}\left(
t\right) \right\} _{k}:=\left\vert \Lambda _{l}\right\vert ^{-1}\rho
_{t}^{(\beta ,\omega ,\lambda ,\eta \mathbf{\bar{A}}_{l})}(\mathbb{I}%
_{k,l})\ .  \label{finite volume current density}
\end{equation}

\item[(d)] The diamagnetic (or ballistic) current density
\begin{equation*}
\mathbb{J}_{\mathrm{d}}^{(\omega ,\eta \mathbf{\bar{A}}_{l})}\left( t\right)
\equiv \mathbb{J}_{\mathrm{d}}^{(\beta ,\omega ,\lambda ,\eta \mathbf{\bar{A}%
}_{l})}\left( t\right) \in \mathbb{R}^{d}
\end{equation*}%
is defined analogously, for any $k\in \{1,\ldots ,d\}$, by
\begin{equation}
\left\{ \mathbb{J}_{\mathrm{d}}^{(\omega ,\eta \mathbf{\bar{A}}_{l})}\left(
t\right) \right\} _{k}:=\left\vert \Lambda _{l}\right\vert ^{-1}\rho
_{t}^{(\beta ,\omega ,\lambda ,\eta \mathbf{\bar{A}}_{l})}(\mathbf{I}%
_{k,l}^{\eta \mathbf{\bar{A}}_{l}})\ .
\label{finite volume current density2}
\end{equation}
\end{itemize}

The paramagnetic current density is only related to the \emph{change of
internal state} $\rho _{t}^{(\beta ,\omega ,\lambda ,\mathbf{A})}$ produced
by the electromagnetic field. We will show below that these currents carry
the paramagnetic energy increment defined in Section \ref{section
para-dia-energy}. The diamagnetic current density corresponds to a raw
ballistic flow of charged particles caused by the electric field, at thermal
equilibrium. It directly comes from the change of the electromagnetic
potential expressed in terms of the observable (\ref{eq def W}) defined
below. We will show that it yields the diamagnetic energy defined in Section %
\ref{section para-dia-energy}. With this, diamagnetic and paramagnetic
currents are respectively \textquotedblleft first order\textquotedblright\
and \textquotedblleft second order\textquotedblright\ with respect to
changes of the electromagnetic potentials and thus have different physical
properties. See for instance Theorems \ref{thm Local Ohm's law} and \ref%
{Local Ohm's law thm copy(2)}.

\subsection{Current Linear Response\label{Sec Ohm law linear response}}

We are now in position to derive a microscopic version of Ohm's law. We use
the space--avera%
%TCIMACRO{\TeXButton{\-}{\-}}%
%BeginExpansion
\-%
%EndExpansion
ged paramagnetic and diamagnetic transport coefficients $\Xi _{\mathrm{p}%
,l}^{(\omega )}$ (\ref{average microscopic AC--conductivity}) and $\Xi _{%
\mathrm{d},l}^{(\omega )}$ (\ref{average microscopic AC--conductivity dia})
to define the $\mathbb{R}^{d}$--valued functions
\begin{equation*}
J_{\mathrm{p},l}^{(\omega ,\mathcal{A})}\equiv J_{\mathrm{p},l}^{(\beta
,\omega ,\lambda ,\vec{w},\mathcal{A})}\qquad \text{and}\qquad J_{\mathrm{d}%
,l}^{(\omega ,\mathcal{A})}\equiv J_{\mathrm{d},l}^{(\beta ,\omega ,\lambda ,%
\vec{w},\mathcal{A})}
\end{equation*}
by%
\begin{eqnarray}
J_{\mathrm{p},l}^{(\omega ,\mathcal{A})}(t) &:=&\int_{t_{0}}^{t}\left( \Xi _{%
\mathrm{p},l}^{(\omega )}\left( t-s\right) \vec{w}\right) \mathcal{E}_{s}%
\mathrm{d}s\ ,\qquad t\geq t_{0}\ ,  \label{linear responses1} \\
J_{\mathrm{d},l}^{(\omega ,\mathcal{A})}(t) &:=&\left( \Xi _{\mathrm{d}%
,l}^{(\omega )}\vec{w}\right) \int_{t_{0}}^{t}\mathcal{E}_{s}\mathrm{d}s\
,\qquad t\geq t_{0}\ ,  \label{linear responses2}
\end{eqnarray}%
for any $l,\beta \in \mathbb{R}^{+}$, $\omega \in \Omega $, $\lambda \in
\mathbb{R}_{0}^{+}$, $\vec{w}\in \mathbb{R}^{d}$ and $\mathcal{A}\in
C_{0}^{\infty }\left( \mathbb{R};\mathbb{R}\right) $. They are the linear
responses of the paramagnetic and diamagnetic current densities,
respectively:

\begin{satz}[Microscopic Ohm's law]
\label{thm Local Ohm's law}\mbox{
}\newline
For any $\vec{w}\in \mathbb{R}^{d}$ and $\mathcal{A}\in C_{0}^{\infty
}\left( \mathbb{R};\mathbb{R}\right) $, there is $\eta _{0}\in \mathbb{R}%
^{+} $ such that, for $|\eta |\in \lbrack 0,\eta _{0}]$,
\begin{equation*}
\mathbb{J}_{\mathrm{p}}^{(\omega ,\eta \mathbf{\bar{A}}_{l})}\left( t\right)
=\eta J_{\mathrm{p},l}^{(\omega ,\mathcal{A})}(t)+\mathcal{O}\left( \eta
^{2}\right) \qquad \text{and}\qquad \mathbb{J}_{\mathrm{d}}^{(\omega ,\eta
\mathbf{\bar{A}}_{l})}\left( t\right) =\eta J_{\mathrm{d},l}^{(\omega ,%
\mathcal{A})}(t)+\mathcal{O}\left( \eta ^{2}\right) \ ,
\end{equation*}%
uniformly for $l,\beta \in \mathbb{R}^{+}$, $\omega \in \Omega $, $\lambda
\in \mathbb{R}_{0}^{+}$ and $t\geq t_{0}$.
\end{satz}

\begin{proof}
See Lemmata \ref{Lemma LR para}--\ref{Lemma LR dia}.
\end{proof}

\noindent The fact that the asymptotics obtained are uniform w.r.t. $l,\beta
\in \mathbb{R}^{+}$, $\omega \in \Omega $, $\lambda \in \mathbb{R}_{0}^{+}$
and $t\geq t_{0}$ is a crucial property to get macroscopic Ohm's law in \cite%
{OhmIII}. Note also that Theorem \ref{thm Local Ohm's law} can easily be
extended to macroscopically space--inhomogeneous electromagnetic fields,
that is, for all space--rescaled vector potentials $\mathbf{A}_{l}$ (\ref%
{rescaled vector potential}) with $\mathbf{A}\in \mathbf{C}_{0}^{\infty }$,
by exactly the same methods as in the proof of Theorem \ref{Local Ohm's law
thm copy(2)}. We refrain from doing it at this point, for technical
simplicity. The result above can indeed be deduced from Theorem \ref{Local
Ohm's law thm copy(2)}, see Equations (\ref{current para})--(\ref{current
dia}).

As a consequence, $\Xi _{\mathrm{p},l}^{(\omega )}$ and $\Xi _{\mathrm{d}%
,l}^{(\omega )}$ can be interpreted as \emph{charge} transport coefficients.
Observe that $\Xi _{\mathrm{p},l}^{(\omega )}\left( 0\right) =0$, by
Corollary \ref{lemma sigma pos type} (i). Therefore, when the electric field
is switched on, it accelerates the charged particles and first induces
diamagnetic currents, cf. (\ref{linear responses2}). This creates a kind of
\textquotedblleft wave front\textquotedblright\ that destabilizes the whole
system by changing its internal state. By the phenomenon of current
viscosity discussed in Section \ref{Sect Trans coeef ddef}, the presence of
such diamagnetic currents leads to the progressive appearance of
paramagnetic currents. We prove in Section \ref{Sect Joule} that these
paramagnetic currents are responsible for heat production and modify as well
the electromagnetic potential energy of charge carriers. Indeed, the
positive measures of Theorem \ref{lemma sigma pos type copy(4)} are directly
related to heat production (cf. Section \ref{Section Joule effect}) and are
the boundary values of the (imaginary part of the) Laplace--Fourier
transforms of the current viscosities as discussed in the previous section.

Note that Theorem \ref{thm Local Ohm's law} also leads to (finite--volume)%
\emph{\ Green--Kubo relations}, by (\ref{average microscopic
AC--conductivity}) and (\ref{linear responses1}). Indeed, by (\ref{current
density=}), $\left\vert \Lambda _{l}\right\vert ^{-\frac{1}{2}}\mathbb{I}%
_{k,l}$ is a \emph{current fluctuation} and (\ref{average microscopic
AC--conductivity}) gives:%
\begin{equation}
\left\{ \Xi _{\mathrm{p},l}^{(\omega )}\left( t\right) \right\}
_{k,q}=\int\nolimits_{0}^{t}\varrho ^{(\beta ,\omega ,\lambda )}\left( i%
\left[ \left\vert \Lambda _{l}\right\vert ^{-\frac{1}{2}}\mathbb{I}%
_{k,l},\left\vert \Lambda _{l}\right\vert ^{-\frac{1}{2}}\tau _{s}^{(\omega
,\lambda )}(\mathbb{I}_{q,l})\right] \right) \mathrm{d}s  \label{green kubo}
\end{equation}%
for any $l,\beta \in \mathbb{R}^{+}$, $\omega \in \Omega $, $\lambda \in
\mathbb{R}_{0}^{+}$, $t\in \mathbb{R}$ and $k,q\in \{1,\ldots ,d\}$. In the
limit $l\rightarrow \infty $ we show in \cite{OhmIII} that $\Xi _{\mathrm{p}%
,l}^{(\omega )}$ is related to a quasi--free dynamics on the CCR algebra of
(current) fluctuations.

Theorem \ref{thm Local Ohm's law} together with (\ref{linear responses1})--(%
\ref{linear responses2}) gives a natural notion of linear conductivity of
the fermion system in the box $\Lambda _{l}$: It is the map%
\begin{equation*}
t\mapsto \mathbf{\Sigma }_{l}^{(\omega )}\equiv \mathbf{\Sigma }_{l}^{(\beta
,\omega ,\lambda )}\left( t\right) \in \mathcal{B}(\mathbb{R}^{d})
\end{equation*}%
defined by%
\begin{equation}
\mathbf{\Sigma }_{l}^{(\omega )}(t):=\left\{
\begin{array}{lll}
0 & , & \qquad t\leq 0\ , \\
\Xi _{\mathrm{d},l}^{(\omega )}+\Xi _{\mathrm{p},l}^{(\omega )}\left(
t\right) & , & \qquad t\geq 0\ ,%
\end{array}%
\right.  \label{local cond}
\end{equation}%
for $l,\beta \in \mathbb{R}^{+}$, $\omega \in \Omega $, $\lambda \in \mathbb{%
R}_{0}^{+}$. The \emph{total} current%
\begin{equation*}
J_{l}^{(\omega ,\mathcal{A})}(t):=J_{\mathrm{p},l}^{(\omega ,\mathcal{A}%
)}(t)+J_{\mathrm{d},l}^{(\omega ,\mathcal{A})}(t)\ ,\qquad t\geq t_{0}\ ,
\end{equation*}%
which as in \cite[Eq. (A2.14)]{dia-current} is the sum of paramagnetic and
diamagnetic current densities, has the following linear response:
\begin{equation}
J_{l}^{(\omega ,\mathcal{A})}(t)=\int\nolimits_{\mathbb{R}}\left( \mathbf{%
\Sigma }_{l}^{(\omega )}\left( t-s\right) \vec{w}\right) \mathcal{E}_{s}%
\mathrm{d}s=\left(
\begin{array}{c}
\{\mathbf{\Sigma }_{l}^{(\omega )}\vec{w}\}_{1}\mathbf{\ast }\mathcal{E} \\
\vdots \\
\{\mathbf{\Sigma }_{l}^{(\omega )}\vec{w}\}_{d}\mathbf{\ast }\mathcal{E}%
\end{array}%
\right) \ .  \label{linear responses3}
\end{equation}%
In particular, if the electric field stays constant for sufficiently large
times, i.e., $\mathcal{E}_{t}=D$ for arbitrary large times $t\in \lbrack
T,\infty )$ with $T>t_{0}$, then in the situation where $t\gg T$, i.e., in
the DC--regime, we deduce from Corollary \ref{lemma sigma pos type} (iii)
and (\ref{local cond})--(\ref{linear responses3}) that
\begin{equation}
\left\vert t\right\vert ^{-1}J_{l}^{(\omega ,\mathcal{A})}(t)=D(\Xi _{%
\mathrm{d},l}^{(\omega )}-\mu _{\mathrm{p},l}^{(\omega )}\left( \mathbb{R}%
\backslash \left\{ 0\right\} \right)) +o\left( 1\right) \ .
\label{linear responses4}
\end{equation}%
It is not a priori clear whether $\mu _{\mathrm{p},l}^{(\omega )}\left(
\mathbb{R}\backslash \left\{ 0\right\} \right) =\Xi _{\mathrm{d},l}^{(\omega
)}$ or not. We prove in \cite{OhmIII} that this last equality actually holds
in the limit $l\rightarrow \infty $. [Recall that $\mathbf{A}\in \mathbf{C}%
_{0}^{\infty }$ is compactly supported in space and time, but it can be
switched off at arbitrary large times.]

In order to express the \emph{in--phase current} from (\ref{linear
responses3}), we define by $\mathbf{\Sigma }_{l,+}^{(\omega )}$ the
symmetrization of $\mathbf{\Sigma }_{l}^{(\omega )}$, that is,
\begin{equation}
\mathbf{\Sigma }_{l,+}^{(\omega )}\left( t\right) :=\mathbf{\Sigma }%
_{l}^{(\omega )}\left( \left\vert t\right\vert \right) =\Xi _{\mathrm{d}%
,l}^{(\omega )}+\Xi _{\mathrm{p},l}^{(\omega )}\left( t\right) \ ,\qquad
t\in \mathbb{R}\ ,  \label{symme conduc}
\end{equation}%
see Corollary \ref{lemma sigma pos type} (i). Similarly, the
anti--symmetrization $\mathbf{\Sigma }_{l,-}^{(\omega )}$ of $\mathbf{\Sigma
}_{l}^{(\omega )}$ is given by
\begin{equation}
\mathbf{\Sigma }_{l,-}^{(\omega )}:=\mathrm{sign}(t)\mathbf{\Sigma }%
_{l}^{(\omega )}\left( \left\vert t\right\vert \right) \ ,\qquad t\in
\mathbb{R}\ .
\end{equation}%
With these definitions the current linear response (\ref{linear responses3})
equals%
\begin{equation}
J_{l}^{(\omega ,\mathcal{A})}(t)=\frac{1}{2}\int\nolimits_{\mathbb{R}}\left(
\mathbf{\Sigma }_{l,+}^{(\omega )}\left( t-s\right) \vec{w}\right) \mathcal{E%
}_{s}\mathrm{d}s+\frac{1}{2}\int\nolimits_{\mathbb{R}}\left( \mathbf{\Sigma }%
_{l,-}^{(\omega )}\left( t-s\right) \vec{w}\right) \mathcal{E}_{s}\mathrm{d}%
s\ .  \label{linear responses3bis}
\end{equation}%
The first part in the right hand side of this equality is by definition the
in--phase current.

This last equation is directly related to Ohm's law in Fourier space:
Similar to \cite{Annale}, it is indeed natural to define the \emph{%
conductivity measure} $\mu _{\Lambda _{l}}^{(\omega )}\equiv \mu _{\Lambda
_{l}}^{(\beta ,\omega ,\lambda )}$ as being the Fourier transform of $%
\mathbf{\Sigma }_{l,+}^{(\omega )}\left( t\right) $. By Theorem \ref{lemma
sigma pos type copy(4)} and (\ref{symme conduc}),
\begin{equation*}
\mu _{\Lambda _{l}}^{(\omega )}(\mathcal{X})=\mu _{\mathrm{p},l}^{(\omega )}(%
\mathcal{X})+(\Xi _{\mathrm{d},l}^{(\omega )}-\mu _{\mathrm{p},l}^{(\omega
)}\left( \mathbb{R}\right) )\mathbf{1}\left[ 0\in \mathcal{X}\right]
\end{equation*}%
with $\mathcal{X}\subset \mathbb{R}$ being any Borel set. Therefore, we can
rewrite the current linear response (\ref{linear responses3bis}) as%
\begin{equation}
J_{l}^{(\omega ,\mathcal{A})}(t)=\frac{1}{2}\int_{\mathbb{R}}\mathcal{\hat{E}%
}_{\nu }^{(t)}\mu _{\Lambda _{l}}^{(\omega )}\left( \mathrm{d}\nu \right)
\vec{w}+\frac{i}{2}\int_{\mathbb{R}}\mathbb{H}(\mathcal{\hat{E}}%
^{(t)})\left( \nu \right) \mu _{\Lambda _{l}}^{(\omega )}\left( \mathrm{d}%
\nu \right) \vec{w}  \label{eq sup curernt}
\end{equation}%
with $\mathcal{\hat{E}}$ being the Fourier transform of $\mathcal{E}$, $%
\mathcal{\hat{E}}_{\nu }^{(t)}:=\mathrm{e}^{i\nu t}\mathcal{\hat{E}}_{\nu }$%
, and where $\mathbb{H}$ is the Hilbert transform, i.e.,
\begin{equation*}
\mathbb{H}\left( f\right) \left( \nu \right) :=-\frac{1}{\pi }\ \underset{%
\varepsilon \rightarrow 0^{+}}{\lim }\int_{\left[ -\varepsilon
^{-1},-\varepsilon \right] \cup \left[ \varepsilon ,\varepsilon ^{-1}\right]
}\frac{f\left( \nu -x\right) }{x}\mathrm{d}x\ ,\qquad \nu \in \mathbb{R}\ .
\end{equation*}%
Here, $f:\mathbb{R}\rightarrow \mathbb{C}$ belongs to the space $\Upsilon $
of functions which are the Fourier transforms of compactly supported and
piece--wise smooth functions $\mathbb{R\rightarrow R}$. Equation (\ref{eq
sup curernt}) corresponds to Ohm's law in Fourier space at microscopic
scales, in accordance with experimental results of \cite{Ohm-exp2,Ohm-exp}.

Moreover, by Corollary \ref{lemma sigma pos type} (v) together with Equation
(\ref{auto-evident}), Theorem \ref{lemma sigma pos type copy(4)} and the
Bolzano--Weierstrass theorem, the family $\{\mu _{\Lambda _{l}}^{(\omega
)}\}_{l\in \mathbb{R}^{+}}$ has weak$^{\ast }$--accumulation points. As a
consequence, the current linear response converges pointwise along a
subsequence to
\begin{equation*}
J_{\infty }^{(\omega ,\mathcal{A})}(t)=\frac{1}{2}\int_{\mathbb{R}}\mathcal{%
\hat{E}}_{\nu }^{(t)}\mu _{\mathbb{R}^{d}}^{(\omega )}\left( \mathrm{d}\nu
\right) \vec{w}+\frac{i}{2}\int_{\mathbb{R}}\mathbb{H}(\mathcal{\hat{E}}%
^{(t)})\left( \nu \right) \mu _{\mathbb{R}^{d}}^{(\omega )}\left( \mathrm{d}%
\nu \right) \vec{w}
\end{equation*}%
with $\mu _{\mathbb{R}^{d}}^{(\omega )}$ being some weak$^{\ast }$%
--accumulation point of $\{\mu _{\Lambda _{l}}^{(\omega )}\}_{l\in \mathbb{R}%
^{+}}$. $\mu _{\mathbb{R}^{d}}^{(\omega )}$\ can be interpreted as a \emph{%
macroscopic conductivity measure} and is under reasonable circumstances
unique. In fact, we give in \cite{OhmIII} a detailed analysis of such limits
by considering random static external potentials.

Observe that $i\mathbb{H}\left( \Upsilon \right) \subset \Upsilon $ and $%
\mathbb{H\circ H}=-1$ on $\Upsilon $. In particular, the two functionals
\begin{eqnarray*}
\mu _{\Lambda _{l}}^{\Vert } &:&\Upsilon \rightarrow \mathbb{R}\ ,\ \mu
_{\Lambda _{l}}^{\Vert }\left( f\right) :=\frac{1}{2}\int_{\mathbb{R}}f(\nu
)\mu _{\Lambda _{l}}^{(\omega )}\left( \mathrm{d}\nu \right) \ , \\
\mu _{\Lambda _{l}}^{\bot } &:&\Upsilon \rightarrow \mathbb{R}\ ,\ \mu
_{\Lambda _{l}}^{\bot }\left( f\right) :=\frac{1}{2}\int_{\mathbb{R}}\mathbb{%
H}\left( f\right) (\nu )\mu _{\Lambda _{l}}^{(\omega )}\left( \mathrm{d}\nu
\right) \ ,
\end{eqnarray*}%
satisfy Kramers--Kronig relations:
\begin{equation}
\mu _{\Lambda _{l}}^{\Vert }\circ \mathbb{H=}\mu _{\Lambda _{l}}^{\bot
}\qquad \text{and}\qquad \mu _{\Lambda _{l}}^{\bot }\circ \mathbb{H=-}\mu
_{\Lambda _{l}}^{\Vert }\ .  \label{Kramers--Kronig relations}
\end{equation}%
Note that, w.r.t. the usual topology of the space $\mathcal{S}\left( \mathbb{%
R};\mathbb{C}\right) $ of Schwartz functions, $\Upsilon \cap \mathcal{S}%
\left( \mathbb{R};\mathbb{C}\right) $ is dense in $\mathcal{S}\left( \mathbb{%
R};\mathbb{C}\right) $ and $\mu _{\Lambda _{l}}^{\Vert },\mu _{\Lambda
_{l}}^{\bot }$ are continuous on $\Upsilon \cap \mathcal{S}\left( \mathbb{R};%
\mathbb{C}\right) $. Hence, each entry of $\mu _{\Lambda _{l}}^{\Vert },\mu
_{\Lambda _{l}}^{\bot }$ w.r.t. the canonical orthonormal basis of $\mathbb{R%
}^{d}$ can be seen as a tempered distribution. Moreover, (\ref{eq sup
curernt}) yields%
\begin{equation}
J_{l}^{(\omega ,\mathcal{A})}(t)=\left( \mu _{\Lambda _{l}}^{\Vert }(%
\mathcal{\hat{E}}^{(t)})+i\mu _{\Lambda _{l}}^{\bot }(\mathcal{\hat{E}}%
^{(t)})\right) \vec{w}\ .  \label{current distribution}
\end{equation}%
Therefore, the $\mathcal{B}(\mathbb{R}^{d})$--valued distribution $\mu
_{\Lambda _{l}}^{\Vert }$ is the linear response in--phase component of the
total conductivity in Fourier space. For this reason, $\mu _{\Lambda
_{l}}^{\Vert }+i\mu _{\Lambda _{l}}^{\bot }$ is named here the (microscopic,
$\mathcal{B}(\mathbb{R}^{d})$--valued) \emph{conductivity distribution} of
the box $\Lambda _{l}$. Similarly, the limit $J_{\infty }^{(\omega ,\mathcal{%
A})}$ obeys (\ref{current distribution}) with $\mu _{\mathbb{R}%
^{d}}^{(\omega )}$ replacing $\mu _{\Lambda _{l}}^{(\omega )}$.

\section{Microscopic Joule's Law\label{Sect Joule}}

\noindent ...\textit{the calorific effects of equal quantities of
transmitted electricity are proportional to the resistances opposed to its
passage, whatever may be the length, thickness, shape, or kind of metal
which closes the circuit : and also that, coeteris paribus, these effects
are in the duplicate ratio of the quantities of transmitted electricity ;
and consequently also in the duplicate ratio of the velocity of
transmission. }\smallskip

\hfill \lbrack Joule, 1840]\bigskip

\noindent In other words, as originally observed \cite{J} by the physicist
J. P. Joule, the heat (per second) produced within an electric circuit is
proportional to the electric resistance and the square of the current.

The aim of this section is to prove such a phenomenology for the fermion
system under consideration. Before studying Joule's effect we need to define
energy observables and increments:

\subsection{Energy Observables}

For any $L\in \mathbb{R}^{+}$, the \emph{internal} energy observable in the
box $\Lambda _{L}$ (\ref{eq:def lambda n}) is defined by%
\begin{equation}
H_{L}^{(\omega ,\lambda )}:=\sum\limits_{x,y\in \Lambda _{L}}\langle
\mathfrak{e}_{x},(\Delta _{\mathrm{d}}+\lambda V_{\omega })\mathfrak{e}%
_{y}\rangle a_{x}^{\ast }a_{y}\in \mathcal{U}\ .  \label{definition de Hl}
\end{equation}%
It is the second quantization of the one--particle operator $\Delta _{%
\mathrm{d}}+\lambda V_{\omega }$ restricted to the subspace $\ell
^{2}(\Lambda _{L})\subset \ell ^{2}(\mathfrak{L})$. When the electromagnetic
field is switched on, i.e., for $t\geq t_{0}$, the (time--dependent) \emph{%
total} energy observable in the box $\Lambda _{L}$ is then equal to $%
H_{L}^{(\omega ,\lambda )}+W_{t}^{\mathbf{A}}$, where, for any $\mathbf{A}%
\in \mathbf{C}_{0}^{\infty }$ and $t\in \mathbb{R}$,
\begin{equation}
W_{t}^{\mathbf{A}}:=\sum\limits_{x,y\in \Lambda _{L}}\langle \mathfrak{e}%
_{x},(\Delta _{\mathrm{d}}^{(\mathbf{A})}-\Delta _{\mathrm{d}})\mathfrak{e}%
_{y}\rangle a_{x}^{\ast }a_{y}\in \mathcal{U}  \label{eq def W}
\end{equation}%
is the electromagnetic\emph{\ potential} energy observable.

We define below four types of energies because we have the two above energy
observables as well as two relevant states, the thermal equilibrium state $%
\varrho ^{(\beta ,\omega ,\lambda )}$ and its time evolution $\rho
_{t}^{(\beta ,\omega ,\lambda ,\mathbf{A})}$.

\subsection{Time--dependent Thermodynamic View Point}

In \cite{OhmI}, we investigate the \emph{heat} production of the
(non--autonomous) $C^{\ast }$--dynamical system $(\mathcal{U},\tau
_{t,s}^{(\omega ,\lambda ,\mathbf{A})})$ for any $\beta \in \mathbb{R}^{+}$,
$\omega \in \Omega $, $\lambda \in \mathbb{R}_{0}^{+}$ and $\mathbf{A}\in
\mathbf{C}_{0}^{\infty }$. We show in \cite[Theorem 3.2]{OhmI} that the
fermion system under consideration obeys the first law of thermodynamics. It
means that the heat production due to the electromagnetic field is equal to
an \emph{internal} energy increment. The latter is directly related to the
family $\{H_{L}^{(\omega ,\lambda )}\}_{L\in \mathbb{R}^{+}}$ of internal
energy observables. We also consider an electromagnetic potential energy
defined from the observable $W_{t}^{\mathbf{A}}$. Hence, we define the
following energy increments:

\begin{itemize}
\item[($\mathbf{Q}$)] The \emph{internal} energy increment $\mathbf{S}%
^{(\omega ,\mathbf{A})}\equiv \mathbf{S}^{(\beta ,\omega ,\lambda ,\mathbf{A}%
)}$ is a map from $\mathbb{R}$ to $\mathbb{R}_{0}^{+}$ defined by%
\begin{equation}
\mathbf{S}^{(\omega ,\mathbf{A})}\left( t\right) :=\lim_{L\rightarrow \infty
}\left\{ \rho _{t}^{(\beta ,\omega ,\lambda ,\mathbf{A})}(H_{L}^{(\omega
,\lambda )})-\varrho ^{(\beta ,\omega ,\lambda )}(H_{L}^{(\omega ,\lambda
)})\right\} \ .  \label{entropic energy increment}
\end{equation}%
It takes positive finite values because of \cite[Theorem 3.2]{OhmI}.

\item[($\mathbf{P}$)] The electromagnetic \emph{potential} energy
(increment) $\mathbf{P}^{(\omega ,\mathbf{A})}\equiv \mathbf{P}^{(\beta
,\omega ,\lambda ,\mathbf{A})}$ is a map from $\mathbb{R}$ to $\mathbb{R}$
defined by%
\begin{equation}
\mathbf{P}^{(\omega ,\mathbf{A})}\left( t\right) :=\rho _{t}^{(\beta ,\omega
,\lambda ,\mathbf{A})}(W_{t}^{\mathbf{A}})=\rho _{t}^{(\beta ,\omega
,\lambda ,\mathbf{A})}(W_{t}^{\mathbf{A}})-\varrho ^{(\beta ,\omega ,\lambda
)}(W_{t_{0}}^{\mathbf{A}})\ .  \label{electro free energy}
\end{equation}
\end{itemize}

In other words, $\mathbf{S}^{(\omega ,\mathbf{A})}$ is the increase of
internal energy of the fermion system\ due to the change of its internal
state, whereas $\mathbf{P}^{(\omega ,\mathbf{A})}$ is the electromagnetic%
\emph{\ }potential energy of the fermion system in the state $\rho
_{t}^{(\beta ,\omega ,\lambda ,\mathbf{A})}$. By \cite[Theorem 3.2]{OhmI}, $%
\mathbf{S}^{(\omega ,\mathbf{A})}$ equals the heat production of the fermion
system. Moreover, by \cite[Eq. (24)]{OhmI}, the increase of \emph{total}
energy of the \emph{infinite} system
\begin{equation}
\lim_{L\rightarrow \infty }\left\{ \rho _{t}^{(\beta ,\omega ,\lambda ,%
\mathbf{A})}(H_{L}^{(\omega ,\lambda )}+W_{t}^{\mathbf{A}})-\varrho ^{(\beta
,\omega ,\lambda )}(H_{L}^{(\omega ,\lambda )})\right\} =\mathbf{S}^{(\omega
,\mathbf{A})}\left( t\right) +\mathbf{P}^{(\omega ,\mathbf{A})}\left(
t\right)  \label{lim_en_incr full}
\end{equation}%
is exactly the work performed by the electromagnetic field at time $t\geq
t_{0}$:%
\begin{equation}
\mathbf{S}^{(\omega ,\mathbf{A})}\left( t\right) +\mathbf{P}^{(\omega ,%
\mathbf{A})}\left( t\right) =\int_{t_{0}}^{t}\rho _{s}^{(\beta ,\omega
,\lambda ,\mathbf{A})}\left( \partial _{s}W_{s}^{\mathbf{A}}\right) \mathrm{d%
}s\ .  \label{work}
\end{equation}

\subsection{Electromagnetic View Point\label{section para-dia-energy}}

In the previous subsection the \emph{total} energy increment is decomposed
into two components (\ref{lim_en_incr full}) that can be identified with
heat production and potential energy. This total energy increment can also
be decomposed in two other components which have interesting features in
terms of currents. Indeed, for any $\beta \in \mathbb{R}^{+}$, $\omega \in
\Omega $, $\lambda \in \mathbb{R}_{0}^{+}$ and $\mathbf{A}\in \mathbf{C}%
_{0}^{\infty }$, we define:

\begin{itemize}
\item[(p)] The \emph{paramagnetic} energy increment $\mathfrak{J}_{\mathrm{p}%
}^{(\omega ,\mathbf{A})}\equiv \mathfrak{I}_{\mathrm{p}}^{(\beta ,\omega
,\lambda ,\mathbf{A})}$ is the map from $\mathbb{R}$ to $\mathbb{R}$ defined
by
\begin{equation}
\mathfrak{I}_{\mathrm{p}}^{(\omega ,\mathbf{A})}\left( t\right)
:=\lim_{L\rightarrow \infty }\left\{ \rho _{t}^{(\beta ,\omega ,\lambda ,%
\mathbf{A})}(H_{L}^{(\omega ,\lambda )}+W_{t}^{\mathbf{A}})-\varrho ^{(\beta
,\omega ,\lambda )}(H_{L}^{(\omega ,\lambda )}+W_{t}^{\mathbf{A}})\right\} \
.  \label{lim_en_incr}
\end{equation}

\item[(d)] The \emph{diamagnetic} energy (increment) $\mathfrak{I}_{\mathrm{d%
}}^{(\omega ,\mathbf{A})}\equiv \mathfrak{I}_{\mathrm{d}}^{(\beta ,\omega
,\lambda ,\mathbf{A})}$ is the map from $\mathbb{R}$ to $\mathbb{R}$ defined
by
\begin{equation}
\mathfrak{I}_{\mathrm{d}}^{(\omega ,\mathbf{A})}\left( t\right) :=\varrho
^{(\beta ,\omega ,\lambda )}(W_{t}^{\mathbf{A}})=\varrho ^{(\beta ,\omega
,\lambda )}(W_{t}^{\mathbf{A}})-\varrho ^{(\beta ,\omega ,\lambda
)}(W_{t_{0}}^{\mathbf{A}})\ .  \label{lim_en_incr dia}
\end{equation}
\end{itemize}

\noindent Note that the limit (\ref{lim_en_incr}) exists at all times
because of (\ref{lim_en_incr full})--(\ref{work}). In particular,
\begin{equation}
\mathfrak{I}_{\mathrm{p}}^{(\omega ,\mathbf{A})}\left( t\right) +\mathfrak{I}%
_{\mathrm{d}}^{(\omega ,\mathbf{A})}\left( t\right) =\int_{t_{0}}^{t}\rho
_{s}^{(\beta ,\omega ,\lambda ,\mathbf{A})}\left( \partial _{s}W_{s}^{%
\mathbf{A}}\right) \mathrm{d}s  \label{work-dia}
\end{equation}%
for any $\beta \in \mathbb{R}^{+}$, $\omega \in \Omega $, $\lambda \in
\mathbb{R}_{0}^{+}$, $\mathbf{A}\in \mathbf{C}_{0}^{\infty }$ and times $%
t\geq t_{0}$.

The term $\mathfrak{J}_{\mathrm{p}}^{(\omega ,\mathbf{A})}$ is the part of
electromagnetic work implying a change of the internal state of the system,
whereas the diamagnetic energy is the raw\emph{\ }electromagnetic energy
given to the system at thermal equilibrium. Indeed, because of the second
law of thermodynamics, in presence of non--zero electromagnetic fields the
system constantly tends to minimize the (instantaneous) free--energy
associated with $H_{L}^{(\omega ,\lambda )}+W_{t}^{\mathbf{A}}$ and it is
thus forced to change its state as time evolves.

We show below that $\mathfrak{J}_{\mathrm{p}}^{(\omega ,\mathbf{A})}$ and $%
\mathfrak{I}_{\mathrm{d}}^{(\omega ,\mathbf{A})}$ \emph{cannot} be
identified with either $\mathbf{P}^{(\omega ,\mathbf{A})}$ or $\mathbf{S}%
^{(\omega ,\mathbf{A})}$ but are directly related to paramagnetic and
diamagnetic currents, respectively.

\subsection{Joule's Effect and Energy Increments\label{Section Joule effect}}

By Theorem \ref{thm Local Ohm's law}, for each $l,\beta \in \mathbb{R}^{+}$,
$\omega \in \Omega $, $\lambda \in \mathbb{R}_{0}^{+}$ and any
electromagnetic potential $\mathbf{A}\in \mathbf{C}_{0}^{\infty }$, the
electric field in its integrated form $\mathbf{E}_{t}^{\eta \mathbf{A}_{l}}$
(cf. (\ref{V bar 0})--(\ref{V bar 0bis}) and (\ref{rescaled vector potential}%
)) implies paramagnetic and diamagnetic currents with linear coefficients
being respectively equal to%
\begin{eqnarray}
J_{\mathrm{p},l}^{(\omega ,\mathbf{A})}(t,\mathbf{x}) &:=&\frac{1}{2}%
\int_{t_{0}}^{t}\underset{\mathbf{y}\in \mathfrak{K}}{\sum }\sigma _{\mathrm{%
p}}^{(\omega )}\left( \mathbf{x},\mathbf{y,}t-s\right) \mathbf{E}_{s}^{%
\mathbf{A}_{l}}(\mathbf{y})\mathrm{d}s\ ,  \label{current para} \\
J_{\mathrm{d},l}^{(\omega ,\mathbf{A})}(t,\mathbf{x})
&:=&\int_{t_{0}}^{t}\sigma _{\mathrm{d}}^{(\omega )}\left( \mathbf{x}\right)
\mathbf{E}_{s}^{\mathbf{A}_{l}}(\mathbf{x})\mathrm{d}s\ ,
\label{current dia}
\end{eqnarray}%
at any bond $\mathbf{x}\in \mathfrak{K}$ (see (\ref{proche voisins0})) and
time $t\geq t_{0}$. Recall that $\sigma _{\mathrm{p}}^{(\omega )}$ and $%
\sigma _{\mathrm{d}}^{(\omega )}$ are the microscopic charge transport
coefficients defined by (\ref{backwards -1bis})--(\ref{backwards -1bispara}).

Provided $|\eta |\ll 1$, the electric work produced at any time $t\geq t_{0}$
by paramagnetic currents is then equal to%
\begin{equation}
\frac{\eta ^{2}}{2}\int\nolimits_{t_{0}}^{t}\underset{\mathbf{x}\in
\mathfrak{K}}{\sum }J_{\mathrm{p},l}^{(\omega ,\mathbf{A})}(s,\mathbf{x})%
\mathbf{E}_{s}^{\mathbf{A}_{l}}(\mathbf{x})\mathrm{d}s\ ,  \label{heuristic1}
\end{equation}%
whereas the diamagnetic work equals
\begin{equation}
\frac{\eta ^{2}}{2}\int\nolimits_{t_{0}}^{t}\underset{\mathbf{x}\in
\mathfrak{K}}{\sum }J_{\mathrm{d},l}^{(\omega ,\mathbf{A})}(s,\mathbf{x})%
\mathbf{E}_{s}^{\mathbf{A}_{l}}(\mathbf{x})\mathrm{d}s = \frac{\eta ^{2}}{4}%
\underset{\mathbf{x}\in \mathfrak{K}}{\sum }J_{\mathrm{d},l}^{(\omega ,%
\mathbf{A})}(t,\mathbf{x})\int\nolimits_{t_{0}}^{t} \mathbf{E}_{s}^{\mathbf{A%
}_{l}}(\mathbf{x})\mathrm{d}s \ .  \label{heuristic0}
\end{equation}%
Remark that the factor $\eta^2/2$ (instead of $\eta^2$) in (\ref{heuristic1}%
)--(\ref{heuristic0}) is due to the fact that $\mathfrak{K}$ is a set of
\emph{oriented} bonds and thus each bond is counted twice.

As explained in Section \ref{Sect para dia current}, there exist also \emph{%
thermal} currents
\begin{equation}
\varrho ^{(\beta ,\omega ,\lambda )}(I_{\mathbf{x}})\ ,\qquad \mathbf{x}\in
\mathfrak{K}\ ,  \label{current thermal}
\end{equation}%
coming from the inhomogeneity of the fermion system for $\lambda \in \mathbb{%
R}^{+}$. Thermal currents imply an additional raw electromagnetic work
\begin{equation}
-\frac{\eta }{2}\underset{\mathbf{x}\in \mathfrak{K}}{\sum }\varrho ^{(\beta
,\omega ,\lambda )}(I_{\mathbf{x}})\int\nolimits_{t_{0}}^{t}\mathbf{E}_{s}^{%
\mathbf{A}_{l}}(\mathbf{x})\mathrm{d}s  \label{heuristic2}
\end{equation}%
at any time $t\geq t_{0}$.

Since $\mathbf{A}$ is by assumption compactly supported in time, the
corresponding electric field satisfies the \emph{AC--condition}%
\begin{equation}
\int\nolimits_{t_{0}}^{t}E_{\mathbf{A}}(s,x)\mathrm{d}s=0\ ,\quad x\in
\mathbb{R}^{d}\ ,  \label{zero mean field}
\end{equation}%
for times $t\geq t_{1}\geq t_{0}$. Here,%
\begin{equation*}
t_{1}:=\min \left\{ t\geq t_{0}:\quad \int\nolimits_{t_{0}}^{t^{\prime }}E_{%
\mathbf{A}}(s,x)\mathrm{d}s=0\quad \text{for all }x\in \mathbb{R}^{d}\text{
and }t^{\prime }\geq t\right\}
\end{equation*}%
is the time at which the electric field is definitively turned off. In this
case, the electric works (\ref{heuristic0}) and (\ref{heuristic2}) vanish
for $t\geq t_{1}$ and (\ref{heuristic1}) stays constant. Following Joule's
effect, for $t\geq t_{1}$, this energy should correspond to a \emph{heat
production} as defined in \cite[Definition 3.1]{OhmI}. The latter equals the
energy increment $\mathbf{S}^{(\omega ,\eta \mathbf{A}_{l})}$, by \cite[%
Theorem 3.2]{OhmI}.

We prove this heuristics in Section \ref{Local AC--Ohm's Law} and obtain the
following theorem:

\begin{satz}[Microscopic Joule's law -- I]
\label{Local Ohm's law thm copy(2)}\mbox{
}\newline
For any $\mathbf{A}\in \mathbf{C}_{0}^{\infty }$, there is $\eta _{0}\in
\mathbb{R}^{+}$ such that, for all $|\eta |\in (0,\eta _{0}]$, $l,\beta \in
\mathbb{R}^{+}$, $\omega \in \Omega $, $\lambda \in \mathbb{R}_{0}^{+}$ and $%
t\geq t_{0}$, the following assertions hold true:\newline
\emph{(p)} Paramagnetic energy increment:%
\begin{equation*}
\mathfrak{I}_{\mathrm{p}}^{(\omega ,\eta \mathbf{A}_{l})}\left( t\right) =%
\frac{\eta ^{2}}{2}\int\nolimits_{t_{0}}^{t}\underset{\mathbf{x}\in
\mathfrak{K}}{\sum }J_{\mathrm{p},l}^{(\omega ,\mathbf{A})}(s,\mathbf{x})%
\mathbf{E}_{s}^{\mathbf{A}_{l}}(\mathbf{x})\mathrm{d}s+\mathcal{O}(\eta
^{3}l^{d})\ .
\end{equation*}%
\emph{(d)} Diamagnetic energy:
\begin{eqnarray*}
\mathfrak{I}_{\mathrm{d}}^{(\omega ,\eta \mathbf{A}_{l})}\left( t\right) &=&-%
\frac{\eta }{2}\underset{\mathbf{x}\in \mathfrak{K}}{\sum }\varrho ^{(\beta
,\omega ,\lambda )}(I_{\mathbf{x}})\left( \int\nolimits_{t_{0}}^{t}\mathbf{E}%
_{s}^{\mathbf{A}_{l}}(\mathbf{x})\mathrm{d}s\right) \\
&&+\frac{\eta ^{2}}{4}\underset{\mathbf{x}\in \mathfrak{K}}{\sum }J_{\mathrm{%
d},l}^{(\omega ,\mathbf{A})}(t,\mathbf{x})\int\nolimits_{t_{0}}^{t}\mathbf{E}%
_{s}^{\mathbf{A}_{l}}(\mathbf{x})\mathrm{d}s+\mathcal{O}(\eta ^{3}l^{d})\ .
\end{eqnarray*}%
\emph{(\textbf{Q})} Heat production -- Internal energy increment:%
\begin{eqnarray*}
\mathbf{S}^{(\omega ,\eta \mathbf{A}_{l})}\left( t\right) &=&-\frac{\eta ^{2}%
}{2}\underset{\mathbf{x}\in \mathfrak{K}}{\sum }J_{\mathrm{p},l}^{(\omega ,%
\mathbf{A})}(t,\mathbf{x})\left( \int\nolimits_{t_{0}}^{t}\mathbf{E}_{s}^{%
\mathbf{A}_{l}}(\mathbf{x})\mathrm{d}s\right) \\
&&+\mathfrak{I}_{\mathrm{p}}^{(\omega ,\eta \mathbf{A}_{l})}\left( t\right) +%
\mathcal{O}(\eta ^{3}l^{d})
\end{eqnarray*}%
\emph{(\textbf{P})} Electromagnetic potential energy:
\begin{eqnarray*}
\mathbf{P}^{(\omega ,\eta \mathbf{A}_{l})}\left( t\right) &=&\frac{\eta ^{2}%
}{2}\underset{\mathbf{x}\in \mathfrak{K}}{\sum }J_{\mathrm{p},l}^{(\omega ,%
\mathbf{A})}(t,\mathbf{x})\left( \int\nolimits_{t_{0}}^{t}\mathbf{E}_{s}^{%
\mathbf{A}_{l}}(\mathbf{x})\mathrm{d}s\right) \\
&&+\mathfrak{I}_{\mathrm{d}}^{(\omega ,\eta \mathbf{A}_{l})}\left( t\right) +%
\mathcal{O}(\eta ^{3}l^{d})\ .
\end{eqnarray*}%
The correction terms of order $\mathcal{O}(\eta ^{3}l^{d})$ in assertions
(p), (d), (\textbf{Q}) and (\textbf{P}) are uniformly bounded in $\beta \in
\mathbb{R}^{+}$, $\omega \in \Omega $, $\lambda \in \mathbb{R}_{0}^{+}$ and $%
t\geq t_{0}$.
\end{satz}

\begin{proof}
The first two assertions are\ Theorem \ref{Local Ohm's law thm}, whereas (%
\textbf{Q}) and (\textbf{P}) are direct consequences of (\ref{entropic
energy increment})--(\ref{electro free energy}), (\ref{lim_en_incr})--(\ref%
{lim_en_incr dia}), Theorem \ref{Local Ohm's law thm} and Lemma \ref{bound
incr 1 Lemma copy(5)}.
\end{proof}

\noindent We emphasize the fact that the asymptotics obtained are uniform
w.r.t. $l,\beta \in \mathbb{R}^{+}$, $\omega \in \Omega $, $\lambda \in
\mathbb{R}_{0}^{+}$ and $t\geq t_{0}$. This is a crucial property to get
macroscopic Joule's law when $l\rightarrow \infty $. See \cite{OhmIII}.

\begin{bemerkung}[Total energy]
\mbox{
}\newline
One can easily deduce from Lemma \ref{bound incr 1 Lemma copy(9)} the
asymptotics of the total work performed by the electric field, which is
equal to
\begin{equation*}
\int_{t_{0}}^{t}\rho _{s}^{(\beta ,\omega ,\lambda ,\mathbf{A})}\left(
\partial _{s}W_{s}^{\mathbf{A}}\right) \mathrm{d}s\ ,
\end{equation*}%
similar to what is done in Theorem \ref{Local Ohm's law thm copy(2)}.
\end{bemerkung}

Theorem \ref{Local Ohm's law thm copy(2)} describes, among other things, how
resistance in the fermion system converts electric energy into heat. Indeed,
by \cite[Theorem 3.2]{OhmI}, for any $\mathbf{A}\in \mathbf{C}_{0}^{\infty }$%
, there is $\eta _{0}\in \mathbb{R}^{+}$ such that, for all $|\eta |\in
(0,\eta _{0}]$, $l,\beta \in \mathbb{R}^{+}$, $\omega \in \Omega $, $\lambda
\in \mathbb{R}_{0}^{+}$ and $t\geq t_{0}$,%
\begin{equation*}
\frac{\eta ^{2}}{2}\int\nolimits_{t_{0}}^{t}\underset{\mathbf{x}\in
\mathfrak{K}}{\sum }J_{\mathrm{p},l}^{(\omega ,\mathbf{A})}(s,\mathbf{x})%
\mathbf{E}_{s}^{\mathbf{A}_{l}}(\mathbf{x})\mathrm{d}s-\frac{\eta ^{2}}{2}%
\underset{\mathbf{x}\in \mathfrak{K}}{\sum }J_{\mathrm{p},l}^{(\omega ,%
\mathbf{A})}(t,\mathbf{x})\left( \int\nolimits_{t_{0}}^{t}\mathbf{E}_{s}^{%
\mathbf{A}_{l}}(\mathbf{x})\mathrm{d}s\right) \geq \mathcal{O}(\eta
^{3}l^{d})\ .
\end{equation*}%
The latter is the positivity of the heat production, i.e., $\mathbf{S}%
^{(\omega ,\eta \mathbf{A}_{l})}\left( t\right) \in \mathbb{R}_{0}^{+}$,
which for times $t\geq t_{1}\geq t_{0}$ equals, at leading order, the work
of paramagnetic currents (\ref{heuristic1}), that is,
\begin{equation}
\frac{\eta ^{2}}{4}\int\nolimits_{t_{0}}^{t}\mathrm{d}s_{1}%
\int_{t_{0}}^{s_{1}}\mathrm{d}s_{2}\underset{\mathbf{x},\mathbf{y}\in
\mathfrak{K}}{\sum }\sigma _{\mathrm{p}}^{(\omega )}\left( \mathbf{x},%
\mathbf{y,}s_{1}-s_{2}\right) \mathbf{E}_{s_{2}}^{\mathbf{A}_{l}}(\mathbf{x})%
\mathbf{E}_{s_{1}}^{\mathbf{A}_{l}}(\mathbf{y})\geq \mathcal{O}(\eta
^{3}l^{d})\ .  \label{Joule law micoro ac}
\end{equation}%
This is nothing but Joule's law expressed w.r.t. electric fields and
conductivity (instead of currents and resistance). Indeed, Joule's law
in its original form describes a quadratic relation between heat production
and currents. The last result gives a quadratic relation between heat
production and electric fields, instead (see also (\ref{Joule law micoro ac
bis}) and (\ref{barack s1})). Joule's law for currents follows from its
version for electric fields above, by taking the Legendre--Fenchel
transform. For more details, see Section \ref{Resistivity and Joule's Law
sect}.

In fact, for any space--homogeneous electric field $\mathcal{E}\in
C_{0}^{\infty }\left( \mathbb{R};\mathbb{R}\right) $ in the box $\Lambda
_{l} $ for $l\in \mathbb{R}^{+}$ (as described at the beginning of Section %
\ref{Sect Local Ohm law}), the left hand side of Equation (\ref{Joule law
micoro ac}) can be rewritten by using (\ref{average conductivity}) and
Theorem \ref{lemma sigma pos type copy(4)} as%
\begin{eqnarray}
\lefteqn{ \eta^2\left\vert \Lambda _{l}\right\vert \int\nolimits_{t_{0}}^{t}%
\mathrm{d}s_{1}\int\nolimits_{t_{0}}^{s_{1}}\mathrm{d}s_{2}\langle \vec{w}%
,\Xi _{\mathrm{p},l}^{(\omega )}(s_{1}-s_{2})\vec{w}\rangle \mathcal{E}%
_{s_{2}}\mathcal{E}_{s_{1}} } &&  \notag \\
&=&\frac{\eta^2 \left\vert \Lambda _{l}\right\vert }{2}\int_{\mathbb{R}}|%
\mathcal{\hat{E}}_{\nu }|^2\ \langle \vec{w},\mu _{\mathrm{p},l}^{(\omega )}(%
\mathrm{d}\nu )\vec{w}\rangle \geq 0  \label{Joule law micoro ac bis}
\end{eqnarray}%
for all $t\geq t_{1}$, with $\mathcal{\hat{E}}_{\nu }$ being the Fourier
transform of $\mathcal{E}_{t}$. In particular,
\begin{equation*}
\frac{\eta^2}{2}|\mathcal{\hat{E}}_{\nu }|^{2}\ \langle \vec{w},\mu _{%
\mathrm{p},l}^{(\omega )}(\mathrm{d}\nu )\vec{w}\rangle
\end{equation*}%
is, at leading order, the heat production per unit volume due to the
component of frequency $\nu $ of the electric field, in accordance with
Joule's law in the AC--regime.

In presence of electromagnetic fields, i.e., at times $t\in \left[
t_{0},t_{1}\right] $ for which the AC--condition (\ref{zero mean field})
does not hold, the situation is more complex. Indeed, at these times, $%
\mathfrak{J}_{\mathrm{p}}^{(\omega ,\mathbf{A})}$ and $\mathfrak{I}_{\mathrm{%
d}}^{(\omega ,\mathbf{A})}$ cannot be identified with either $\mathbf{P}%
^{(\omega ,\mathbf{A})}$ or $\mathbf{S}^{(\omega ,\mathbf{A})}$. From
Theorem \ref{Local Ohm's law thm copy(2)}\thinspace (p), the energy $%
\mathfrak{J}_{\mathrm{p}}^{(\omega ,\mathbf{A})}$ is generated by
paramagnetic\emph{\ }currents, see (\ref{current para}). By contrast, the
raw electromagnetic energy $\mathfrak{I}_{\mathrm{d}}^{(\omega ,\mathbf{A})}$
is carried by diamagnetic and thermal currents, see (\ref{current dia}) and (%
\ref{current thermal}) and compare Theorem \ref{Local Ohm's law thm copy(2)}%
\thinspace (d) with (\ref{heuristic0}) and (\ref{heuristic2}). These
currents are physically different: Diamagnetic currents correspond to the
raw ballistic flow of charged particles due to the electric field, whereas
only paramagnetic currents \emph{partially} participates to the heat
production $\mathbf{S}^{(\omega ,\mathbf{A})}$, a portion of paramagnetic
currents being also responsible for the modification of the electromagnetic
potential energy:

\begin{itemize}
\item Part of the electric work performed by paramagnetic currents
participates to the electromagnetic potential energy as explained in Theorem %
\ref{Local Ohm's law thm copy(2)}\thinspace (\textbf{P}). The same
phenomenon appears for thermal currents defined by (\ref{current thermal}).
Indeed, observe that any current $J(t,\mathbf{x})$ on the bound $\mathbf{x}$
at time $t$ yields a contribution
\begin{equation*}
J(t,\mathbf{x})\left( \int\nolimits_{t_{0}}^{t}\mathbf{E}_{s}^{\mathbf{A}%
_{l}}(\mathbf{x})\mathrm{d}s\right)
\end{equation*}%
to the electromagnetic potential energy. Compare (\ref{heuristic2}) and $%
\mathbf{P}^{(\omega ,\eta \mathbf{A}_{l})}-\mathfrak{I}_{\mathrm{d}%
}^{(\omega ,\eta \mathbf{A}_{l})}$ via Theorem \ref{Local Ohm's law thm
copy(2)}\thinspace (\textbf{P}) . This potential energy disappears as soon
as the electromagnetic potential is switched off.

\item Then, the remaining energy coming from the whole paramagnetic energy $%
\mathfrak{I}_{\mathrm{p}}^{(\omega ,\eta \mathbf{A}_{l})}$ is a heat energy
or quantity of heat, by Theorem \ref{Local Ohm's law thm copy(2)}\thinspace (%
\textbf{Q}) and \cite[Theorem 3.2]{OhmI}. It survives even after turning off
the electromagnetic potential.
\end{itemize}

\subsection{Resistivity and Joule's Law\label{Resistivity and Joule's Law
sect}}

Joule's observation in \cite{J} associates heat production in electric
circuits with currents and resistance, rather than electric fields and
conductivity. We thus explain in this subsection how to get such a relation
between heat production and currents from (\ref{Joule law micoro ac})--(\ref%
{Joule law micoro ac bis}), which express the total heat production as a
function of electric fields and conductivity.

Note that the concept of \emph{resistivity} is less natural as the one
of \emph{conductivity}. Indeed, the current is an effect of the imposed
electric field (and not the other way around). Moreover, from the
mathematical point of view, the resistivity is a kind of inverse of the
conductivity, which is a measure, as shown above. See Theorem \ref{lemma sigma
pos type copy(4)}. To give a precise mathematical meaning to such an inverse
of the conductivity measure we use the following observation: Take the
function $\mathrm{q}:e\mapsto ae^{2}/2$ from $\mathbb{R}$ to $\mathbb{R}$
with $a>0$. Its Legendre--Fenchel transform is the function $\mathrm{q}%
^{\ast }$ from $\mathbb{R}$ to $\mathbb{R}$ defined by%
\begin{equation*}
\mathrm{q}^{\ast }\left( j\right) :=\underset{e\in \mathbb{R}}{\sup }\left\{
je-\mathrm{q}\left( e\right) \right\} =je_{j}-\mathrm{q}\left( e_{j}\right) =%
\frac{j^{2}}{2a}\ ,\qquad j\in \mathbb{R}\ .
\end{equation*}%
Similarly,
\begin{equation*}
\mathrm{q}\left( e\right) :=\underset{j\in \mathbb{R}}{\sup }\left\{ ej-%
\mathrm{q}^{\ast }\left( j\right) \right\} =ej_{e}-\mathrm{q}^{\ast }\left(
j_{e}\right) \ ,\qquad e\in \mathbb{R}\ .
\end{equation*}%
Their derivatives are respectively equal to%
\begin{equation*}
\partial _{e}\mathrm{q}\left( e\right) =ae=j_{e}\text{\qquad and\qquad }%
\partial _{j}\mathrm{q}^{\ast }\left( j\right) =\frac{j}{a}=e_{j}\ .
\end{equation*}%
Hence, for any $j,e\in \mathbb{R}$,
\begin{equation}
\partial _{e}\mathrm{q}\left( \partial _{j}\mathrm{q}^{\ast }\left( j\right)
\right) =\partial _{e}\mathrm{q}\left( e_{j}\right) =j_{e_{j}}=j\ ,\text{%
\quad }\partial _{j}\mathrm{q}^{\ast }\left( \partial _{e}\mathrm{q}\left(
e\right) \right) =\partial _{j}\mathrm{q}^{\ast }\left( j_{e}\right)
=e_{j_{e}}=e\ .  \label{long eq}
\end{equation}%
In our construction below, $j$ corresponds to a current $\mathcal{J}$,
whereas $e$ refers to an electric field $\mathcal{E}$. Thus, the derivative $%
\partial _{e}\mathrm{q}\left( e\right) $ can be seen as a function that maps
each electric field $e$ in the current $j_{e}$ produced by it, i.e., $%
\partial _{e}\mathrm{q}$ is the conductivity (map) of the system. By (\ref%
{long eq}), $\partial _{j}\mathrm{q}^{\ast }$ gives thus the corresponding
resistivity (map).

Below, the function $\mathrm{q}$ is replaced by the heat production $%
\mathbf{Q}_{\Lambda _{l}}$, which is a quadratic functional of the electric
field $\mathcal{E}$, see (\ref{barack s0})--(\ref{barack s1}). The
derivatives $\partial _{e}\mathrm{q}$ and $\partial _{j}\mathrm{q}^{\ast }$
define usual functions. In the case of the Legendre--Fenchel transform $%
\mathbf{Q}_{\Lambda _{l}}^{\ast }$ of $\mathbf{Q}_{\Lambda _{l}}$, we do not
have usual derivatives, but only subdifferentials $\partial \mathbf{Q}%
_{\Lambda _{l}}^{\ast }$. Hence, in general, the resistivity is a
set--valued map (i.e., a multifunction), see (\ref{definition resistivity}).
This makes the mathematical statement of Joule's law in its original
formulation more abstract, see Theorem \ref{Local Ohm's law thm copy(1)}.

For the sake of simplicity, we restrict our analysis to space--homogeneous
electric fields $\mathcal{E}_{t}\vec{w}$ in the box $\Lambda _{l}$ for $l\in
\mathbb{R}^{+}$, as described at the beginning of Section \ref{Sect Local
Ohm law}. Here, $\mathcal{E}\in C_{0}^{\infty }\left( \mathbb{R};\mathbb{R}%
\right) $ and $\vec{w}:=(w_{1},\ldots ,w_{d})\in \mathbb{R}^{d}$. In this
subsection, we fix $l,\beta \in \mathbb{R}^{+}$, $\omega \in \Omega $, $%
\lambda \in \mathbb{R}_{0}^{+}$. Now, we are position to perform the
construction heuristically presented above.

By Corollary \ref{lemma sigma pos type} (i), observe that, for times $t\geq
t_{1}\geq t_{0}$,
\begin{eqnarray*}
&&\int\nolimits_{t_{0}}^{t}\mathrm{d}s_{1}\int\nolimits_{t_{0}}^{s_{1}}%
\mathrm{d}s_{2}\langle \vec{w},\Xi _{\mathrm{p},l}^{(\omega )}(s_{1}-s_{2})%
\vec{w}\rangle \mathcal{E}_{s_{2}}\mathcal{E}_{s_{1}} \\
&=&\frac{1}{2}\int\nolimits_{\mathbb{R}}\mathrm{d}s_{1}\int\nolimits_{%
\mathbb{R}}\mathrm{d}s_{2}\langle \vec{w},\Xi _{\mathrm{p},l}^{(\omega
)}(s_{1}-s_{2})\vec{w}\rangle \mathcal{E}_{s_{2}}\mathcal{E}_{s_{1}}\mathrm{d%
}s_{2}\mathrm{d}s_{1}\ .
\end{eqnarray*}%
Therefore, we define the subspace
\begin{equation*}
\mathcal{S}_{0}:=\left\{ \mathcal{E}\in \mathcal{S}\left( \mathbb{R};\mathbb{%
R}\right) :\int\nolimits_{\mathbb{R}}\mathcal{E}_{s}\mathrm{d}s=0\right\}
\end{equation*}%
of $\mathbb{R}$--valued Schwartz functions satisfying the AC--condition as
well as the functional $\mathbf{Q}_{\Lambda _{l}}\equiv \mathbf{Q}_{\Lambda
_{l}}^{(\beta ,\omega ,\lambda )}$ on $\mathcal{S}_{0}$, the \emph{total}
heat production per unit of volume, by
\begin{equation}
\mathbf{Q}_{\Lambda _{l}}\left( \mathcal{E}\right) :=\frac{1}{2}%
\int\nolimits_{\mathbb{R}}\mathrm{d}s_{1}\int\nolimits_{\mathbb{R}}\mathrm{d}%
s_{2}\langle \vec{w},\Xi _{\mathrm{p},l}^{(\omega )}(s_{1}-s_{2})\vec{w}%
\rangle \mathcal{E}_{s_{2}}\mathcal{E}_{s_{1}}\mathrm{d}s_{2}\mathrm{d}%
s_{1}\ ,\quad \mathcal{E}\in \mathcal{S}_{0}\ .  \label{barack s0}
\end{equation}%
It is a finite, positive quadratic form on $\mathcal{S}_{0}$. Indeed, by
Theorem \ref{lemma sigma pos type copy(4)},%
\begin{equation}
\mathbf{Q}_{\Lambda _{l}}\left( \mathcal{E}\right) =\frac{1}{2}\int_{\mathbb{%
R}}|\mathcal{\hat{E}}_{\nu }|^{2}\ \langle \vec{w},\mu _{\mathrm{p}%
,l}^{(\omega )}(\mathrm{d}\nu )\vec{w}\rangle \in \mathbb{R}_{0}^{+}\ ,\quad
\mathcal{E}\in \mathcal{S}_{0}\ ,  \label{barack s1}
\end{equation}%
and $\langle \vec{w},\mu _{\mathrm{p},l}^{(\omega )}\vec{w}\rangle $ is a
positive measure. It thus defines a semi--norm $\left\Vert \cdot \right\Vert
_{\Lambda _{l}}\equiv \left\Vert \cdot \right\Vert _{\Lambda _{l}}^{(\beta
,\omega ,\lambda )}$ on $\mathcal{S}_{0}$ by
\begin{equation}
\left\Vert \mathcal{E}\right\Vert _{\Lambda _{l}}:=\sqrt{\mathbf{Q}_{\Lambda
_{l}}\left( \mathcal{E}\right) }\ ,\qquad \mathcal{E}\in \mathcal{S}_{0}\ .
\label{barack s111}
\end{equation}

Note that $\mathcal{S}_{0}$ is a closed subspace of the locally convex (Fr%
\'{e}chet) space $\mathcal{S}\left( \mathbb{R};\mathbb{R}\right) $. Let $%
\mathcal{S}_{0}^{\ast }$ be the dual space of $\mathcal{S}_{0}$, i.e., the
set of all continuous linear functionals on $\mathcal{S}_{0}$. $\mathcal{S}%
_{0}^{\ast }$ is equipped with the weak$^{\ast }$--topology. By the
Hahn--Banach theorem, the elements of the dual $\mathcal{S}_{0}^{\ast }$ are
restrictions to $\mathcal{S}_{0}$ of tempered distributions. $\mathcal{S}%
_{0}^{\ast }$ is in fact a space of in--phase AC--currents.

Let $\partial \mathbf{Q}_{\Lambda _{l}}\left( \mathcal{E}\right) \subset
\mathcal{S}_{0}^{\ast }$ be the subdifferential of $\mathbf{Q}_{\Lambda
_{l}} $ at the point $\mathcal{E}\in \mathcal{S}_{0}$. The multifunction $%
\mathbf{\sigma }_{\Lambda _{l}}\equiv \mathbf{\sigma }_{\Lambda
_{l}}^{(\beta ,\omega ,\lambda )}$ from $\mathcal{S}_{0}$ to $\mathcal{S}%
_{0}^{\ast }$ (i.e., the set--valued map from $\mathcal{S}_{0}$ to $2^{%
\mathcal{S}_{0}^{\ast }}$) is defined by
\begin{equation*}
\mathcal{E}\mapsto \mathbf{\sigma }_{\Lambda _{l}}\left( \mathcal{E}\right) =%
\frac{1}{2}\partial \mathbf{Q}_{\Lambda _{l}}\left( \mathcal{E}\right) \ .
\end{equation*}%
It is single--valued with domain $\mathrm{Dom}(\mathbf{\sigma }_{\Lambda
_{l}})=\mathcal{S}_{0}$:

\begin{lemma}[Properties of the AC--conductivity]
\label{Lemma LR dia copy(3)}\mbox{
}\newline
The multifunction $\mathbf{\sigma }_{\Lambda _{l}}$ has domain
\begin{equation*}
\mathrm{Dom}(\mathbf{\sigma }_{\Lambda _{l}}):=\left\{ \mathcal{E}\in
\mathcal{S}_{0}:\partial \mathbf{Q}_{\Lambda _{l}}\left( \mathcal{E}\right)
\neq \emptyset \right\} =\mathcal{S}_{0}
\end{equation*}%
and, for all $\mathcal{E}\in \mathcal{S}_{0}$, $\mathbf{\sigma }_{\Lambda
_{l}}\left( \mathcal{E}\right) =\{\mathcal{J}_{\mathcal{E}}\}$ with%
\begin{equation}
\left\langle \mathcal{J}_{\mathcal{E}},\mathcal{\tilde{E}}\right\rangle =%
\frac{1}{2}\int\nolimits_{\mathbb{R}}\int\nolimits_{\mathbb{R}}\langle \vec{w%
},\Xi _{\mathrm{p},l}^{(\omega )}(s_{1}-s_{2})\vec{w}\rangle \mathcal{\tilde{%
E}}_{s_{1}}\mathcal{E}_{s_{2}}\mathrm{d}s_{2}\mathrm{d}s_{1}\ ,\quad
\mathcal{\tilde{E}}\in \mathcal{S}_{0}\ .  \label{barack s2}
\end{equation}%
[We use the standard notation for distributions: $\langle \mathcal{J}_{%
\mathcal{E}},\mathcal{\tilde{E}}\rangle \equiv \mathcal{J}_{\mathcal{E}}(%
\mathcal{\tilde{E}})$.]
\end{lemma}

\begin{proof}
We prove that, for all $\mathcal{E}\in \mathcal{S}_{0}$, $2\mathcal{J}_{%
\mathcal{E}}$ is the unique tangent functional of $\mathbf{Q}_{\Lambda _{l}}$
at the point $\mathcal{E}$. Indeed,
\begin{equation}
\mathbf{Q}_{\Lambda _{l}}\left( \mathcal{E}+\mathcal{E}_{1}\right) -\mathbf{Q%
}_{\Lambda _{l}}\left( \mathcal{E}\right) =2\left\langle \mathcal{J}_{%
\mathcal{E}},\mathcal{E}_{1}\right\rangle +\mathbf{Q}_{\Lambda _{l}}\left(
\mathcal{E}_{1}\right)  \label{barack s3}
\end{equation}%
for all $\mathcal{E}_{1}\in \mathcal{S}_{0}$. Since $\mathbf{Q}_{\Lambda
_{l}}\left( \mathcal{E}_{1}\right) \geq 0$, the functional $2\mathcal{J}_{%
\mathcal{E}}$ is tangent to $\mathbf{Q}_{\Lambda _{l}}$ at $\mathcal{E}\in
\mathcal{S}_{0}$. In particular, $\mathrm{Dom}(\mathbf{\sigma }_{\Lambda
_{l}})=\mathcal{S}_{0}$. The uniqueness of the tangent functional follows
from the fact that $2\mathcal{J}_{\mathcal{E}}$ is the G\^{a}teaux
derivative of $\mathbf{Q}_{\Lambda _{l}}$ at $\mathcal{E}\in \mathcal{S}_{0}$%
. To see this, replace $\mathcal{E}_{1}$ with $\epsilon \mathcal{E}_{1}$ in (%
\ref{barack s3}) and take the limit $\epsilon \rightarrow 0$.
\end{proof}

\noindent Equation (\ref{barack s2}) is directly related to Ohm's law in
Fourier space. For this reason, $\mathbf{\sigma }_{\Lambda _{l}}$ is named
here the \emph{AC--conductivity} of the region $\Lambda _{l}$.

By Ohm and Joule's laws, a more resistive system produces less heat at fixed
electric field. We thus define a \emph{AC--resistivity order} from the total
heat production $\mathbf{Q}_{\Lambda _{l}}\equiv \mathbf{Q}_{\Lambda
_{l}}^{(\beta ,\omega ,\lambda )}$ (per unit of volume) on the space $%
\mathcal{S}_{0}$ of electric fields:

\begin{definition}[AC--Resistivity order]
\label{def barck def}\mbox{ }\newline
For all $l\in \mathbb{R}^{+}$, we define the partial order relation $\prec $
for the system parameters $\left( \beta ,\omega ,\lambda \right) \in \mathbb{%
R}^{+}\times \Omega \times \mathbb{R}_{0}^{+}$ by%
\begin{equation*}
\left( \beta _{1},\omega _{1},\lambda _{1}\right) \prec \left( \beta
_{2},\omega _{2},\lambda _{2}\right) \text{\qquad iff\qquad }\mathbf{Q}%
_{\Lambda _{l}}^{\left( \beta _{1},\omega _{1},\lambda _{1}\right) }\geq
\mathbf{Q}_{\Lambda _{l}}^{\left( \beta _{2},\omega _{2},\lambda _{2}\right)
}\ .
\end{equation*}
\end{definition}

\noindent This definition is reminiscent of the approach of \cite%
{lieb-yngvason} to the entropy. Observe also that%
\begin{equation*}
\left( \beta _{1},\omega _{1},\lambda _{1}\right) \prec \left( \beta
_{2},\omega _{2},\lambda _{2}\right) \text{\qquad iff\qquad }\mu _{\mathrm{p}%
,l}^{\left( \beta _{1},\omega _{1},\lambda _{1}\right) }|_{\mathbb{R}%
\backslash \left\{ 0\right\} }\geq \mu _{\mathrm{p},l}^{\left( \beta
_{2},\omega _{2},\lambda _{2}\right) }|_{\mathbb{R}\backslash \left\{
0\right\} }\ .
\end{equation*}%
Furthermore, this partial order can be rewritten in terms of a quadratic
function of currents, in accordance with Joule's law in its original form.

To see this, observe that $(\mathcal{S}_{0},\mathcal{S}_{0}^{\ast })$ is a
dual pair, by \cite[Theorem 3.10]{Rudin}. Therefore, $\mathbf{Q}_{\Lambda
_{l}}:\mathcal{S}_{0}\rightarrow \lbrack 0,\infty )$ has a well--defined
Legendre--Fenchel transform $\mathbf{Q}_{\Lambda _{l}}^{\ast }\equiv (%
\mathbf{Q}_{\Lambda _{l}}^{(\beta ,\omega ,\lambda )})^{\ast }$ which is the
convex lower semi--continuous functional from $\mathcal{S}_{0}^{\ast }$ to $%
\left( -\infty ,\infty \right] $ defined in our setting by
\begin{equation}
\mathbf{Q}_{\Lambda _{l}}^{\ast }\left( \mathcal{J}\right) :=2\underset{%
\mathcal{E}\in \mathcal{S}_{0}}{\sup }\left\{ \left\langle \mathcal{J},%
\mathcal{E}\right\rangle -\frac{1}{2}\mathbf{Q}_{\Lambda _{l}}\left(
\mathcal{E}\right) \right\} \ ,\qquad \mathcal{J}\in \mathcal{S}_{0}^{\ast
}\ .  \label{bliblib}
\end{equation}%
The square root of $\mathbf{Q}_{\Lambda _{l}}^{\ast }\left( \mathcal{J}%
\right) $ can be seen as the norm of the linear map $\mathcal{J}:(\mathcal{S}%
_{0},\left\Vert \cdot \right\Vert _{\Lambda _{l}})\rightarrow \mathbb{R}$:

\begin{lemma}[$\mathbf{Q}_{\Lambda _{l}}^{\ast }$ as a semi--norm on $%
\mathcal{S}_{0}^{\ast }$]
\label{Lemma LR dia copy(4)}\mbox{
}\newline
Assume that $\mathbf{Q}_{\Lambda _{l}}$ is not identically zero. Then,
\begin{equation*}
\mathbf{Q}_{\Lambda _{l}}^{\ast }\left( \mathcal{J}\right) =\left( \sup
\left\{ \left\vert \left\langle \mathcal{J},\mathcal{E}\right\rangle
\right\vert :\mathcal{E}\in \mathcal{S}_{0},\ \left\Vert \mathcal{E}%
\right\Vert _{\Lambda _{l}}=1\right\} \right) ^{2}.
\end{equation*}%
If $\mathbf{Q}_{\Lambda _{l}}$ is identically zero, $\mathbf{Q}_{\Lambda
_{l}}^{\ast }\left( \mathcal{J}\right) =\infty $ for all $\mathcal{J}\in
\mathcal{S}_{0}^{\ast }\backslash \{0\}$ and $\mathbf{Q}_{\Lambda
_{l}}^{\ast }\left( 0\right) =0$.
\end{lemma}

\begin{proof}
The assertion for $\mathbf{Q}_{\Lambda _{l}}\equiv 0$ is a direct
consequence of (\ref{bliblib}). Assume that $\mathbf{Q}_{\Lambda _{l}}$ is
not identically zero. For any $\mathcal{J}\in \mathcal{S}_{0}^{\ast }$,
define the map
\begin{equation*}
x\mapsto f_{\mathcal{J}}\left( x\right) :=\underset{\mathcal{E}\in \mathcal{S%
}_{0}:\left\Vert \mathcal{E}\right\Vert _{\Lambda _{l}}=x}{\sup }\left\{
\left\vert \left\langle \mathcal{J},\mathcal{E}\right\rangle \right\vert -%
\frac{x^{2}}{2}\right\}
\end{equation*}%
from $\mathbb{R}_{0}^{+}$ to $\mathbb{R}$. By rescaling, observe that, for
any $x\in \mathbb{R}^{+}$,%
\begin{equation}
f_{\mathcal{J}}\left( x\right) =\underset{\mathcal{E}\in \mathcal{S}%
_{0}:\left\Vert \mathcal{E}\right\Vert _{\Lambda _{l}}=1}{\sup }\left\{
x\left\vert \left\langle \mathcal{J},\mathcal{E}\right\rangle \right\vert -%
\frac{x^{2}}{2}\right\} \ .  \label{jolie2}
\end{equation}%
In particular, for any $\mathcal{J}\in \mathcal{S}_{0}^{\ast }$, $f_{%
\mathcal{J}}$ is clearly continuous. Therefore, we infer from (\ref{bliblib}%
) that
\begin{equation}
\mathbf{Q}_{\Lambda _{l}}^{\ast }\left( \mathcal{J}\right) =2\underset{x\in
\mathbb{R}_{0}^{+}}{\sup }f_{\mathcal{J}}\left( x\right) =2\underset{x\in
\mathbb{R}^{+}}{\sup }f_{\mathcal{J}}\left( x\right) \ ,  \label{jolie1}
\end{equation}%
which, combined with (\ref{jolie2}) and straightforward computations, leads
to the assertion.
\end{proof}

\noindent The above lemma implies that the domain
\begin{equation*}
\mathrm{Dom}\left( \mathbf{Q}_{\Lambda _{l}}^{\ast }\right) =\left\{
\mathcal{J}\in \mathcal{S}_{0}^{\ast }:\mathbf{Q}_{\Lambda _{l}}^{\ast
}\left( \mathcal{J}\right) <\infty \right\}
\end{equation*}%
of the functional $\mathbf{Q}_{\Lambda _{l}}^{\ast }$ is a subspace of $%
\mathcal{S}_{0}^{\ast }$. Similar to (\ref{barack s111}), we define the
semi--norm $\left\Vert \cdot \right\Vert _{\Lambda _{l}}^{(\ast )}\equiv
\left\Vert \cdot \right\Vert _{\Lambda _{l}}^{(\ast ,\beta ,\omega ,\lambda
)}$ by
\begin{equation}
\left\Vert \mathcal{J}\right\Vert _{\Lambda _{l}}^{(\ast )}:=\sqrt{\mathbf{Q}%
_{\Lambda _{l}}^{\ast }\left( \mathcal{J}\right) }=\sup \left\{ \left\vert
\left\langle \mathcal{J},\mathcal{E}\right\rangle \right\vert :\mathcal{E}%
\in \mathcal{S}_{0},\ \left\Vert \mathcal{E}\right\Vert _{\Lambda
_{l}}=1\right\}  \label{barack s112}
\end{equation}%
for any $\mathcal{J}\in \mathcal{S}_{0}^{\ast }$.

Let $\partial \mathbf{Q}_{\Lambda _{l}}^{\ast }\left( \mathcal{J}\right)
\subset \mathcal{S}_{0}$ be the subdifferential of $\mathbf{Q}_{\Lambda
_{l}}^{\ast }$ at the point $\mathcal{J}\in \mathcal{S}_{0}^{\ast }$. The
multifunction $\mathbf{\rho }_{\Lambda _{l}}\equiv \mathbf{\rho }_{\Lambda
_{l}}^{(\beta ,\omega ,\lambda )}$ from $\mathcal{S}_{0}^{\ast }$ to $%
\mathcal{S}_{0}$ (i.e., the set--valued map from $\mathcal{S}_{0}$ to $2^{%
\mathcal{S}_{0}^{\ast }}$) is defined by
\begin{equation}
\mathcal{J}\mapsto \mathbf{\rho }_{\Lambda _{l}}\left( \mathcal{J}\right) =%
\frac{1}{2}\partial \mathbf{Q}_{\Lambda _{l}}^{\ast }\left( \mathcal{J}%
\right) \ .  \label{definition resistivity}
\end{equation}%
It is named here the \emph{AC--resistivity} of the region $\Lambda _{l}$\
because it is a sort of inverse of the AC--conductivity:

\begin{lemma}[Properties of the AC--resistivity]
\label{Lemma LR dia copy(5)}\mbox{
}\newline
The multifunction $\mathbf{\rho }_{\Lambda _{l}}$ has non--empty domain
equal to%
\begin{equation*}
\mathrm{Dom}(\mathbf{\rho }_{\Lambda _{l}}):=\left\{ \mathcal{J}\in \mathcal{%
S}_{0}^{\ast }:\partial \mathbf{Q}_{\Lambda _{l}}^{\ast }\left( \mathcal{J}%
\right) \neq \emptyset \right\} =\underset{\mathcal{E}\in \mathcal{S}_{0}}{%
\bigcup }\mathbf{\sigma }_{\Lambda _{l}}\left( \mathcal{E}\right) \ .
\end{equation*}%
Furthermore, for all $\mathcal{J}\in \mathrm{Dom}(\mathbf{\rho }_{\Lambda
_{l}})$ and $\mathcal{E}\in \mathrm{Dom}(\mathbf{\sigma }_{\Lambda _{l}})=%
\mathcal{S}_{0}$,
\begin{equation}
\mathbf{\sigma }_{\Lambda _{l}}\left( \mathbf{\rho }_{\Lambda _{l}}\left(
\mathcal{J}\right) \right) =\{\mathcal{J}\}\qquad \text{and}\qquad \mathbf{%
\rho }_{\Lambda _{l}}\left( \mathbf{\sigma }_{\Lambda _{l}}\left( \mathcal{E}%
\right) \right) \supset \{\mathcal{E}\}\ .  \label{totototototototototot}
\end{equation}
\end{lemma}

\begin{proof}
Young's inequality asserts that
\begin{equation*}
\frac{1}{2}\mathbf{Q}_{\Lambda _{l}}^{\ast }\left( \mathcal{J}\right) +\frac{%
1}{2}\mathbf{Q}_{\Lambda _{l}}\left( \mathcal{E}\right) \geq \left\langle
\mathcal{J},\mathcal{E}\right\rangle
\end{equation*}%
with equality iff $2\mathcal{J}\in \partial \mathbf{Q}_{\Lambda _{l}}\left(
\mathcal{E}\right) $. As $\mathbf{Q}_{\Lambda _{l}}=\mathbf{Q}_{\Lambda
_{l}}^{\ast \ast }$,
\begin{equation*}
\frac{1}{2}\mathbf{Q}_{\Lambda _{l}}^{\ast }\left( \mathcal{J}\right) +\frac{%
1}{2}\mathbf{Q}_{\Lambda _{l}}\left( \mathcal{E}\right) =\left\langle
\mathcal{J},\mathcal{E}\right\rangle
\end{equation*}%
iff $2\mathcal{E}\in \partial \mathbf{Q}_{\Lambda _{l}}^{\ast }\left(
\mathcal{J}\right) $. In other words,
\begin{equation}
\mathcal{E}\in \mathbf{\rho }_{\Lambda _{l}}\left( \mathcal{J}\right)
\Longleftrightarrow \mathcal{J}\in \mathbf{\sigma }_{\Lambda _{l}}\left(
\mathcal{E}\right) \ .  \label{equivalence yaahoo1}
\end{equation}%
As a consequence, $\mathcal{J}_{\mathcal{E}}\in \mathbf{\sigma }_{\Lambda
_{l}}\left( \mathcal{E}\right) $ (cf. Lemma \ref{Lemma LR dia copy(3)})
yields $\mathcal{E}\in \mathbf{\rho }_{\Lambda _{l}}\left( \mathcal{J}_{%
\mathcal{E}}\right) $. It follows that%
\begin{equation*}
\underset{\mathcal{E}\in \mathcal{S}_{0}}{\bigcup }\mathbf{\sigma }_{\Lambda
_{l}}\left( \mathcal{E}\right) \subset \mathrm{Dom}(\mathbf{\rho }_{\Lambda
_{l}})
\end{equation*}%
and
\begin{equation*}
\mathbf{\rho }_{\Lambda _{l}}\left( \mathbf{\sigma }_{\Lambda _{l}}\left(
\mathcal{E}\right) \right) \supset \{\mathcal{E}\}\ .
\end{equation*}%
Now, let $\mathcal{J}\in \mathrm{Dom}(\mathbf{\rho }_{\Lambda _{l}})$ and $%
\mathcal{E}\in \mathbf{\rho }_{\Lambda _{l}}\left( \mathcal{J}\right) $.
Then, by (\ref{equivalence yaahoo1}), $\mathcal{J}\in \mathbf{\sigma }%
_{\Lambda _{l}}\left( \mathcal{E}\right) $ and we infer from the uniqueness
of the tangent functional (Lemma \ref{Lemma LR dia copy(3)}) that $\mathcal{J%
}=\mathcal{J}_{\mathcal{E}}$. Therefore,%
\begin{equation*}
\mathbf{\sigma }_{\Lambda _{l}}\left( \mathbf{\rho }_{\Lambda _{l}}\left(
\mathcal{J}\right) \right) =\{\mathcal{J}\}
\end{equation*}%
and%
\begin{equation*}
\mathrm{Dom}(\mathbf{\rho }_{\Lambda _{l}})\subset \underset{\mathcal{E}\in
\mathcal{S}_{0}}{\bigcup }\mathbf{\sigma }_{\Lambda _{l}}\left( \mathcal{E}%
\right) \ .
\end{equation*}
\end{proof}

Note that $\mathbf{Q}_{\Lambda _{l}}:\mathcal{S}_{0}\rightarrow \lbrack
0,\infty )$ is a convex continuous functional, by positivity of the
conductivity measure, see Theorem \ref{lemma sigma pos type copy(4)} and (%
\ref{barack s1}). In particular,
\begin{equation}
\mathbf{Q}_{\Lambda _{l}}\left( \mathcal{E}\right) :=2\underset{\mathcal{J}%
\in \mathcal{S}_{0}^{\ast }}{\sup }\left\{ \left\langle \mathcal{J},\mathcal{%
E}\right\rangle -\frac{1}{2}\mathbf{Q}_{\Lambda _{l}}^{\ast }\left( \mathcal{%
J}\right) \right\} \ .  \label{bararck s1bis}
\end{equation}%
Therefore, we deduce from (\ref{bliblib}) and (\ref{bararck s1bis}) that%
\begin{equation*}
\left( \beta _{1},\omega _{1},\lambda _{1}\right) \prec \left( \beta
_{2},\omega _{2},\lambda _{2}\right) \text{\qquad iff\qquad }(\mathbf{Q}%
_{\Lambda _{l}}^{\left( \beta _{1},\omega _{1},\lambda _{1}\right) })^{\ast
}\leq (\mathbf{Q}_{\Lambda _{l}}^{\left( \beta _{2},\omega _{2},\lambda
_{2}\right) })^{\ast }\ .
\end{equation*}%
Furthermore, by using (\ref{barack s111}) and similar arguments as in Lemma %
\ref{Lemma LR dia copy(4)}, if $\mathbf{Q}_{\Lambda _{l}}$ is not
identically zero, then:
\begin{equation*}
\left\Vert \mathcal{E}\right\Vert _{\Lambda _{l}}=\sup \left\{ \left\vert
\left\langle \mathcal{J},\mathcal{E}\right\rangle \right\vert :\mathcal{J}%
\in \mathcal{S}_{0}^{\ast },\ \left\Vert \mathcal{J}\right\Vert _{\Lambda
_{l}}^{(\ast )}=1\right\} \ .
\end{equation*}

We are now in position to obtain Joule's law in its original form. To this
end, we say that a multifunction $\mathbf{\rho }$ from $\mathcal{S}%
_{0}^{\ast }$ to $\mathcal{S}_{0}$ is \emph{linear} if:

\begin{itemize}
\item[(a)] Its domain $\mathrm{Dom}(\mathbf{\rho })$ is a subspace of $%
\mathcal{S}_{0}^{\ast }$.

\item[(b)] For $\alpha \in \mathbb{R}\backslash \{0\}$ and $\mathcal{J}\in
\mathrm{Dom}(\mathbf{\rho })$, $\mathbf{\rho }\left( \alpha \mathcal{J}%
\right) =\alpha \mathbf{\rho }\left( \mathcal{J}\right) $ and $0\in \mathbf{%
\rho }\left( 0\right) $.

\item[(c)] For $\mathcal{J}_{1},\mathcal{J}_{2}\in \mathrm{Dom}(\mathbf{\rho
})$, $\mathbf{\rho }\left( \mathcal{J}_{1}+\mathcal{J}_{2}\right) =\mathbf{%
\rho }\left( \mathcal{J}_{1}\right) +\mathbf{\rho }\left( \mathcal{J}%
_{2}\right) $.
\end{itemize}

\noindent Then, one gets that the heat produced by currents is proportional
to the resistivity and the square of currents:

\begin{satz}[Microscopic Joule's law -- II]
\label{Local Ohm's law thm copy(1)}\mbox{
}\newline
\emph{(i)} $\mathbf{\rho }_{\Lambda _{l}}$ is a linear multifunction and $%
\mathbf{\sigma }_{\Lambda _{l}}\left( \mathbf{\rho }_{\Lambda _{l}}\left(
\mathcal{J}\right) \right) =\{\mathcal{J}\}$ for all $\mathcal{J}\in \mathrm{%
Dom}(\mathbf{\rho }_{\Lambda _{l}})$.\newline
\emph{(ii)}\ For any $\mathcal{J}\in \mathrm{Dom}(\mathbf{\rho }_{\Lambda
_{l}})$,
\begin{equation*}
\{\mathbf{Q}_{\Lambda _{l}}^{\ast }\left( \mathcal{J}\right) \}=\left\langle
\mathcal{J},\mathbf{\rho }_{\Lambda _{l}}\left( \mathcal{J}\right)
\right\rangle =\mathbf{Q}_{\Lambda _{l}}\left( \mathbf{\rho }_{\Lambda
_{l}}\left( \mathcal{J}\right) \right) \ .
\end{equation*}%
\emph{(iii)}\ There is a bilinear symmetric positive map $\left( \cdot
,\cdot \right) _{\Lambda _{l}}^{(\ast )}$ on $\mathrm{Dom}(\mathbf{\rho }%
_{\Lambda _{l}})$ such that
\begin{equation*}
\mathbf{Q}_{\Lambda _{l}}^{\ast }\left( \mathcal{J}_{1}\right) =\left(
\mathcal{J}_{1},\mathcal{J}_{1}\right) _{\Lambda _{l}}^{(\ast )}\qquad \text{%
and}\qquad \left\langle \mathcal{J}_{1},\mathbf{\rho }_{\Lambda _{l}}\left(
\mathcal{J}_{2}\right) \right\rangle =\{\left( \mathcal{J}_{1},\mathcal{J}%
_{2}\right) _{\Lambda _{l}}^{(\ast )}\}
\end{equation*}%
for all $\mathcal{J}_{1},\mathcal{J}_{2}\in \mathrm{Dom}(\mathbf{\rho }%
_{\Lambda _{l}})$.
\end{satz}

\begin{proof}
(i.a) The fact that $\mathrm{Dom}(\mathbf{\rho }_{\Lambda _{l}})$ is a
subspace of $\mathcal{S}_{0}^{\ast }$ is a direct consequence of Lemmata \ref%
{Lemma LR dia copy(3)} and \ref{Lemma LR dia copy(5)}.

\noindent (i.b) Let $\alpha \in \mathbb{R}$ and $\mathcal{J}\in \mathrm{Dom}(%
\mathbf{\rho }_{\Lambda _{l}})$. Take any $\mathcal{E}_{\mathcal{J}}\in
\mathbf{\rho }_{\Lambda _{l}}\left( \mathcal{J}\right) $ and observe that $%
\mathcal{J}=\mathcal{J}_{\mathcal{E}_{\mathcal{J}}}$, by using Lemmata \ref%
{Lemma LR dia copy(3)} and \ref{Lemma LR dia copy(5)}. Then,
\begin{equation*}
\alpha \mathcal{J}=\mathcal{J}_{\alpha \mathcal{E}_{\mathcal{J}}}\in \mathbf{%
\sigma }_{\Lambda _{l}}\left( \alpha \mathcal{E}_{\mathcal{J}}\right) \ .
\end{equation*}%
From (\ref{equivalence yaahoo1}) it follows that $\alpha \mathbf{\rho }%
_{\Lambda _{l}}\left( \mathcal{J}\right) \subset \mathbf{\rho }_{\Lambda
_{l}}\left( \alpha \mathcal{J}\right) $. If $\alpha \neq 0$ then, by
replacing $\left( \mathcal{J},\alpha \right) $ with $\left( \alpha \mathcal{J%
},\alpha ^{-1}\right) $, one gets that $\mathbf{\rho }_{\Lambda _{l}}\left(
\alpha \mathcal{J}\right) \subset \alpha \mathbf{\rho }_{\Lambda _{l}}\left(
\mathcal{J}\right) $.

\noindent (i.c) Let $\mathcal{J}_{1},\mathcal{J}_{2}\in \mathrm{Dom}(\mathbf{%
\rho }_{\Lambda _{l}})$ and take any $\mathcal{E}_{\mathcal{J}_{1}}\in
\mathbf{\rho }_{\Lambda _{l}}\left( \mathcal{J}_{1}\right) $ and $\mathcal{E}%
_{\mathcal{J}_{2}}\in \mathbf{\rho }_{\Lambda _{l}}\left( \mathcal{J}%
_{2}\right) $. As above, $\mathcal{J}_{1}=\mathcal{J}_{\mathcal{E}_{\mathcal{%
J}_{1}}}$ and $\mathcal{J}_{2}=\mathcal{J}_{\mathcal{E}_{\mathcal{J}_{2}}}$.
Then,
\begin{equation*}
\mathcal{J}_{1}+\mathcal{J}_{2}=\mathcal{J}_{\mathcal{E}_{\mathcal{J}_{1}}+%
\mathcal{E}_{\mathcal{J}_{2}}}\in \mathbf{\sigma }_{\Lambda _{l}}\left(
\mathcal{E}_{\mathcal{J}_{1}}+\mathcal{E}_{\mathcal{J}_{2}}\right) \ .
\end{equation*}%
Hence, using again (\ref{equivalence yaahoo1}), we arrive at
\begin{equation*}
\mathbf{\rho }_{\Lambda _{l}}\left( \mathcal{J}_{1}\right) +\mathbf{\rho }%
_{\Lambda _{l}}\left( \mathcal{J}_{2}\right) \subset \mathbf{\rho }_{\Lambda
_{l}}\left( \mathcal{J}_{1}+\mathcal{J}_{2}\right) \ .
\end{equation*}

Now, let $\mathcal{J}_{1},\mathcal{J}_{2}\in \mathrm{Dom}(\mathbf{\rho }%
_{\Lambda _{l}})$ and take any $\mathcal{E}_{\mathcal{J}_{1}+\mathcal{J}%
_{2}}\in \mathbf{\rho }_{\Lambda _{l}}\left( \mathcal{J}_{1}+\mathcal{J}%
_{2}\right) $. Then, $\mathcal{J}_{\mathcal{E}_{\mathcal{J}_{1}+\mathcal{J}%
_{2}}}=\mathcal{J}_{1}+\mathcal{J}_{2}$. Similarly, choose also $\mathcal{E}%
_{\mathcal{J}_{1}}\in \mathbf{\rho }_{\Lambda _{l}}\left( \mathcal{J}%
_{1}\right) $ and $\mathcal{E}_{\mathcal{J}_{2}}\in \mathbf{\rho }_{\Lambda
_{l}}\left( \mathcal{J}_{2}\right) $ with $\mathcal{J}_{1}=\mathcal{J}_{%
\mathcal{E}_{\mathcal{J}_{1}}}$ and $\mathcal{J}_{2}=\mathcal{J}_{\mathcal{E}%
_{\mathcal{J}_{2}}}$. Obviously, by Equation (\ref{barack s2}),
\begin{equation*}
\mathcal{J}_{2}=\mathcal{J}_{\mathcal{E}_{\mathcal{J}_{2}}}=\mathcal{J}_{%
\mathcal{E}_{\mathcal{J}_{1}+\mathcal{J}_{2}}}-\mathcal{J}_{\mathcal{E}_{%
\mathcal{J}_{1}}}=\mathcal{J}_{\mathcal{E}_{\mathcal{J}_{1}+\mathcal{J}_{2}}-%
\mathcal{E}_{\mathcal{J}_{1}}}\ ,
\end{equation*}%
which together with (\ref{equivalence yaahoo1}) yields the converse inclusion%
\begin{equation*}
\mathbf{\rho }_{\Lambda _{l}}\left( \mathcal{J}_{1}+\mathcal{J}_{2}\right)
\subset \mathbf{\rho }_{\Lambda _{l}}\left( \mathcal{J}_{1}\right) +\mathbf{%
\rho }_{\Lambda _{l}}\left( \mathcal{J}_{2}\right) \ .
\end{equation*}

\noindent (ii) Take any $\mathcal{J}\in \mathrm{Dom}(\mathbf{\rho }_{\Lambda
_{l}})$ and $\mathcal{E}_{\mathcal{J}}\in \mathbf{\rho }_{\Lambda
_{l}}\left( \mathcal{J}\right) $. We infer from (\ref{barack s0}) and Lemma %
\ref{Lemma LR dia copy(3)} that
\begin{equation*}
\left\langle \mathcal{J},\mathcal{E}_{\mathcal{J}}\right\rangle
=\left\langle \mathcal{J}_{\mathcal{E}_{\mathcal{J}}},\mathcal{E}_{\mathcal{J%
}}\right\rangle =\mathbf{Q}_{\Lambda _{l}}\left( \mathcal{E}_{\mathcal{J}%
}\right) \ .
\end{equation*}%
Since
\begin{equation*}
\frac{1}{2}\mathbf{Q}_{\Lambda _{l}}^{\ast }\left( \mathcal{J}\right) +\frac{%
1}{2}\mathbf{Q}_{\Lambda _{l}}\left( \mathcal{E}_{\mathcal{J}}\right)
=\left\langle \mathcal{J},\mathcal{E}_{\mathcal{J}}\right\rangle \ ,
\end{equation*}%
we also deduce that $\mathbf{Q}_{\Lambda _{l}}^{\ast }\left( \mathcal{J}%
\right) =\mathbf{Q}_{\Lambda _{l}}\left( \mathcal{E}_{\mathcal{J}}\right) $.

\noindent (iii) For all $\mathcal{J}_{1},\mathcal{J}_{2}\in \mathrm{Dom}(%
\mathbf{Q}_{\Lambda _{l}}^{\ast })$, define
\begin{equation}
\left( \mathcal{J}_{1},\mathcal{J}_{2}\right) _{\Lambda _{l}}^{(\ast )}:=%
\frac{1}{4}\left( \mathbf{Q}_{\Lambda _{l}}^{\ast }\left( \mathcal{J}_{1}+%
\mathcal{J}_{2}\right) -\mathbf{Q}_{\Lambda _{l}}^{\ast }\left( \mathcal{J}%
_{1}-\mathcal{J}_{2}\right) \right) \ .  \label{definition yahoo}
\end{equation}%
This quantity is clearly symmetric w.r.t. $\mathcal{J}_{1},\mathcal{J}_{2}$
and
\begin{equation*}
\left( \mathcal{J},\mathcal{J}\right) _{\Lambda _{l}}^{(\ast )}=\mathbf{Q}%
_{\Lambda _{l}}^{\ast }\left( \mathcal{J}\right) \geq 0\ ,\qquad \mathcal{J}%
\in \mathrm{Dom}(\mathbf{Q}_{\Lambda _{l}}^{\ast })\ ,
\end{equation*}%
by Lemma \ref{Lemma LR dia copy(4)}. Using the linearity of $\mathbf{\rho }%
_{\Lambda _{l}}$ and the fact that $\langle \mathcal{J},\mathbf{\rho }%
_{\Lambda _{l}}(\mathcal{J})\rangle \subset \mathbb{R}_{0}^{+}$ contains
exactly one element for all $\mathcal{J}\in \mathrm{Dom}(\mathbf{\rho }%
_{\Lambda _{l}})$, we compute that, for any $\mathcal{J}_{1},\mathcal{J}%
_{2}\in \mathrm{Dom}(\mathbf{\rho }_{\Lambda _{l}})$,
\begin{equation*}
\frac{1}{2}\{\mathbf{Q}_{\Lambda _{l}}^{\ast }\left( \mathcal{J}_{1}+%
\mathcal{J}_{2}\right) -\mathbf{Q}_{\Lambda _{l}}^{\ast }\left( \mathcal{J}%
_{1}-\mathcal{J}_{2}\right) \}=\left\langle \mathcal{J}_{2},\mathbf{\rho }%
_{\Lambda _{l}}\left( \mathcal{J}_{1}\right) \right\rangle +\left\langle
\mathcal{J}_{1},\mathbf{\rho }_{\Lambda _{l}}\left( \mathcal{J}_{2}\right)
\right\rangle \ .
\end{equation*}%
Again by linearity of $\mathbf{\rho }_{\Lambda _{l}}$, this implies that (%
\ref{definition yahoo}) defines a bilinear form on $\mathrm{Dom}(\mathbf{%
\rho }_{\Lambda _{l}})$. We also infer from the above equation that the set $%
\langle \mathcal{J}_{2},\mathbf{\rho }_{\Lambda _{l}}(\mathcal{J}%
_{1})\rangle \subset \mathbb{R}$ contains exactly one element. Let $\mathcal{%
E}_{\mathcal{J}_{1}}\in \mathbf{\rho }_{\Lambda _{l}}\left( \mathcal{J}%
_{1}\right) $ and $\mathcal{E}_{\mathcal{J}_{2}}\in \mathbf{\rho }_{\Lambda
_{l}}\left( \mathcal{J}_{2}\right) $ with $\mathcal{J}_{1}=\mathcal{J}_{%
\mathcal{E}_{\mathcal{J}_{1}}}$ and $\mathcal{J}_{2}=\mathcal{J}_{\mathcal{E}%
_{\mathcal{J}_{2}}}$. Then, by Lemma \ref{Lemma LR dia copy(3)},%
\begin{equation*}
\left\langle \mathcal{J}_{2},\mathbf{\rho }_{\Lambda _{l}}\left( \mathcal{J}%
_{1}\right) \right\rangle =\left\{ \left\langle \mathcal{J}_{\mathcal{E}_{%
\mathcal{J}_{2}}},\mathcal{E}_{\mathcal{J}_{1}}\right\rangle \right\}
=\left\{ \left\langle \mathcal{J}_{\mathcal{E}_{\mathcal{J}_{1}}},\mathcal{E}%
_{\mathcal{J}_{2}}\right\rangle \right\} =\left\langle \mathcal{J}_{1},%
\mathbf{\rho }_{\Lambda _{l}}\left( \mathcal{J}_{2}\right) \right\rangle \ .
\end{equation*}
\end{proof}

\section{Technical Proofs\label{sect technical proofs}}

This section is divided in two parts: Section \ref{Sect Trport coeff} gives
a detailed proof of Theorem \ref{lemma sigma pos type copy(4)} as well as
additional properties of paramagnetic transport coefficients defined in
Section \ref{Sect Trans coeef ddef}. In Section \ref{Section local joule-ohm}
we prove Theorems \ref{thm Local Ohm's law} and \ref{Local Ohm's law thm
copy(2)}. Note that we start in this second subsection with the proof of
Theorem \ref{Local Ohm's law thm copy(2)} because the other one is simpler
and uses similar arguments.

\subsection{Paramagnetic Transport Coefficients\label{Sect Trport coeff}}

\subsubsection{Microscopic Paramagnetic Transport Coefficients\label%
{Micorscopic section}}

We study in this subsection the microscopic paramagnetic transport
coefficient $\sigma _{\mathrm{p}}^{(\omega )}$ which is defined by (\ref%
{backwards -1bis}), that is,
\begin{equation*}
\sigma _{\mathrm{p}}^{(\omega )}\left( \mathbf{x},\mathbf{y},t\right)
:=\int\nolimits_{0}^{t}\varrho ^{(\beta ,\omega ,\lambda )}\left( i[I_{%
\mathbf{y}},\tau _{s}^{(\omega ,\lambda )}(I_{\mathbf{x}})]\right) \mathrm{d}%
s\ ,\quad \mathbf{x},\mathbf{y}\in \mathfrak{L}^{2}\ ,\ t\in \mathbb{R}\ .
\end{equation*}%
Recall that $I_{\mathbf{x}}$ is the paramagnetic current observable defined
by (\ref{current observable}), that is,
\begin{equation}
I_{\mathbf{x}}:=i(a_{x^{(2)}}^{\ast }a_{x^{(1)}}-a_{x^{(1)}}^{\ast
}a_{x^{(2)}})\ ,\qquad \mathbf{x}:=(x^{(1)},x^{(2)})\in \mathfrak{L}^{2}\ .
\label{proche voisinsbis}
\end{equation}

The coefficient $\sigma _{\mathrm{p}}^{(\omega )}$ can explicitly be written
in terms of a scalar product involving current observables. To show this, we
introduce the Duhamel two--point function $(\cdot ,\cdot )_{\sim }^{(\omega
)}$ defined by
\begin{equation}
(B_{1},B_{2})_{\sim }\equiv (B_{1},B_{2})_{\sim }^{(\beta ,\omega ,\lambda
)}:=\int\nolimits_{0}^{\beta }\varrho ^{(\beta ,\omega ,\lambda )}\left(
B_{1}^{\ast }\tau _{i\alpha }^{(\omega ,\lambda )}(B_{2})\right) \mathrm{d}%
\alpha  \label{def bogo jetman}
\end{equation}%
for any $B_{1},B_{2}\in \mathcal{U}$. The properties of this sesquilinear
form are described in detail in Appendix \ref{Section Duhamel Two--Point
Functions}. In particular, by Theorem \ref{toto fluctbis copy(2)} for $%
\mathcal{X}=\mathcal{U}$, $\tau =\tau ^{(\omega ,\lambda )}$ and $\varrho
=\varrho ^{(\beta ,\omega ,\lambda )}$, $(B_{1},B_{2})\mapsto
(B_{1},B_{2})_{\sim }$ is a positive sesquilinear form on $\mathcal{U}$. We
then infer from Lemma \ref{lemma conductivty4 copy(1)}\ that%
\begin{equation}
\sigma _{\mathrm{p}}^{(\omega )}\left( \mathbf{x},\mathbf{y},t\right) =(I_{%
\mathbf{y}},\tau _{t}^{(\omega ,\lambda )}(I_{\mathbf{x}}))_{\sim }-(I_{%
\mathbf{y}},I_{\mathbf{x}})_{\sim }\ ,  \label{lemma duhamel jetman eq}
\end{equation}%
for all $\beta \in \mathbb{R}^{+}$, $\omega \in \Omega $, $\lambda \in
\mathbb{R}_{0}^{+}$, $\mathbf{x},\mathbf{y}\in \mathfrak{L}^{2}$ and $t\in
\mathbb{R}$. By Theorem \ref{Thm important equality asymptotics copy(1)}, it
follows that $\sigma _{\mathrm{p}}^{(\omega )}$ is symmetric w.r.t.
time--reversal and permutation of bonds.

Indeed, by using the time--reversal operation $\Theta :\mathcal{U}%
\rightarrow \mathcal{U}$ defined in Section \ref{Dynamics free sect}, one
proves:

\begin{lemma}[Time--reversal symmetry of the fermion system]
\label{lemma time reversal}\mbox{
}\newline
Let $\beta \in \mathbb{R}^{+}$, $\omega \in \Omega $ and $\lambda \in
\mathbb{R}_{0}^{+}$. Then,%
\begin{equation}
\Theta \circ \tau _{t}^{(\omega ,\lambda )}=\tau _{-t}^{(\omega ,\lambda
)}\circ \Theta \ ,\text{\qquad }t\in \mathbb{R}\ ,  \label{time reversal1}
\end{equation}%
and
\begin{equation}
\varrho ^{(\beta ,\omega ,\lambda )}\left( B\right) =\overline{\varrho
^{(\beta ,\omega ,\lambda )}\circ \Theta \left( B\right) }\ ,\qquad B\in
\mathcal{U}\ .  \label{time reversal2}
\end{equation}
\end{lemma}

\begin{proof}
By continuity of the maps $\Theta $ and $\tau _{t}^{(\omega ,\lambda )}$ as
well as the density of polynomials in the creation and annihilation
operators in $\mathcal{U}$, it suffices to prove the first assertion for
monomials in $a_{x},a_{x}^{\ast }$, $x\in \mathfrak{L}$. Now, since $\Theta
(H_{L}^{(\omega ,\lambda )})=H_{L}^{(\omega ,\lambda )}$ (see (\ref%
{definition de Hl})), by \cite[Theorem A.3 (i)]{OhmI},
\begin{equation*}
\Theta \circ \tau _{t}^{(\omega ,\lambda )}\left( B\right) =\tau
_{-t}^{(\omega ,\lambda )}\circ \Theta \left( B\right) \ ,\text{\qquad }B\in
\mathcal{U}_{0},\ t\in \mathbb{R}\ ,
\end{equation*}%
which implies (\ref{time reversal1}). The second assertion is a consequence
of the uniqueness of the $(\tau ^{(\omega ,\lambda )},\beta )$--KMS state $%
\varrho ^{(\beta ,\omega ,\lambda )}$ together with Lemma \ref{Lemma LR dia
copy(1)}.
\end{proof}

\noindent Since $\Theta \left( I_{\mathbf{x}}\right) =-I_{\mathbf{x}}$ for
any $\mathbf{x}\in \mathfrak{L}^{2}$, we deduce from Lemma \ref{lemma time
reversal} and Theorem \ref{Thm important equality asymptotics copy(1)} for $%
\mathcal{X}=\mathcal{U}$, $\tau =\tau ^{(\omega ,\lambda )}$ and $\varrho
=\varrho ^{(\beta ,\omega ,\lambda )}$ that the function $\sigma _{\mathrm{p}%
}^{(\omega )}$ from $\mathfrak{L}^{4}\times \mathbb{R}$ to $\mathbb{R}$ is
symmetric w.r.t. time--reversal and permutation of bonds:%
\begin{equation*}
\sigma _{\mathrm{p}}^{(\omega )}\left( \mathbf{x},\mathbf{y},t\right)
=\sigma _{\mathrm{p}}^{(\omega )}\left( \mathbf{x},\mathbf{y},-t\right)
=\sigma _{\mathrm{p}}^{(\omega )}\left( \mathbf{y},\mathbf{x},t\right) \
,\qquad \mathbf{x},\mathbf{y}\in \mathfrak{L}^{2}\ ,\ t\in \mathbb{R}\ .
\end{equation*}

Thermal equilibrium states $\varrho ^{(\beta ,\omega ,\lambda )}$ are by
construction quasi--free and gauge--invariant. This fact implies that $%
\sigma _{\mathrm{p}}^{(\omega )}$ can be expressed in terms of complex--time
two--point correlation functions $C_{t+i\alpha }^{(\omega )}\equiv
C_{t+i\alpha }^{(\beta ,\omega ,\lambda )}$ defined by%
\begin{equation}
C_{t+i\alpha }^{(\omega )}(\mathbf{x}):=\varrho ^{(\beta ,\omega ,\lambda
)}(a_{x^{(1)}}^{\ast }\tau _{t+i\alpha }^{(\omega ,\lambda )}(a_{x^{(2)}}))\
,\qquad \mathbf{x}:=(x^{(1)},x^{(2)})\in \mathfrak{L}^{2}\ ,
\label{def.propagator}
\end{equation}%
for all $\beta \in \mathbb{R}^{+}$, $\omega \in \Omega $, $\lambda \in
\mathbb{R}_{0}^{+}$, $t\in {\mathbb{R}}$ and $\alpha \in \lbrack 0,\beta ]$.
This is shown in the following assertion:

\begin{lemma}[$\protect\sigma _{\mathrm{p}}^{(\protect\omega )}$ in terms of
two--point correlation functions]
\label{local AC-conductivity lemma copy(1)}\mbox{
}\newline
Let $\beta \in \mathbb{R}^{+}$, $\omega \in \Omega $ and $\lambda \in
\mathbb{R}_{0}^{+}$. Then, for all $\mathbf{x},\mathbf{y}\in \mathfrak{L}%
^{2} $ and $t\in \mathbb{R}$,
\begin{equation*}
\sigma _{\mathrm{p}}^{(\omega )}\left( \mathbf{x},\mathbf{y},t\right)
=\int\nolimits_{0}^{\beta }\left( \mathfrak{C}_{t+i\alpha }^{(\omega )}(%
\mathbf{x},\mathbf{y})-\mathfrak{C}_{i\alpha }^{(\omega )}(\mathbf{x},%
\mathbf{y})\right) \mathrm{d}\alpha \in \mathbb{R}\ ,
\end{equation*}%
where $\mathfrak{C}_{t+i\alpha }^{(\omega )}\equiv \mathfrak{C}_{t+i\alpha
}^{(\beta ,\omega ,\lambda )}$ is the map from $\mathfrak{L}^{4}$ to $%
\mathbb{C}$ defined by%
\begin{equation}
\mathfrak{C}_{t+i\alpha }^{(\omega )}(\mathbf{x},\mathbf{y}):=\underset{\pi
,\pi ^{\prime }\in S_{2}}{\sum }\varepsilon _{\pi }\varepsilon _{\pi
^{\prime }}C_{t+i\alpha }^{(\omega )}(y^{\pi ^{\prime }(1)},x^{\pi
(1)})C_{-t+i(\beta -\alpha )}^{(\omega )}(x^{\pi (2)},y^{\pi ^{\prime }(2)})
\label{map cool}
\end{equation}%
for any $\mathbf{x}:=(x^{(1)},x^{(2)})\in \mathfrak{L}^{2}$ and $\mathbf{y}%
:=(y^{(1)},y^{(2)})\in \mathfrak{L}^{2}$. Here, $\pi ,\pi ^{\prime }\in
S_{2} $ are by definition permutations of $\{1,2\}$ with signatures $%
\varepsilon _{\pi },\varepsilon _{\pi ^{\prime }}\in \{-1,1\}$.
\end{lemma}

\begin{proof}
Fix $\beta \in \mathbb{R}^{+}$, $\omega \in \Omega $, $\lambda \in \mathbb{R}%
_{0}^{+}$, $t\in \mathbb{R}$, $\alpha \in \lbrack 0,\beta ]$, $\mathbf{x}%
:=(x^{(1)},x^{(2)})\in \mathfrak{L}^{2}$ and $\mathbf{y}:=(y^{(1)},y^{(2)})%
\in \mathfrak{L}^{2}$. From Equation (\ref{lemma duhamel jetman eq})
together with (\ref{eq caclu a la con}),
\begin{equation}
\sigma _{\mathrm{p}}^{(\omega )}\left( \mathbf{x},\mathbf{y},t\right)
=\int\nolimits_{0}^{\beta }\left( \varrho ^{(\beta ,\omega ,\lambda )}\left(
I_{\mathbf{y}}\tau _{t+i\alpha }^{(\omega ,\lambda )}(I_{\mathbf{x}})\right)
-\varrho ^{(\beta ,\omega ,\lambda )}\left( I_{\mathbf{y}}\tau _{i\alpha
}^{(\omega ,\lambda )}(I_{\mathbf{x}})\right) \right) \mathrm{d}\alpha \ .
\label{backwards}
\end{equation}%
Direct computations using (\ref{inequality idiote}) and (\ref{current
observable}) yield%
\begin{eqnarray}
I_{\mathbf{y}}\tau _{t+i\alpha }^{(\omega ,\lambda )}(I_{\mathbf{x}})
&=&-\left( a_{y^{(1)}}^{\ast }a_{y^{(2)}}-a_{y^{(2)}}^{\ast
}a_{y^{(1)}}\right) \tau _{t+i\alpha }^{(\omega ,\lambda
)}(a_{x^{(1)}}^{\ast })\tau _{t+i\alpha }^{(\omega ,\lambda )}(a_{x^{(2)}})
\label{sigma ohm2} \\
&&+\left( a_{y^{(1)}}^{\ast }a_{y^{(2)}}-a_{y^{(2)}}^{\ast
}a_{y^{(1)}}\right) \tau _{t+i\alpha }^{(\omega ,\lambda
)}(a_{x^{(2)}}^{\ast })\tau _{t+i\alpha }^{(\omega ,\lambda )}(a_{x^{(1)}})\
.  \notag
\end{eqnarray}%
Note that, for all $\mathbf{x}\in \mathfrak{L}^{2}$ and $x\in \mathfrak{L}$,
the maps
\begin{equation}
z\mapsto \tau _{z}^{(\omega ,\lambda )}(I_{\mathbf{x}})\ ,\quad z\mapsto
\tau _{z}^{(\omega ,\lambda )}(a_{x}^{\ast })\ ,\quad z\mapsto \tau
_{z}^{(\omega ,\lambda )}(a_{x})\ ,  \label{analytic map}
\end{equation}%
defined on $\mathbb{R}$ have unique analytic continuations for $z\in \mathbb{%
C}$ and (\ref{sigma ohm2}) makes sense.

Recall that $\mathfrak{e}_{x}(y)\equiv \delta _{x,y}$ is the canonical
orthonormal basis of $\ell ^{2}(\mathfrak{L})$ and, as usual,
\begin{equation*}
\{B_{1},B_{2}\}:=B_{1}B_{2}+B_{2}B_{1}\ ,\qquad B_{1},B_{2}\in \mathcal{U}\ .
\end{equation*}%
Therefore, using the anti--commutator relation
\begin{equation*}
\{a_{y^{(2)}},\tau _{t+i\alpha }^{(\omega ,\lambda )}(a_{x^{(1)}}^{\ast
})\}=\langle \mathfrak{e}_{y^{(2)}},(\mathrm{U}_{t+i\alpha }^{(\omega
,\lambda )})^{\ast }\mathfrak{e}_{x^{(1)}}\rangle \mathbf{1}\ ,
\end{equation*}%
see (\ref{CAR}) and (\ref{rescaledbis}), we get the equality%
\begin{eqnarray}
&&\varrho ^{(\beta ,\omega ,\lambda )}\left( a_{y^{(1)}}^{\ast
}a_{y^{(2)}}\tau _{t+i\alpha }^{(\omega ,\lambda )}(a_{x^{(1)}}^{\ast })\tau
_{t+i\alpha }^{(\omega ,\lambda )}(a_{x^{(2)}})\right)  \notag \\
&=&-\varrho ^{(\beta ,\omega ,\lambda )}\left( a_{y^{(1)}}^{\ast }\tau
_{t+i\alpha }^{(\omega ,\lambda )}(a_{x^{(1)}}^{\ast })a_{y^{(2)}}\tau
_{t+i\alpha }^{(\omega ,\lambda )}(a_{x^{(2)}})\right)  \notag \\
&&+\varrho ^{(\beta ,\omega ,\lambda )}\left( \{a_{y^{(2)}},\tau _{t+i\alpha
}^{(\omega ,\lambda )}(a_{x^{(1)}}^{\ast })\}\right) \varrho ^{(\beta
,\omega ,\lambda )}\left( a_{y^{(1)}}^{\ast }\tau _{t+i\alpha }^{(\omega
,\lambda )}(a_{x^{(2)}})\right) \ .  \label{quasi-free blabla0}
\end{eqnarray}
Since $\varrho ^{(\beta ,\omega ,\lambda )}$\ is by construction a
quasi--free state, we use \cite[p. 48]{BratteliRobinson}, that is here,
\begin{eqnarray*}
&&\varrho ^{(\beta ,\omega ,\lambda )}(a^{\ast }\left( f_{1}\right) a^{\ast
}\left( f_{2}\right) a\left( g_{1}\right) a\left( g_{2}\right) ) \\
&=&\varrho ^{(\beta ,\omega ,\lambda )}(a^{\ast }\left( f_{1}\right) a\left(
g_{2}\right) )\varrho ^{(\beta ,\omega ,\lambda )}(a^{\ast }\left(
f_{2}\right) a\left( g_{1}\right) ) \\
&&-\varrho ^{(\beta ,\omega ,\lambda )}(a^{\ast }\left( f_{1}\right) a\left(
g_{1}\right) )\varrho ^{(\beta ,\omega ,\lambda )}(a^{\ast }\left(
f_{2}\right) a\left( g_{2}\right) )\ ,
\end{eqnarray*}%
to infer from Equation (\ref{quasi-free blabla0}) that%
\begin{eqnarray}
&&\varrho ^{(\beta ,\omega ,\lambda )}\left( a_{y^{(1)}}^{\ast
}a_{y^{(2)}}\tau _{t+i\alpha }^{(\omega ,\lambda )}(a_{x^{(1)}}^{\ast })\tau
_{t+i\alpha }^{(\omega ,\lambda )}(a_{x^{(2)}})\right)  \notag \\
&=&\varrho ^{(\beta ,\omega ,\lambda )}(a_{y^{(1)}}^{\ast
}a_{y^{(2)}})\varrho ^{(\beta ,\omega ,\lambda )}(\tau _{t+i\alpha
}^{(\omega ,\lambda )}(a_{x^{(1)}}^{\ast })\tau _{t+i\alpha }^{(\omega
,\lambda )}(a_{x^{(2)}}))  \notag \\
&&+\varrho ^{(\beta ,\omega ,\lambda )}\left( a_{y^{(1)}}^{\ast }\tau
_{t+i\alpha }^{(\omega ,\lambda )}(a_{x^{(2)}})\right) \varrho ^{(\beta
,\omega ,\lambda )}\left( a_{y^{(2)}}\tau _{t+i\alpha }^{(\omega ,\lambda
)}(a_{x^{(1)}}^{\ast })\right) \ .  \label{quasi-free blabla1}
\end{eqnarray}%
Remark that the KMS property (\ref{KMS property}) together with (\ref%
{stationary}) and the Phragm{\'e}n--Lindel\"{o}f theorem \cite[Proposition
5.3.5]{BratteliRobinson} yields%
\begin{equation}
\varrho ^{(\beta ,\omega ,\lambda )}(\tau _{t+i\alpha }^{(\omega ,\lambda
)}(B))=\varrho ^{(\beta ,\omega ,\lambda )}(B)\ ,\qquad B\in \mathcal{U}\ .
\label{stationarybis}
\end{equation}%
See also \cite[Proposition 5.3.7]{BratteliRobinson}. We thus combine (\ref%
{stationarybis}) and (\ref{KMS property}) with Equation (\ref{inequality
idiote}) and the analyticity of the maps (\ref{analytic map}) to deduce from
(\ref{def.propagator}) that%
\begin{equation*}
C_{-t+i(\beta -\alpha )}^{(\omega )}(\mathbf{x})=\varrho ^{(\beta ,\omega
,\lambda )}(a_{x^{(2)}}\tau _{t+i\alpha }^{(\omega ,\lambda
)}(a_{x^{(1)}}^{\ast }))\ .
\end{equation*}%
Using this together with (\ref{def.propagator}), (\ref{stationarybis}) and
again the analyticity of the maps (\ref{analytic map}), we get from Equation
(\ref{quasi-free blabla1}) that%
\begin{eqnarray*}
&&\varrho ^{(\beta ,\omega ,\lambda )}\left( a_{y^{(1)}}^{\ast
}a_{y^{(2)}}\tau _{t+i\alpha }^{(\omega ,\lambda )}(a_{x^{(1)}}^{\ast })\tau
_{t+i\alpha }^{(\omega ,\lambda )}(a_{x^{(2)}})\right) \\
&=&C_{0}^{(\omega )}(y^{(1)},y^{(2)})C_{0}^{(\omega
)}(x^{(1)},x^{(2)})+C_{t+i\alpha }^{(\omega )}(y^{(1)},x^{(2)})C_{-t+i(\beta
-\alpha )}^{(\omega )}(x^{(1)},y^{(2)})\ .
\end{eqnarray*}%
Then we use this last equality together with (\ref{sigma ohm2}) to get%
\begin{eqnarray}
&&\varrho ^{(\beta ,\omega ,\lambda )}\left( I_{\mathbf{y}}\tau _{t+i\alpha
}^{(\omega ,\lambda )}(I_{\mathbf{x}})\right)  \notag \\
&=&-\underset{\pi ,\pi ^{\prime }\in S_{2}}{\sum }\varepsilon _{\pi
}\varepsilon _{\pi ^{\prime }}\left( C_{t+i\alpha }^{(\omega )}(y^{\pi
^{\prime }(1)},x^{\pi (2)})C_{-t+i(\beta -\alpha )}^{(\omega )}(x^{\pi
(1)},y^{\pi ^{\prime }(2)})\right.  \notag \\
&&\left. +C_{0}^{(\omega )}(y^{\pi ^{\prime }(1)},y^{\pi ^{\prime
}(2)})C_{0}^{(\omega )}(x^{\pi (1)},x^{\pi (2)})\right) \ .
\label{sigma ohm3}
\end{eqnarray}%
Therefore, the assertion follows by combining (\ref{backwards}) with (\ref%
{sigma ohm3}) for any $\beta \in \mathbb{R}^{+}$, $\omega \in \Omega $, $%
\lambda \in \mathbb{R}_{0}^{+}$, $t\in \mathbb{R}$, $\alpha \in \lbrack
0,\beta ]$, $\mathbf{x}:=(x^{(1)},x^{(2)})\in \mathfrak{L}^{2}$ and $\mathbf{%
y}:=(y^{(1)},y^{(2)})\in \mathfrak{L}^{2}$.
\end{proof}

Lemma \ref{local AC-conductivity lemma copy(1)} is a useful technical result
because the complex--time two--point correlation functions $C_{t+i\alpha
}^{(\omega )}$ can be expressed in terms of the one--particle bounded
self--adjoint operator $(\Delta _{\mathrm{d}}+\lambda V_{\omega })\in
\mathcal{B}(\ell ^{2}(\mathfrak{L}))$ to which the spectral theorem can be
applied. Indeed, for all $\beta \in \mathbb{R}^{+}$, $\omega \in \Omega $, $%
\lambda \in \mathbb{R}_{0}^{+}$, $t\in {\mathbb{R}}$ and $\alpha \in \lbrack
0,\beta ]$, one gets from (\ref{rescaledbis}), (\ref{Fermi statistic}) and (%
\ref{def.propagator}) that%
\begin{equation}
C_{t+i\alpha }^{(\omega )}(\mathbf{x})=\langle \mathfrak{e}_{x^{(2)}},%
\mathrm{e}^{-it\left( \Delta _{\mathrm{d}}+\lambda V_{\omega }\right)
}F_{\alpha }^{\beta }\left( \Delta _{\mathrm{d}}+\lambda V_{\omega }\right)
\mathfrak{e}_{x^{(1)}}\rangle \ ,  \label{correlation operator}
\end{equation}%
where $F_{\alpha }^{\beta }$ is the real function defined, for every $\beta
\in \mathbb{R}^{+}$ and $\alpha \in {\mathbb{R}}$, by
\begin{equation*}
F_{\alpha }^{\beta }\left( \varkappa \right) :=\frac{\mathrm{e}^{\alpha
\varkappa }}{1+\mathrm{e}^{\beta \varkappa }}\ ,\qquad \varkappa \in {%
\mathbb{R}}\ .
\end{equation*}%
Equation (\ref{correlation operator}) provides useful estimates like
space--decay properties of complex--time two--point correlation functions $%
C_{t+i\alpha }^{(\omega )}$, see \cite{OhmIII}. An important consequence of (%
\ref{correlation operator}) is the fact that the coefficient $\mathfrak{C}%
_{t+i\alpha }^{(\omega )}$ defined by (\ref{map cool}) can be seen as the
kernel (w.r.t. the canonical basis $\{\mathfrak{e}_{x}\otimes \mathfrak{e}%
_{x^\prime}\}_{x,x^{\prime }\in \mathfrak{L}}$) of a bounded operator on $%
\ell ^{2}(\mathfrak{L})\otimes \ell ^{2}(\mathfrak{L})$. This operator is
again denoted by $\mathfrak{C}_{t+i\alpha }^{(\omega )}$:

\begin{lemma}[$\mathfrak{C}_{t+i\protect\alpha }^{(\protect\omega )}$ as a
bounded operator]
\label{remark operator cool}\mbox{
}\newline
Let $\beta \in \mathbb{R}^{+}$, $\omega \in \Omega $, $\lambda \in \mathbb{R}%
_{0}^{+}$, $t\in \mathbb{R}$ and $\alpha \in \lbrack 0,\beta ]$. Then, there
is a unique bounded operator $\mathfrak{C}_{t+i\alpha }^{(\omega )}$ on $%
\ell ^{2}(\mathfrak{L})\otimes \ell ^{2}(\mathfrak{L})$ with
\begin{equation*}
\langle \mathfrak{e}_{x^{(1)}}\otimes \mathfrak{e}_{x^{(2)}}, \mathfrak{C}%
_{t+i\alpha }^{(\omega )} (\mathfrak{e}_{y^{(1)}}\otimes \mathfrak{e}%
_{y^{(2)}}) \rangle_{\ell ^{2}(\mathfrak{L})\otimes \ell ^{2}(\mathfrak{L})}
= \mathfrak{C}_{t+i\alpha }^{(\omega )}((x^{(1)},x^{(2)}),(y^{(1)},y^{(2)}))
\end{equation*}
for all $(x^{(1)},x^{(2)}),(y^{(1)},y^{(2)}) \in \mathfrak{L}^2$, and
\begin{equation*}
\Vert \mathfrak{C}_{t+i\alpha }^{(\omega )}\Vert _{\mathrm{op}}\leq 4\text{%
\qquad and\qquad }\underset{\alpha \rightarrow 0^{+}}{\lim }\Vert \mathfrak{C%
}_{i\alpha }^{(\omega )}-\mathfrak{C}_{0}^{(\omega )}\Vert _{\mathrm{op}}=0\
,
\end{equation*}%
where $\Vert \cdot \Vert _{\mathrm{op}}$ is the operator norm.
\end{lemma}

\begin{proof}
By (\ref{map cool}) and (\ref{correlation operator}), the bounded operator $%
\mathfrak{C}_{t+i\alpha }^{(\omega )}$ exists, is unique, and one directly
gets
\begin{equation*}
\frac{1}{4}\Vert \mathfrak{C}_{t+i\alpha }^{(\omega )}\Vert _{\mathrm{op}%
}\leq \left\Vert \frac{\mathrm{e}^{(-it+\alpha )(\Delta _{\mathrm{d}%
}+\lambda V_{\omega })}}{1+\mathrm{e}^{\beta (\Delta _{\mathrm{d}}+\lambda
V_{\omega })}}\right\Vert _{\mathrm{op}}\left\Vert \frac{\mathrm{e}%
^{(it+\beta -\alpha )(\Delta _{\mathrm{d}}+\lambda V_{\omega })}}{1+\mathrm{e%
}^{\beta (\Delta _{\mathrm{d}}+\lambda V_{\omega })}}\right\Vert _{\mathrm{op%
}}\leq 1
\end{equation*}%
for any $\beta \in \mathbb{R}^{+}$, $\omega \in \Omega $, $\lambda \in
\mathbb{R}_{0}^{+}$, $t\in \mathbb{R}$ and $\alpha \in \lbrack 0,\beta ]$.
Moreover, in the same way, (\ref{map cool}) and (\ref{correlation operator})
also lead to
\begin{equation}
\frac{1}{4}\Vert \mathfrak{C}_{i\alpha }^{(\omega )}-\mathfrak{C}%
_{0}^{(\omega )}\Vert _{\mathrm{op}}\leq \left\Vert \mathrm{e}^{\alpha
(\Delta _{\mathrm{d}}+\lambda V_{\omega })}-\mathbf{1}\right\Vert _{\mathrm{%
op}}+\left\Vert \mathrm{e}^{-\alpha (\Delta _{\mathrm{d}}+\lambda V_{\omega
})}-\mathbf{1}\right\Vert _{\mathrm{op}}  \label{easy1}
\end{equation}%
for any $\beta \in \mathbb{R}^{+}$, $\omega \in \Omega $, $\lambda \in
\mathbb{R}_{0}^{+}$, and $\alpha \in \lbrack 0,\beta ]$. Recall that the
self--adjoint operator $\Delta _{\mathrm{d}}+\lambda V_{\omega }$ is
bounded, i.e., $\Delta _{\mathrm{d}}+\lambda V_{\omega }\in \mathcal{B}(\ell
^{2}(\mathfrak{L}))$. It follows that the one--parameter group $\{\mathrm{e}%
^{\alpha (\Delta _{\mathrm{d}}+\lambda V_{\omega })}\}_{\alpha \in \mathbb{R}%
}$ is uniformly continuous (norm continuous). Therefore, the second
assertion is deduced from (\ref{easy1}) in the limit $\alpha \rightarrow
0^{+}$.
\end{proof}

\subsubsection{Space--Averaged Paramagnetic Transport Coefficients\label%
{Section Local AC--Conductivity0 copy(2)}}

Equation (\ref{average microscopic AC--conductivity}) and Lemma \ref{lemma
conductivty4 copy(1)} for $\mathcal{X}=\mathcal{U}$, $\tau =\tau ^{(\omega
,\lambda )}$ and $\varrho =\varrho ^{(\beta ,\omega ,\lambda )}$ yield%
\begin{equation}
\left\{ \Xi _{\mathrm{p},l}^{(\omega )}\left( t\right) \right\} _{k,q}=\frac{%
1}{\left\vert \Lambda _{l}\right\vert }\left[ (\mathbb{I}_{k,l},\tau
_{t}^{(\omega ,\lambda )}(\mathbb{I}_{q,l}))_{\sim }-(\mathbb{I}_{k,l},%
\mathbb{I}_{q,l})_{\sim }\right]  \label{inequality cool bog scalar}
\end{equation}%
for any $l,\beta \in \mathbb{R}^{+}$, $\omega \in \Omega $, $\lambda \in
\mathbb{R}_{0}^{+}$, $k,q\in \{1,\ldots ,d\}$ and $t\in \mathbb{R}$. Since $%
\Theta \left( I_{\mathbf{x}}\right) =-I_{\mathbf{x}}$ for any $\mathbf{x}\in
\mathfrak{L}^{2}$, by Theorem \ref{Thm important equality asymptotics
copy(1)}, the operator $\Xi _{\mathrm{p},l}^{(\omega )}\left( t\right) $ is
symmetric at any fixed time $t\in \mathbb{R}$ while the $\mathcal{B}(\mathbb{%
R}^{d})$--valued function $\Xi _{\mathrm{p},l}^{(\omega )}$ is symmetric
w.r.t. time--reversal. In other words,%
\begin{equation}
\left\{ \Xi _{\mathrm{p},l}^{(\omega )}\left( t\right) \right\}
_{k,q}=\left\{ \Xi _{\mathrm{p},l}^{(\omega )}\left( -t\right) \right\}
_{k,q}=\left\{ \Xi _{\mathrm{p},l}^{(\omega )}\left( t\right) \right\}
_{q,k}\in \mathbb{R}  \label{average time reversal symetry}
\end{equation}%
for any $l,\beta \in \mathbb{R}^{+}$, $\omega \in \Omega $, $\lambda \in
\mathbb{R}_{0}^{+}$, $k,q\in \{1,\ldots ,d\}$ and $t\in \mathbb{R}$.

Because of (\ref{inequality cool bog scalar}) it is convenient to use the
\emph{Duhamel GNS} (Gelfand-–Naimark–-Segal)) representation%
\begin{equation*}
(\mathcal{\tilde{H}},\tilde{\pi},\tilde{\Psi})\equiv (\mathcal{\tilde{H}}%
^{(\beta ,\omega ,\lambda )},\tilde{\pi}^{(\beta ,\omega ,\lambda )},\tilde{%
\Psi}^{(\beta ,\omega ,\lambda )})
\end{equation*}%
of the $(\tau ^{(\omega ,\lambda )},\beta )$--KMS state $\varrho ^{(\beta
,\omega ,\lambda )}$ for any $\beta \in \mathbb{R}^{+}$, $\omega \in \Omega $
and $\lambda \in \mathbb{R}_{0}^{+}$. See Definition \ref{GNS Duhamel} with $%
\mathcal{X}=\mathcal{U}$ and $\varrho =\varrho ^{(\beta ,\omega ,\lambda )}$%
. Note that we identify here the Duhamel two--point function defined by (\ref%
{def bogo jetman}) on the CAR algebra $\mathcal{U}$ with the scalar product $%
(\cdot ,\cdot )_{\sim }$ of the Hilbert space $\mathcal{\tilde{H}}$, see
Remark \ref{Reminder}. Other cyclic representations could be used
instead, but the Duhamel one makes the proofs involving the representation
of the paramagnetic conductivity as a spectral measure more transparent via
the results of \cite{Nau2}.

The CAR $C^{\ast }$--algebra $\mathcal{U}$ is the inductive limit of (finite
dimensional) simple $C^{\ast }$--algebras $\{\mathcal{U}_{\Lambda
}\}_{\Lambda \in \mathcal{P}_{f}(\mathfrak{L})}$, see \cite[Lemma IV.1.2]%
{simon}. By \cite[Corollary 2.6.19.]{BratteliRobinsonI}, $\mathcal{U}$ is
thus \emph{simple}. This property has some important consequences: The $%
(\tau ^{(\omega ,\lambda )},\beta )$--KMS state $\varrho ^{(\beta ,\omega
,\lambda )}$ is faithful. In particular, $\tilde{\pi}$ is injective. Remark
that $\tilde{\Psi}\equiv \mathbf{1}\in \mathcal{U}$ and $\mathcal{U}$ is a
dense set of $\mathcal{\tilde{H}}$, but $\tilde{\pi}\left( B\right) \tilde{%
\Psi}$ is generally not equal to $B\in \mathcal{U}$, in contrast to the
usual GNS representation. For this reason, we do not identify $\tilde{\pi}%
\left( \mathcal{U}\right) $ with $\mathcal{U}$. Moreover, by Theorem \ref%
{Thm important equality asymptotics} for $\mathcal{X}=\mathcal{U}$ and $%
\varrho =\varrho ^{(\beta ,\omega ,\lambda )}$, the $\ast $--automorphism
group $\tau =\tau ^{(\omega ,\lambda )}$ can be extended to a unitary group
on the whole Hilbert space $\mathcal{\tilde{H}}$:
\begin{equation}
\tau _{t}^{(\omega ,\lambda )}(B)=\mathrm{e}^{it\mathcal{\tilde{L}}}B\
,\qquad t\in \mathbb{R}\ ,\ B\in \mathcal{U}\subset \mathcal{\tilde{H}}\ ,
\label{inequality really cool}
\end{equation}%
with $\mathcal{\tilde{L}}\equiv \mathcal{\tilde{L}}^{(\beta ,\omega ,\lambda
)}$ being a self--adjoint operator acting on $\mathcal{\tilde{H}}$. The
domain of $\mathcal{\tilde{L}}$ includes the domain of the generator $\delta
^{(\omega ,\lambda )}$ of the one--parameter group $\tau ^{(\omega ,\lambda
)}$, i.e., $\mathrm{Dom}(\mathcal{\tilde{L}})\supset \mathrm{Dom}(\delta
^{(\omega ,\lambda )})$, while
\begin{equation}
\mathcal{\tilde{L}}\left( B\right) =-i\delta ^{(\omega ,\lambda )}\left(
B\right) \ ,\qquad B\in \mathrm{Dom}(\delta ^{(\omega ,\lambda )})\subset
\mathcal{U}\subset \mathcal{\tilde{H}}\ .  \label{domain delata2}
\end{equation}%
Equation (\ref{inequality really cool}) is an important representation of
the dynamics because we can deduce from (\ref{inequality cool bog scalar})
the existence of the paramagnetic conductivity measure from the spectral
theorem.

To present this result, recall that $\mathcal{B}_{+}(\mathbb{R}^{d})\subset
\mathcal{B}(\mathbb{R}^{d})$ denotes the set of positive linear operators on
$\mathbb{R}^{d}$ and any $\mathcal{B}(\mathbb{R}^{d})$--valued measure $\mu $
on $\mathbb{R}$ is symmetric iff $\mu (\mathcal{X})=\mu (-\mathcal{X})$ for
any Borel set $\mathcal{X}\subset \mathbb{R}$. Then, we derive the
paramagnetic conductivity measure:

\begin{satz}[Conductivity measures as spectral measures]
\label{lemma sigma pos type copy(3)}\mbox{
}\newline
For any $l,\beta \in \mathbb{R}^{+}$, $\omega \in \Omega $ and $\lambda \in
\mathbb{R}_{0}^{+}$, there exists a finite symmetric $\mathcal{B}_{+}(%
\mathbb{R}^{d})$--valued measure $\mu _{\mathrm{p},l}^{(\omega )}\equiv \mu
_{\mathrm{p},l}^{(\beta ,\omega ,\lambda )}$ on $\mathbb{R}$ such that%
\begin{equation}
\Xi _{\mathrm{p},l}^{(\omega )}(t)=\int_{\mathbb{R}}\left( \cos \left( t\nu
\right) -1\right) \mu _{\mathrm{p},l}^{(\omega )}(\mathrm{d}\nu )\ ,\qquad
t\in \mathbb{R}\ .  \label{cosinus su}
\end{equation}
\end{satz}

\begin{proof}
Fix $l,\beta \in \mathbb{R}^{+}$, $\omega \in \Omega $ and $\lambda \in
\mathbb{R}_{0}^{+}$. Let $\tilde{E}\equiv \tilde{E}^{(\beta ,\omega ,\lambda
)}$ be the (projection--valued) spectral measure of the self--adjoint
operator $\mathcal{\tilde{L}}$. Then, by combining (\ref{inequality cool bog
scalar})--(\ref{average time reversal symetry}) with (\ref{inequality really
cool}), we directly arrive at the equality
\begin{eqnarray}
\left\{ \Xi _{\mathrm{p},l}^{(\omega )}\left( t\right) \right\} _{k,q} &=&%
\frac{1}{4\left\vert \Lambda _{l}\right\vert }\int_{\mathbb{R}}\left(
\mathrm{e}^{it\nu }-1\right) (\mathbb{I}_{k,l},\tilde{E}(\mathrm{d}\nu )%
\mathbb{I}_{q,l})_{\sim }  \notag \\
&&+\frac{1}{4\left\vert \Lambda _{l}\right\vert }\int_{\mathbb{R}}\left(
\mathrm{e}^{it\nu }-1\right) (\mathbb{I}_{q,l},\tilde{E}(\mathrm{d}\nu )%
\mathbb{I}_{k,l})_{\sim }  \notag \\
&&+\frac{1}{4\left\vert \Lambda _{l}\right\vert }\int_{\mathbb{R}}\left(
\mathrm{e}^{-it\nu }-1\right) (\mathbb{I}_{k,l},\tilde{E}(\mathrm{d}\nu )%
\mathbb{I}_{q,l})_{\sim }  \notag \\
&&+\frac{1}{4\left\vert \Lambda _{l}\right\vert }\int_{\mathbb{R}}\left(
\mathrm{e}^{-it\nu }-1\right) (\mathbb{I}_{q,l},\tilde{E}(\mathrm{d}\nu )%
\mathbb{I}_{k,l})_{\sim }  \label{abama s2}
\end{eqnarray}%
for any $k,q\in \{1,\ldots ,d\}$ and $t\in \mathbb{R}$. Note that, for any
Borel set $\mathcal{X}\subset \mathbb{R}$ and all $k,q\in \{1,\ldots ,d\}$,
\begin{equation*}
(\mathbb{I}_{k,l},\tilde{E}\left( \mathcal{X}\right) \mathbb{I}_{q,l})_{\sim
}+(\mathbb{I}_{q,l},\tilde{E}\left( \mathcal{X}\right) \mathbb{I}%
_{k,l})_{\sim }\in \mathbb{R}\ .
\end{equation*}%
Thus, define the $\mathcal{B}(\mathbb{R}^{d})$--valued measure $\mu _{%
\mathrm{p},l}^{(\omega )}$ by%
\begin{eqnarray}
\left\langle \vec{u},\mu _{\mathrm{p},l}^{(\omega )}\left( \mathcal{X}%
\right) \vec{w}\right\rangle &=&\frac{1}{4\left\vert \Lambda _{l}\right\vert
}\underset{k,q\in \{1,\ldots ,d\}}{\sum }u_{k}w_{q}(\mathbb{I}_{k,l},\tilde{E%
}\left( \mathcal{X}\right) \mathbb{I}_{q,l})_{\sim }  \notag \\
&&+\frac{1}{4\left\vert \Lambda _{l}\right\vert }\underset{k,q\in \{1,\ldots
,d\}}{\sum }u_{k}w_{q}(\mathbb{I}_{q,l},\tilde{E}\left( \mathcal{X}\right)
\mathbb{I}_{k,l})_{\sim }  \notag \\
&&+\frac{1}{4\left\vert \Lambda _{l}\right\vert }\underset{k,q\in \{1,\ldots
,d\}}{\sum }u_{k}w_{q}(\mathbb{I}_{k,l},\tilde{E}\left( -\mathcal{X}\right)
\mathbb{I}_{q,l})_{\sim }  \notag \\
&&+\frac{1}{4\left\vert \Lambda _{l}\right\vert }\underset{k,q\in \{1,\ldots
,d\}}{\sum }u_{k}w_{q}(\mathbb{I}_{q,l},\tilde{E}\left( -\mathcal{X}\right)
\mathbb{I}_{k,l})_{\sim }  \label{von braun}
\end{eqnarray}%
for any $\vec{u}:=(u_{1},\ldots ,u_{d})\in \mathbb{R}^{d}$, $\vec{w}%
:=(w_{1},\ldots ,w_{d})\in \mathbb{R}^{d}$ and all Borel sets $\mathcal{X}%
\subset \mathbb{R}$. Here, $\left\langle \cdot ,\cdot \right\rangle $
denotes the usual scalar product of $\mathbb{R}^{d}$. Obviously, by
construction,
\begin{equation*}
\left\langle \vec{u},\mu _{\mathrm{p},l}^{(\omega )}\left( \mathcal{X}%
\right) \vec{w}\right\rangle =\left\langle \vec{w},\mu _{\mathrm{p}%
,l}^{(\omega )}\left( \mathcal{X}\right) \vec{u}\right\rangle \qquad \text{%
and}\qquad \left\langle \vec{w},\mu _{\mathrm{p},l}^{(\omega )}\left(
\mathcal{X}\right) \vec{w}\right\rangle \geq 0\ ,
\end{equation*}%
for any $\vec{u},\vec{w}\in \mathbb{R}^{d}$ and all Borel sets $\mathcal{X}%
\subset \mathbb{R}$. Moreover, $\mu _{\mathrm{p},l}^{(\omega )}$ is a
symmetric measure and, by (\ref{abama s2}), we obtain Equation (\ref{cosinus
su}).
\end{proof}

For any $\beta \in \mathbb{R}^{+}$, $\omega \in \Omega $ and $\lambda \in
\mathbb{R}_{0}^{+}$, it is useful at this point to also consider the GNS
representation
\begin{equation*}
(\mathcal{H},\pi ,\Psi )\equiv (\mathcal{H}^{(\beta ,\omega ,\lambda )},\pi
^{(\beta ,\omega ,\lambda )},\Psi ^{(\beta ,\omega ,\lambda )})
\end{equation*}%
of the $(\tau ^{(\omega ,\lambda )},\beta )$--KMS state $\varrho ^{(\beta
,\omega ,\lambda )}$ and to describe its relation to the Duhamel GNS
representation. To this end, we denote by $\mathcal{L}\equiv \mathcal{L}%
^{(\beta ,\omega ,\lambda )}$ the standard Liouvillean of the system under
consideration, i.e., the self--adjoint operator acting on $\mathcal{H}$
which implements the dynamics as
\begin{equation}
\pi \left( \tau _{t}\left( B\right) \right) =\mathrm{e}^{it\mathcal{L}}\pi
\left( B\right) \mathrm{e}^{-it\mathcal{L}}\ ,\qquad t\in \mathbb{R},\ B\in
\mathcal{U}\ ,  \label{dynamic}
\end{equation}%
with $\mathcal{L}\Psi =\Psi $. Let $E\equiv E^{(\beta ,\omega ,\lambda )}$
be the (projection--valued) spectral measure of $\mathcal{L}$. We also use
the (Tomita--Takesaki) modular objects
\begin{equation*}
\Delta \equiv \Delta ^{(\beta ,\omega ,\lambda )}:=\mathrm{e}^{-\beta
\mathcal{L}}\ ,\qquad J\equiv J^{(\beta ,\omega ,\lambda )}\ ,
\end{equation*}%
of the pair $(\pi \left( \mathcal{U}\right) ^{\prime \prime },\Psi )$.

Theorem \ref{toto fluctbis copy(2)} says that
\begin{equation}
(B_{1},B_{2})_{\sim }=\left\langle \mathfrak{T}\pi \left( B_{1}\right) \Psi ,%
\mathfrak{T}\pi \left( B_{2}\right) \Psi \right\rangle _{\mathcal{H}}\
,\qquad B_{1},B_{2}\in \mathcal{U}\ ,  \label{barck cool1}
\end{equation}%
where $\mathfrak{T}\equiv \mathfrak{T}^{(\beta ,\omega ,\lambda )}$ is the
operator defined by (\ref{operator bogoliubov}) for $\tau =\tau ^{(\omega
,\lambda )}$ and $\varrho =\varrho ^{(\beta ,\omega ,\lambda )}$, that is,
\begin{equation}
\mathfrak{T}:=\beta ^{1/2}\left( \frac{1-\mathrm{e}^{-\beta \mathcal{L}}}{%
\beta \mathcal{L}}\right) ^{1/2}\ .  \label{eq barack s0}
\end{equation}%
Note that $\mathfrak{T}$ is unbounded, but
\begin{equation}
\pi \left( \mathcal{U}\right) \Psi \subset \mathrm{Dom}(\Delta
^{1/2})\subset \mathrm{Dom}(\mathfrak{T})\ .  \label{barck cool2}
\end{equation}%
The $\mathcal{B}_{+}(\mathbb{R}^{d})$--valued measure $\mu _{\mathrm{p}%
,l}^{(\omega )}$ of Theorem \ref{lemma sigma pos type copy(3)}, which is
defined by (\ref{von braun}), can also be studied via (\ref{barck cool1}).
Indeed, (\ref{barck cool1}) and (\ref{barck cool2}) together with Theorem %
\ref{toto fluctbis copy(4)} and (\ref{toto oublie}) imply that
\begin{equation}
(\mathbb{I}_{k,l},\tilde{E}\left( \mathcal{X}\right) \mathbb{I}_{q,l})_{\sim
}=\left\langle \mathfrak{T}E\left( \mathcal{X}\right) \pi \left( \mathbb{I}%
_{k,l}\right) \Psi ,\mathfrak{T}E\left( \mathcal{X}\right) \pi \left(
\mathbb{I}_{q,l}\right) \Psi \right\rangle _{\mathcal{H}}
\label{eq barack s1}
\end{equation}%
for any $l,\beta \in \mathbb{R}^{+}$, $\omega \in \Omega $, $\lambda \in
\mathbb{R}_{0}^{+}$, $k,q\in \{1,\ldots ,d\}$ and any Borel set $\mathcal{X}%
\subset \mathbb{R}$. The existence of the first moment of $\mu _{\mathrm{p}%
,l}^{(\omega )}$ is a direct consequence of the above equation.

To see this, recall that $\Vert \mu _{\mathrm{p},l}^{(\omega )}\Vert _{%
\mathrm{op}}$ is the measure on $\mathbb{R}$ taking values in $\mathbb{R}%
_{0}^{+}$ that is defined, for any Borel set $\mathcal{X}\subset \mathbb{R}$
and $\mu =\mu _{\mathrm{p},l}^{(\omega )}$, by (\ref{definion opera measure}%
). Then, one gets the following assertions:

\begin{satz}[Existence of the first moment of $\protect\mu _{\mathrm{p},l}^{(%
\protect\omega )}$]
\label{lemma sigma pos type copy(1)}\mbox{
}\newline
For any $l,\beta \in \mathbb{R}^{+}$, $\omega \in \Omega $ and $\lambda \in
\mathbb{R}_{0}^{+}$, the $\mathcal{B}_{+}(\mathbb{R}^{d})$--valued measure $%
\mu _{\mathrm{p},l}^{(\omega )}$ of Theorem \ref{lemma sigma pos type
copy(3)} satisfies the following bounds:%
\begin{eqnarray*}
\int_{\mathbb{R}}\Vert \mu _{\mathrm{p},l}^{(\omega )}\Vert _{\mathrm{op}}(%
\mathrm{d}\nu ) &\leq &\frac{1}{\left\vert \Lambda _{l}\right\vert }\underset%
{k=1}{\overset{d}{\sum }}\varrho ^{(\beta ,\omega ,\lambda )}\left( \mathbb{I%
}_{k,l}^{2}\right) \ , \\
\int_{\mathbb{R}}\left\vert \nu \right\vert \Vert \mu _{\mathrm{p}%
,l}^{(\omega )}\Vert _{\mathrm{op}}(\mathrm{d}\nu ) &\leq &\frac{2}{%
\left\vert \Lambda _{l}\right\vert }\underset{k=1}{\overset{d}{\sum }}%
\varrho ^{(\beta ,\omega ,\lambda )}\left( \mathbb{I}_{k,l}^{2}\right) \ , \\
\int_{\mathbb{R}}\left\vert \nu \right\vert \Vert \mu _{\mathrm{p}%
,l}^{(\omega )}\Vert _{\mathrm{op}}(\mathrm{d}\nu ) &\leq &\frac{2}{%
\left\vert \Lambda _{l}\right\vert }\underset{k=1}{\overset{d}{\sum }}\sqrt{%
\varrho ^{(\beta ,\omega ,\lambda )}\left( \mathbb{I}_{k,l}^{2}\right) }%
\sqrt{\varrho ^{(\beta ,\omega ,\lambda )}\left( \left( \delta ^{(\omega
,\lambda )}\left( \mathbb{I}_{k,l}\right) \right) ^{2}\right) }\ .
\end{eqnarray*}
\end{satz}

\begin{proof}
Fix $l,\beta \in \mathbb{R}^{+}$, $\omega \in \Omega $ and $\lambda \in
\mathbb{R}_{0}^{+}$. By positivity of the measure $\mu _{\mathrm{p}%
,l}^{(\omega )}$ and linearity of the trace,
\begin{equation*}
\Vert \mu _{\mathrm{p},l}^{(\omega )}\Vert _{\mathrm{op}}\left( \mathcal{X}%
\right) \leq \mathrm{Trace}_{\mathcal{B}(\mathbb{R}^{d})}\left( \mu _{%
\mathrm{p},l}^{(\omega )}\left( \mathcal{X}\right) \right)
\end{equation*}%
for any Borel set $\mathcal{X}\subset \mathbb{R}$. This implies that
\begin{equation*}
\int_{\mathbb{R}}\Vert \mu _{\mathrm{p},l}^{(\omega )}\Vert _{\mathrm{op}}(%
\mathrm{d}\nu )\leq \mathrm{Trace}_{\mathcal{B}(\mathbb{R}^{d})}\left( \int_{%
\mathbb{R}}\mu _{\mathrm{p},l}^{(\omega )}(\mathrm{d}\nu )\right)
\end{equation*}%
and
\begin{equation*}
\int_{\mathbb{R}}\left\vert \nu \right\vert \Vert \mu _{\mathrm{p}%
,l}^{(\omega )}\Vert _{\mathrm{op}}(\mathrm{d}\nu )\leq \mathrm{Trace}_{%
\mathcal{B}(\mathbb{R}^{d})}\left( \int_{\mathbb{R}}\left\vert \nu
\right\vert \mu _{\mathrm{p},l}^{(\omega )}(\mathrm{d}\nu )\right) \ .
\end{equation*}%
Hence, by (\ref{von braun}), it suffices to prove that%
\begin{eqnarray}
\int_{\mathbb{R}}(\mathbb{I}_{k,l},\tilde{E}(\mathrm{d}\nu )\mathbb{I}%
_{k,l})_{\sim } &\leq &\varrho ^{(\beta ,\omega ,\lambda )}\left( \mathbb{I}%
_{k,l}^{2}\right) \ ,  \label{ea q a la con0} \\
\int_{\mathbb{R}}\left\vert \nu \right\vert (\mathbb{I}_{k,l},\tilde{E}(%
\mathrm{d}\nu )\mathbb{I}_{k,l})_{\sim } &\leq &2\varrho ^{(\beta ,\omega
,\lambda )}\left( \mathbb{I}_{k,l}^{2}\right) \ ,  \label{ea q a la con} \\
\int_{\mathbb{R}}\left\vert \nu \right\vert (\mathbb{I}_{k,l},\tilde{E}(%
\mathrm{d}\nu )\mathbb{I}_{k,l})_{\sim } &\leq &2\sqrt{\varrho ^{(\beta
,\omega ,\lambda )}\left( \mathbb{I}_{k,l}^{2}\right) \varrho ^{(\beta
,\omega ,\lambda )}\left( \left( \delta ^{(\omega ,\lambda )}\left( \mathbb{I%
}_{k,l}\right) \right) ^{2}\right) }\ ,  \notag \\
&&  \label{ea q a la con2}
\end{eqnarray}%
for any $k\in \{1,\ldots ,d\}$.

Inequality (\ref{ea q a la con0}) is a direct consequence of Theorem \ref%
{thm auto--correlation upper bounds}. The second upper bound is derived as
follows: Fix $k\in \{1,\ldots ,d\}$.\ We infer from (\ref{eq barack s0}) and
(\ref{eq barack s1}) that
\begin{eqnarray}
\int_{\mathbb{R}}\left\vert \nu \right\vert (\mathbb{I}_{k,l},\tilde{E}(%
\mathrm{d}\nu )\mathbb{I}_{k,l})_{\sim } &=&\left\Vert \left( 1-\mathrm{e}%
^{-\beta \mathcal{L}}\right) ^{1/2}E\left( \mathbb{R}_{0}^{+}\right) \pi
\left( \mathbb{I}_{k,l}\right) \Psi \right\Vert _{\mathcal{H}}^{2}
\label{barck s-tomita1bis} \\
&&+\left\Vert \left( \mathrm{e}^{-\beta \mathcal{L}}-1\right) ^{1/2}E\left(
\mathbb{R}^{-}\right) \pi \left( \mathbb{I}_{k,l}\right) \Psi \right\Vert _{%
\mathcal{H}}^{2}\ .  \notag
\end{eqnarray}%
Clearly, one has the upper bound
\begin{equation}
\left\Vert \left( 1-\mathrm{e}^{-\beta \mathcal{L}}\right) ^{1/2}E\left(
\mathbb{R}_{0}^{+}\right) \pi \left( \mathbb{I}_{k,l}\right) \Psi
\right\Vert _{\mathcal{H}}^{2}\leq \varrho ^{(\beta ,\omega ,\lambda
)}\left( \mathbb{I}_{k,l}^{2}\right) \ ,  \label{barck s-tomita2}
\end{equation}%
while
\begin{equation}
\left\Vert \left( \mathrm{e}^{-\beta \mathcal{L}}-1\right) ^{1/2}E\left(
\mathbb{R}^{-}\right) \pi \left( \mathbb{I}_{k,l}\right) \Psi \right\Vert _{%
\mathcal{H}}^{2}\leq \left\Vert \Delta ^{1/2}\pi \left( \mathbb{I}%
_{k,l}\right) \Psi \right\Vert _{\mathcal{H}}^{2}\ ,  \label{barck s-tomita3}
\end{equation}%
with $\Delta :=\mathrm{e}^{-\beta \mathcal{L}}$ being the modular operator.
Using now the anti--unitarity of $J$, $J^{2}=\mathbf{1}$ and
\begin{equation*}
J\Delta ^{1/2}\pi \left( \mathbb{I}_{k,l}\right) \Psi =\pi \left( \mathbb{I}%
_{k,l}\right) ^{\ast }\Psi =\pi \left( \mathbb{I}_{k,l}\right) \Psi \ ,
\end{equation*}%
one gets that%
\begin{equation}
\left\Vert \Delta ^{1/2}\pi \left( \mathbb{I}_{k,l}\right) \Psi \right\Vert
_{\mathcal{H}}^{2}=\left\Vert \pi \left( \mathbb{I}_{k,l}\right) \Psi
\right\Vert _{\mathcal{H}}^{2}=\varrho ^{(\beta ,\omega ,\lambda )}\left(
\mathbb{I}_{k,l}^{2}\right) \ .  \label{barck s-tomita5}
\end{equation}%
Therefore, by combining Equation (\ref{barck s-tomita1bis}) with (\ref{barck
s-tomita2})--(\ref{barck s-tomita5}) we arrive at Inequality (\ref{ea q a la
con}).

Finally, to prove (\ref{ea q a la con2}), observe that
\begin{eqnarray}
\int_{\mathbb{R}}\left\vert \nu \right\vert (\mathbb{I}_{k,l},\tilde{E}(%
\mathrm{d}\nu )\mathbb{I}_{k,l})_{\sim } &=&\left\langle \mathfrak{T}\pi
\left( \mathbb{I}_{k,l}\right) \Psi ,E\left( \mathbb{R}_{0}^{+}\right)
\mathfrak{T}\mathcal{L}\pi \left( \mathbb{I}_{k,l}\right) \Psi \right\rangle
_{\mathcal{H}}  \label{bark s1} \\
&&-\left\langle \mathfrak{T}\pi \left( \mathbb{I}_{k,l}\right) \Psi ,E\left(
\mathbb{R}^{-}\right) \mathfrak{T}\mathcal{L}\pi \left( \mathbb{I}%
_{k,l}\right) \Psi \right\rangle _{\mathcal{H}}\ .  \notag
\end{eqnarray}%
Since $\mathbb{I}_{k,l}\in \mathcal{U}_{0}\subset \mathrm{Dom}(\delta
^{(\omega ,\lambda )})$,%
\begin{equation}
\mathcal{L}\pi \left( \mathbb{I}_{k,l}\right) \Psi =-i\pi \left( \delta
^{(\omega ,\lambda )}\left( \mathbb{I}_{k,l}\right) \right) \Psi \ ,
\label{bark s2}
\end{equation}%
see (\ref{dynamic}). Therefore, by additionally using the Cauchy--Schwarz
inequality of $(\mathbb{\cdot },\mathbb{\cdot })_{\sim }$ and Theorem \ref%
{thm auto--correlation upper bounds}, one gets (\ref{ea q a la con2})
similarly as above.
\end{proof}

Equation (\ref{eq barack s1}) also leads to a characterization of the
non--triviality of the conductivity measure at non--zero frequencies via a
geometric condition:

\begin{satz}[Geometric interpretation of the AC--conductivity measure]
\label{lemma sigma pos type copy(2)}\mbox{
}\newline
Let $l,\beta \in \mathbb{R}^{+}$, $\omega \in \Omega $ and $\lambda \in
\mathbb{R}_{0}^{+}$. Then,
\begin{equation*}
\mathrm{lin}\left\{ \pi \left( \mathbb{I}_{k,l}\right) \Psi :k\in \{1,\ldots
,d\}\right\} \subset \ker \left( \mathcal{L}\right) \qquad \text{iff}\qquad
\mu _{\mathrm{p},l}^{(\omega )}\left( \mathbb{R}\backslash \{0\}\right) =0\ .
\end{equation*}%
Here, $\mathrm{lin}$ stands for the linear hull of some set.
\end{satz}

\begin{proof}
Fix $l,\beta \in \mathbb{R}^{+}$, $\omega \in \Omega $ and $\lambda \in
\mathbb{R}_{0}^{+}$. If
\begin{equation*}
\mathrm{lin}\left\{ \pi \left( \mathbb{I}_{k,l}\right) \Psi :k\in \{1,\ldots
,d\}\right\} \subset \ker \left( \mathcal{L}\right) \ ,
\end{equation*}%
then we infer from (\ref{von braun}) and (\ref{eq barack s1}) that $\mu _{%
\mathrm{p},l}^{(\omega )}\left( \mathbb{R}\backslash \{0\}\right) =0$.
Observe that $\mathfrak{T}$ acts as the identity on the kernel of $\mathcal{L%
}$. Assume now that $\mu _{\mathrm{p},l}^{(\omega )}\left( \mathbb{R}%
\backslash \{0\}\right) =0$. Then,
\begin{equation*}
\mu _{\mathrm{p},l}^{(\omega )}\left( \mathbb{R}\backslash \{0\}\right) =0\ ,
\end{equation*}
which, by (\ref{von braun}) for $\mathcal{X}=\mathbb{R}\backslash \{0\}$,
implies that%
\begin{equation*}
(\mathbb{I}_{k,l},\tilde{E}\left( \mathbb{R}\backslash \{0\}\right) \mathbb{I%
}_{k,l})_{\sim }=0\ ,\qquad k\in \{1,\ldots ,d\}\ .
\end{equation*}%
As a consequence, any linear combination of elements $\{\mathbb{I}%
_{k,l}\}_{k\in \{1,\ldots ,d\}}\in \mathcal{U}\subset \mathcal{\tilde{H}}$
belongs to the kernel of $\mathcal{\tilde{L}}$, i.e.,
\begin{equation*}
\mathrm{lin}\left\{ \mathbb{I}_{k,l}:k\in \{1,\ldots ,d\}\right\} \subset
\ker (\mathcal{\tilde{L}})\ .
\end{equation*}%
By Theorem \ref{toto fluctbis copy(4)} and (\ref{toto oublie}), this
property in turn yields
\begin{equation*}
\mathrm{lin}\left\{ \pi \left( \mathbb{I}_{k,l}\right) \Psi :k\in \{1,\ldots
,d\}\right\} \subset \ker \left( \mathcal{L}\right) \ .
\end{equation*}
\end{proof}

\begin{koro}[Non--triviality of the measure $\protect\mu _{\mathrm{p},l}^{(%
\protect\omega )}$]
\label{Corollary Stationarity copy(1)}\mbox{ }\newline
For any $l,\beta \in \mathbb{R}^{+}$, $\omega \in \Omega $ and $\lambda \in
\mathbb{R}_{0}^{+}$, the $\mathcal{B}_{+}(\mathbb{R}^{d})$--valued measure $%
\mu _{\mathrm{p},l}^{(\omega )}$ of Theorem \ref{lemma sigma pos type
copy(3)} satisfies $\mu _{\mathrm{p},l}^{(\omega )}\left( \mathbb{R}%
\backslash \{0\}\right) >0$.
\end{koro}

\begin{proof}
By explicit computations, for any $k\in \{1,\ldots ,d\}$,
\begin{equation}
\delta ^{(\omega ,\lambda )}\left( \mathbb{I}_{k,l}\right) =\lambda \mathbb{%
\mathbb{A}}_{k,l}^{(\omega )}+\mathbb{B}_{k,l}\ ,  \label{tho super2}
\end{equation}%
where $\mathbb{\mathbb{A}}_{k,l}^{(\omega )},\mathbb{\mathbb{B}}_{k,l}\in
\mathbb{\mathcal{U}}$ are defined, for $\omega \in \Omega $ and $l\in
\mathbb{R}^{+}$, by
\begin{equation}
\mathbb{\mathbb{A}}_{k,l}^{(\omega )}:=\underset{x\in \Lambda _{l}}{\sum }%
\left( V_{\omega }\left( x+e_{k}\right) -V_{\omega }\left( x\right) \right)
P_{(x,x+e_{k})}  \notag  \label{tho super3}
\end{equation}%
and%
\begin{multline}
\mathbb{B}_{k,l}:=\underset{x,z\in \mathfrak{L},|z|=1,z\neq \pm e_{k}}{\sum }%
\left( \mathbf{1}\left[ x\in \left( \Lambda _{l}+z\right) \backslash \Lambda
_{l}\right] -\mathbf{1}\left[ x\in \Lambda _{l}\backslash \left( \Lambda
_{l}+z\right) \right] \right) P_{(x,x+e_{k}+z)}  \notag \\
+\underset{x\in \mathfrak{L}}{\sum }\left( \mathbf{1}\left[ x\in \left(
\Lambda _{l}+e_{k}\right) \backslash \Lambda _{l}\right] -\mathbf{1}\left[
x\in \Lambda _{l}\backslash \left( \Lambda _{l}+e_{k}\right) \right] \right)
\left( 2a_{x}^{\ast }a_{x}-P_{(x+e_{k},x-e_{k})}\right)  \label{tho super4}
\end{multline}%
with $P_{(x,y)}$ being defined by (\ref{R x}) for any $x,y\in \mathfrak{L}$.
In particular, $\delta ^{(\omega ,\lambda )}\left( \mathbb{I}_{k,l}\right) $
is not zero and hence $\pi \left( \mathbb{I}_{k,l}\right) \Psi \notin \ker
\left( \mathcal{L}\right) $, because $\pi$ is injective and the cyclic
vector $\Psi $ is separating for $\pi \left( \mathcal{U}\right) ^{\prime
\prime }$, see \cite[Corollary 5.3.9.]{BratteliRobinson}. Therefore, the
assertion is a direct consequence of Theorem \ref{lemma sigma pos type
copy(2)}.
\end{proof}

We now give another construction of the (AC--conductivity) measure $\mu _{%
\mathrm{p},l}^{(\omega )}$ on $\mathbb{R}\backslash \{0\}$ from the
diamagnetic transport coefficient $\Xi _{\mathrm{d},l}^{(\omega )}$ (\ref%
{average microscopic AC--conductivity dia}) and the space--averaged quantum
current viscosity%
\begin{equation*}
t\mapsto \mathbf{V}_{l}^{(\omega )}\left( t\right) \equiv \mathbf{V}%
_{l}^{(\beta ,\omega ,\lambda )}\left( t\right) \in \mathcal{B}(\mathbb{R}%
^{d})\ ,
\end{equation*}
see (\ref{quantum viscosity bis bis}). W.r.t. the canonical orthonormal
basis of $\mathbb{R}^{d}$,%
\begin{equation}
\left\{ \mathbf{V}_{l}^{(\omega )}\left( t\right) \right\} _{k,q}=\frac{1}{%
\varrho ^{(\beta ,\omega ,\lambda )}\left( \mathbb{P}_{k,l}\right) }\varrho
^{(\beta ,\omega ,\lambda )}\left( i[\mathbb{I}_{k,l},\tau _{t}^{(\omega
,\lambda )}(\mathbb{I}_{q,l})]\right)  \label{quantum viscositybis}
\end{equation}%
for any $k,q\in \{1,\ldots ,d\}$ and $t\in \mathbb{R}$. Compare (\ref%
{quantum viscositybis}) with (\ref{quantum viscosity}). Its Laplace
transform
\begin{equation*}
\mathbf{L}[\mathbf{V}_{l}^{(\omega )}](\epsilon ):=\int_{0}^{\infty }\mathrm{%
e}^{-\epsilon s}\mathbf{V}_{l}^{(\omega )}\left( s\right) \mathrm{d}s
\end{equation*}%
exists for all $\epsilon \in \mathbb{R}^{+}$, by the boundedness of $\mathbf{%
V}_{l}^{(\omega )}$. In fact, one has:

\begin{satz}[Static admittance]
\label{lemma sigma pos type copy(8)}\mbox{
}\newline
Let $l,\beta \in \mathbb{R}^{+}$, $\omega \in \Omega $ and $\lambda \in
\mathbb{R}_{0}^{+}$. Then the limit of $\mathbf{L} \lbrack \mathbf{V}%
_{l}^{(\omega )}](\epsilon )$ exists as $\epsilon \downarrow 0$ and
satisfies:
\begin{equation*}
\Xi _{\mathrm{d},l}^{(\omega )}\,\underset{\epsilon \downarrow 0}{\lim }\,%
\mathbf{L}[\mathbf{V}_{l}^{(\omega )}](\epsilon )=\mu _{\mathrm{p}%
,l}^{(\omega )}\left( \mathbb{R}\backslash \{0\}\right) =\frac{1}{\left\vert
\Lambda _{l}\right\vert }\left\{ (\mathbb{I}_{k,l},\tilde{E}\left( \mathbb{R}%
\backslash \{0\}\right) \mathbb{I}_{q,l})_{\sim }^{(\omega )}\right\}
_{k,q\in \{1,\ldots ,d\}}
\end{equation*}%
Note that $\tilde{E}\left( \mathbb{R}\backslash \{0\}\right) $ is not the
identity because $\mathcal{\tilde{L}}\mathbf{1}=0$.
\end{satz}

\begin{proof}
Fix $l,\beta \in \mathbb{R}^{+}$, $\omega \in \Omega $ and $\lambda \in
\mathbb{R}_{0}^{+}$. By \cite[Theorems III.3-III.4]{Nau2}, observe that
\begin{equation*}
\Xi _{\mathrm{d},l}^{(\omega )}\,\underset{\epsilon \downarrow 0}{\lim }\,%
\mathbf{L}[\mathbf{V}_{l}^{(\omega )}](\epsilon )=\frac{1}{\left\vert
\Lambda _{l}\right\vert }\left\{ (\mathbb{I}_{k,l},\tilde{E}\left( \mathbb{R}%
\backslash \{0\}\right) \mathbb{I}_{q,l})_{\sim }\right\} _{k,q\in
\{1,\ldots ,d\}}\ .
\end{equation*}%
On the other hand, by (\ref{inequality cool bog scalar}) and (\ref%
{inequality really cool}),%
\begin{equation}
\frac{1}{t}\int_{0}^{t}\left\{ \Xi _{\mathrm{p},l}^{(\omega )}\left(
s\right) \right\} _{k,q}\mathrm{d}s=\frac{1}{t\left\vert \Lambda
_{l}\right\vert }\int_{0}^{t}(\mathbb{I}_{k,l},\mathrm{e}^{it\mathcal{\tilde{%
L}}}\mathbb{I}_{q,l})_{\sim }\mathrm{d}s-\frac{1}{\left\vert \Lambda
_{l}\right\vert }(\mathbb{I}_{k,l},\mathbb{I}_{q,l})_{\sim }
\label{inequality cool2+0}
\end{equation}%
for any $t\in \mathbb{R}^{+}$ and $k,q\in \{1,\ldots ,d\}$. The von Neumann
or mean ergodic theorem (see, e.g., \cite[Theorem 3.13]{AttalJoyePillet2006a}%
) implies that
\begin{equation}
\underset{t\rightarrow \infty }{\lim }\ \frac{1}{t}\int_{0}^{t}(\mathbb{I}%
_{k,l},\mathrm{e}^{it\mathcal{\tilde{L}}}\mathbb{I}_{q,l})_{\sim }\mathrm{d}%
s=(\mathbb{I}_{k,l},\tilde{E}\left( \{0\}\right) \mathbb{I}_{q,l})_{\sim }\ ,
\label{inequality cool2+2}
\end{equation}%
where $\tilde{E}\left( \{0\}\right) $ is the orthogonal projection on the
kernel of $\mathcal{\tilde{L}}$. By combining (\ref{inequality cool2+0})--(%
\ref{inequality cool2+2}) we obviously get
\begin{equation*}
\underset{t\rightarrow \infty }{\lim }\ \frac{1}{t}\int_{0}^{t}\left\{ \Xi _{%
\mathrm{p},l}^{(\omega )}\left( s\right) \right\} _{k,q}\mathrm{d}s=-\frac{1%
}{\left\vert \Lambda _{l}\right\vert }(\mathbb{I}_{k,l},\tilde{E}\left(
\mathbb{R}\backslash \{0\}\right) \mathbb{I}_{q,l})_{\sim }\ ,
\end{equation*}%
which, by Corollary \ref{lemma sigma pos type} (iii), implies that
\begin{equation}
\mu _{\mathrm{p},l}^{(\omega )}\left( \mathbb{R}\backslash \{0\}\right) =%
\frac{1}{\left\vert \Lambda _{l}\right\vert }\left\{ (\mathbb{I}_{k,l},%
\tilde{E}\left( \mathbb{R}\backslash \{0\}\right) \mathbb{I}_{q,l})_{\sim
}\right\} _{k,q\in \{1,\ldots ,d\}}\ .  \notag
\end{equation}
\end{proof}

Note that the quantity%
\begin{equation*}
\Xi _{\mathrm{d},l}^{(\omega )}\,\underset{\epsilon \downarrow 0}{\lim }\,%
\mathbf{L}[\mathbf{V}_{l}^{(\omega )}](\epsilon )\in \mathcal{B}(\mathbb{R}%
^{d})
\end{equation*}%
is the so--called \emph{static admittance} of linear response theory, which
equals, in our case, the measure of $\mathbb{R}\backslash \{0\}$ w.r.t. the
AC--conductivity measure. In fact, the quantum current viscosity uniquely
defines the AC--conductivity measure:

\begin{satz}[Reconstruction of $\protect\mu _{\mathrm{p},l}^{(\protect\omega %
)}$ from the quantum current viscosity]
\label{lemma sigma pos type copy(6)}\mbox{
}\newline
Let $l,\beta \in \mathbb{R}^{+}$, $\omega \in \Omega $ and $\lambda \in
\mathbb{R}_{0}^{+}$. Then, for all $\vec{w}:=(w_{1},\ldots ,w_{d})\in
\mathbb{R}^{d}$ and any continuous and compactly supported real--valued
function $\mathcal{\hat{E}}$ with $\mathcal{\hat{E}}_{0}=0$,
\begin{multline*}
\int_{\mathbb{R}}\mathcal{\hat{E}}_{\nu }\left\langle \vec{w},\mu _{\mathrm{p%
},l}^{(\omega )}(\mathrm{d}\nu )\vec{w}\right\rangle =\underset{\epsilon
\downarrow 0}{\lim }\frac{1}{\pi }\int_{\mathbb{R}}\mathrm{d}\nu
\int_{0}^{\infty }\mathrm{d}s\frac{\left( \epsilon \cos \left( \nu s\right)
-\nu \sin \left( \nu s\right) \right) \mathrm{e}^{-\epsilon s}}{\nu
^{2}+\epsilon ^{2}} \\
\times \mathcal{\hat{E}}_{\nu }\left\langle \vec{w},\Xi _{\mathrm{d}%
,l}^{(\omega )}\mathbf{V}_{l}^{(\omega )}\left( s\right) \vec{w}%
\right\rangle \ .
\end{multline*}
\end{satz}

\begin{proof}
Fix $l,\beta \in \mathbb{R}^{+}$, $\omega \in \Omega $ and $\lambda \in
\mathbb{R}_{0}^{+}$. For any $\vec{w}\in \mathbb{R}^{d}$, define the
complex--valued function
\begin{equation*}
F_{\vec{w}}\left( z\right) :=\int_{\mathbb{R}}\frac{1}{\nu -z}\left\langle
\vec{w},\mu _{\mathrm{p},l}^{(\omega )}(\mathrm{d}\nu )\vec{w}\right\rangle
\ ,\qquad z\in \mathbb{C}^{+}\ ,
\end{equation*}%
where $\mathbb{C}^{+}$ is the set of complex numbers with strictly positive
imaginary part. $F_{\mathrm{p},l}^{(\omega )}$ is the so--called Borel
transform of the positive measure
\begin{equation}
\left\langle \vec{w},\mu _{\mathrm{p},l}^{(\omega )}(\mathrm{d}\nu )\vec{w}%
\right\rangle \ .  \label{positive measure}
\end{equation}%
By (\ref{von braun}), observe that
\begin{eqnarray*}
F_{\vec{w}}\left( z\right)  &=&\frac{1}{4\left\vert \Lambda _{l}\right\vert }%
\underset{k,q\in \{1,\ldots ,d\}}{\sum }w_{k}w_{q}\left( \mathbb{I}_{k,l},((%
\mathcal{\tilde{L}}-z)^{-1}+(-\mathcal{\tilde{L}}-z)^{-1})\mathbb{I}%
_{q,l}\right) _{\sim } \\
&&+\frac{1}{4\left\vert \Lambda _{l}\right\vert }\underset{k,q\in \{1,\ldots
,d\}}{\sum }w_{k}w_{q}\left( \mathbb{I}_{q,l},((\mathcal{\tilde{L}}%
-z)^{-1}+(-\mathcal{\tilde{L}}-z)^{-1})\mathbb{I}_{k,l}\right) _{\sim }
\end{eqnarray*}%
for any $z\in \mathbb{C}^{+}$ and $\vec{w}:=(w_{1},\ldots ,w_{d})\in \mathbb{%
R}^{d}$. Using%
\begin{equation*}
(\pm \mathcal{\tilde{L}}-z)^{-1}=i\int_{0}^{\infty }\mathrm{e}^{izs}\mathrm{e%
}^{\mp is\mathcal{\tilde{L}}}\mathrm{d}s\ ,\qquad z\in \mathbb{C}^{+}\ ,
\end{equation*}%
as well as Theorem \ref{Thm important equality asymptotics copy(1)} for $%
\mathcal{X}=\mathcal{U}$, $\tau =\tau ^{(\omega ,\lambda )}$\ and $\varrho
=\varrho ^{(\beta ,\omega ,\lambda )}$, we obtain%
\begin{equation*}
F_{\vec{w}}\left( z\right) =\frac{i}{\left\vert \Lambda _{l}\right\vert }%
\underset{k,q\in \{1,\ldots ,d\}}{\sum }w_{k}w_{q}\int_{0}^{\infty }\mathrm{e%
}^{izs}(\mathbb{I}_{k,l},\tau _{t}^{(\omega ,\lambda )}\left( \mathbb{I}%
_{q,l}\right) )_{\sim }\mathrm{d}s
\end{equation*}%
for every $z\in \mathbb{C}^{+}$ and $\vec{w}:=(w_{1},\ldots ,w_{d})\in
\mathbb{R}^{d}$. Using (\ref{average microscopic AC--conductivity}) and (\ref%
{inequality cool bog scalar}), we now integrate by parts the r.h.s of the
above equation to get%
\begin{eqnarray}
F_{\vec{w}}\left( z\right)  &=&-\frac{1}{\left\vert \Lambda _{l}\right\vert }%
\underset{k,q\in \{1,\ldots ,d\}}{\sum }w_{k}w_{q}z^{-1}\int_{0}^{\infty }%
\mathrm{e}^{izs}\varrho ^{(\beta ,\omega ,\lambda )}\left( i[\mathbb{I}%
_{k,l},\tau _{s}^{(\omega ,\lambda )}(\mathbb{I}_{q,l})]\right) \mathrm{d}s
\notag \\
&&-\frac{1}{\left\vert \Lambda _{l}\right\vert }\underset{k,q\in \{1,\ldots
,d\}}{\sum }w_{k}w_{q}z^{-1}(\mathbb{I}_{k,l},\mathbb{I}_{q,l})_{\sim }
\label{tot s1}
\end{eqnarray}%
for any $z\in \mathbb{C}^{+}$ and $\vec{w}:=(w_{1},\ldots ,w_{d})\in \mathbb{%
R}^{d}$. The function $\mathrm{Im}F_{\vec{w}}$ is the Poisson transform of
the positive measure (\ref{positive measure}). Hence, we invoke \cite[%
Theorem 3.7]{jaksich} to conclude that, for any real--valued continuous
compactly supported function $\mathcal{\hat{E}}:\mathbb{R}\rightarrow
\mathbb{R}$,
\begin{equation*}
\underset{\epsilon \downarrow 0}{\lim }\int_{\mathbb{R}}\mathcal{\hat{E}}%
_{\nu }\mathrm{Im}F_{\vec{w}}\left( \nu +i\epsilon \right) \mathrm{d}\nu
=\int_{\mathbb{R}}\mathcal{\hat{E}}_{\nu }\left\langle \vec{w},\mu _{\mathrm{%
p},l}^{(\omega )}(\mathrm{d}\nu )\vec{w}\right\rangle \ .
\end{equation*}%
In particular, by (\ref{tot s1}) and under the condition that $\mathcal{\hat{%
E}}_{0}=0$, we arrive at the assertion.
\end{proof}

To conclude, we show the uniformity of the upper bounds of Theorem \ref%
{lemma sigma pos type copy(1)} w.r.t. to the parameters $l,\beta \in \mathbb{%
R}^{+}$, $\omega \in \Omega $ and $\lambda \in \mathbb{R}_{0}^{+}$. These
upper bounds all depend on the observable $\left\vert \Lambda
_{l}\right\vert ^{-\frac{1}{2}}\mathbb{I}_{k,l}$, which is a \emph{current
fluctuation}, by (\ref{current density=}).

With this aim we define the linear subspace%
\begin{equation}
\mathcal{I}:=\mathrm{lin}\left\{ \mathrm{Im}(a^{\ast }\left( \psi
_{1}\right) a\left( \psi _{2}\right) ):\psi _{1},\psi _{2}\in \ell ^{1}(%
\mathfrak{L})\subset \ell ^{2}(\mathfrak{L})\right\} \subset \mathcal{U}\ ,
\label{space of currents}
\end{equation}%
which is the linear hull ($\mathrm{lin}$) of short range bond currents. It
is an invariant subspace of the one--parameter group $\tau ^{(\omega
,\lambda )}=\{\tau _{t}^{(\omega ,\lambda )}\}_{t\in {\mathbb{R}}}$ for any $%
\omega \in \Omega $ and $\lambda \in \mathbb{R}_{0}^{+}$. Indeed, the
unitary group $\{(\mathrm{U}_{t}^{(\omega ,\lambda )})^{\ast }\}_{t\in {%
\mathbb{R}}}$ (see (\ref{rescaled}) and (\ref{rescaledbis})) defines a
strongly continuous group on $(\ell ^{1}(\mathfrak{L})\subset \ell ^{2}(%
\mathfrak{L}),\Vert \cdot \Vert _{1})$.

Let the positive sesquilinear form $\langle \cdot ,\cdot \rangle _{\mathcal{I%
},l}^{(\omega )}\equiv \langle \cdot ,\cdot \rangle _{\mathcal{I},l}^{(\beta
,\omega ,\lambda )}$ in $\mathcal{I}$ be defined by
\begin{equation}
\langle I,I^{\prime }\rangle _{\mathcal{I},l}^{(\omega )}:=\varrho ^{(\beta
,\omega ,\lambda )}\left( \mathbb{F}^{(l)}\left( I\right) ^{\ast }\mathbb{F}%
^{(l)}\left( I^{\prime }\right) \right) \ ,\qquad I,I^{\prime }\in \mathcal{I%
}\ ,  \label{Fluctuation2bisbis}
\end{equation}%
for any $l,\beta \in \mathbb{R}^{+}$, $\omega \in \Omega $ and $\lambda \in
\mathbb{R}_{0}^{+}$. Here, $\mathbb{F}^{(l)}$ is the fluctuation observable
defined by%
\begin{equation}
\mathbb{F}^{(l)}\left( I\right) =\frac{1}{\left\vert \Lambda _{l}\right\vert
^{1/2}}\underset{x\in \Lambda _{l}}{\sum }\left\{ \chi _{x}\left( I\right)
-\varrho ^{(\beta ,\omega ,\lambda )}\left( I\right) \mathbf{1}\right\} \
,\qquad I\in \mathcal{I}\ ,  \label{Fluctuation2bis}
\end{equation}%
for each $l\in \mathbb{R}^{+}$, where $\chi _{x}$, $x\in \mathfrak{L}$, are
the (space) translation automorphisms. Compare (\ref{current density=}) with
(\ref{Fluctuation2bis}). For instance, the first upper bound of Theorem \ref%
{lemma sigma pos type copy(1)} can be rewritten as%
\begin{equation*}
\int_{\mathbb{R}}\Vert \mu _{\mathrm{p},l}^{(\omega )}\Vert _{\mathrm{op}}(%
\mathrm{d}\nu )\leq \underset{k=1}{\overset{d}{\sum }}\langle
I_{(e_{k},0)},I_{(e_{k},0)}\rangle _{\mathcal{I},l}^{(\omega )}\ .
\end{equation*}%
Therefore, we show that the fermion system has uniformly bounded
fluctuations, i.e., the quantity $\langle I , I^\prime \rangle _{\mathcal{I}%
,l}^{(\omega )}$, $I,I^\prime \in \mathcal{I}$, is uniformly bounded w.r.t. $%
l,\beta \in \mathbb{R}^{+}$, $\omega \in \Omega $, $\lambda \in \mathbb{R}%
_{0}^{+}$:

\begin{lemma}[Uniform boundedness of $\langle \cdot ,\cdot \rangle _{%
\mathcal{I},l}^{(\protect\omega )}$]
\label{bound incr 1 Lemma copy(3)}\mbox{
}\newline
There is a constant $D\in \mathbb{R}^{+}$ such that, for any $l,\beta \in
\mathbb{R}^{+}$, $\omega \in \Omega $, $\lambda \in \mathbb{R}_{0}^{+}$ and
all $\psi _{1},\psi _{2},\psi _{1}^{\prime },\psi _{2}^{\prime }\in \ell
^{1}(\mathfrak{L})$,
\begin{equation*}
\left\vert \langle \mathrm{Im}(a^{\ast }\left( \psi _{1}\right) a\left( \psi
_{2}\right) ),\mathrm{Im}(a^{\ast }\left( \psi _{1}^{\prime }\right) a\left(
\psi _{2}^{\prime }\right) )\rangle _{\mathcal{I},l}^{(\omega )}\right\vert
\leq D\left\Vert \psi _{1}\right\Vert _{1}\left\Vert \psi _{2}\right\Vert
_{1}\left\Vert \psi _{1}^{\prime }\right\Vert _{1}\left\Vert \psi
_{2}^{\prime }\right\Vert _{1}\ .
\end{equation*}
\end{lemma}

\begin{proof}
Let $\psi _{1},\psi _{2},\psi _{1}^{\prime },\psi _{2}^{\prime }\in \ell
^{1}(\mathfrak{L})\subset \ell ^{2}(\mathfrak{L})$ and without loss of
generality assume that the functions $\psi _{1},\psi _{2},\psi _{1}^{\prime
},\psi _{2}^{\prime }$ are real--valued. Then, by definition,
\begin{eqnarray*}
&&\langle \mathrm{Im}(a^{\ast }\left( \psi _{1}\right) a\left( \psi
_{2}\right) ),\mathrm{Im}(a^{\ast }\left( \psi _{1}^{\prime }\right) a\left(
\psi _{2}^{\prime }\right) )\rangle _{\mathcal{I},l}^{(\omega )} \\
&=&\sum_{\mathbf{x}:=(x^{(1)},x^{(2)}),\mathbf{y}:=(y^{(1)},y^{(2)})\in
\mathfrak{L}^{2}}\psi _{1}(y^{(1)})\psi _{2}(y^{(2)})\psi _{1}^{\prime
}(x^{(1)})\psi _{2}^{\prime }(x^{(2)}) \\
&&\times \left[ \frac{1}{4\left\vert \Lambda _{l}\right\vert }%
\sum_{z_{1},z_{2}\in \Lambda _{l}}\varrho ^{(\beta ,\omega ,\lambda )}\left(
I_{\mathbf{y}+(z_{2},z_{2})}^{\mathrm{fl}}I_{\mathbf{x}+(z_{1},z_{1})}^{%
\mathrm{fl}}\right) \right] \ ,
\end{eqnarray*}%
where
\begin{equation*}
I_{\mathbf{x}}^{\mathrm{fl}}:=I_{\mathbf{x}}-\varrho ^{(\beta ,\omega
,\lambda )}\left( I_{\mathbf{x}}\right) \mathbf{1}\ ,\qquad \mathbf{x}\in
\mathfrak{L}^{2}\ .
\end{equation*}%
Recall that $I_{\mathbf{x}}$ is the paramagnetic current observable defined
by (\ref{current observable}). Hence, it suffices to prove the existence of
a finite constant $D\in \mathbb{R}^{+}$ such that, for any $l,\beta \in
\mathbb{R}^{+}$, $\omega \in \Omega $, $\lambda \in \mathbb{R}_{0}^{+}$ and
all $\mathbf{x},\mathbf{y}\in \mathfrak{L}^{2}$,
\begin{equation}
\left\vert \frac{1}{4\left\vert \Lambda _{l}\right\vert }\sum_{z_{1},z_{2}%
\in \Lambda _{l}}\varrho ^{(\beta ,\omega ,\lambda )}\left( I_{\mathbf{y}%
+(z_{2},z_{2})}^{\mathrm{fl}}I_{\mathbf{x}+(z_{1},z_{1})}^{\mathrm{fl}%
}\right) \right\vert \leq D\ .  \label{ineq important Hilbert}
\end{equation}%
This can be shown by using Lemma \ref{remark operator cool}.

Indeed, we infer from (\ref{sigma ohm3}) at $t=\alpha =0$ that, for any $%
l,\beta \in \mathbb{R}^{+}$, $\omega \in \Omega $, $\lambda \in \mathbb{R}%
_{0}^{+}$, $\mathbf{x},\mathbf{y}\in \mathfrak{L}^{2}$ and all $%
z_{1},z_{2}\in \Lambda _{l}$,%
\begin{eqnarray}
\varrho ^{(\beta ,\omega ,\lambda )}\left( I_{\mathbf{y}+(z_{2},z_{2})}^{%
\mathrm{fl}}I_{\mathbf{x}+(z_{1},z_{1})}^{\mathrm{fl}}\right) &=&\varrho
^{(\beta ,\omega ,\lambda )}\left( I_{\mathbf{y}+(z_{2},z_{2})}I_{\mathbf{x}%
+(z_{1},z_{1})}\right)  \notag \\
&&-\varrho ^{(\beta ,\omega ,\lambda )}\left( I_{\mathbf{y}%
+(z_{2},z_{2})}\right) \varrho ^{(\beta ,\omega ,\lambda )}\left( I_{\mathbf{%
x}+(z_{1},z_{1})}\right)  \notag \\
&=&\mathfrak{C}_{0}^{(\omega )}\left( \mathbf{x}+\left( z_{1},z_{1}\right) ,%
\mathbf{y}+\left( z_{2},z_{2}\right) \right) \ ,  \label{equality a etudier}
\end{eqnarray}%
where $\mathfrak{C}_{t+i\alpha }^{(\omega )}$ is the map from $\mathfrak{L}%
^{4}$ to $\mathbb{C}$ defined at $t\in \mathbb{R}$ and $\alpha \in \lbrack
0,\beta ]$ by (\ref{map cool}). Now, take the canonical orthonormal basis $\{%
\mathbf{e}_{\mathbf{x}}\}_{\mathbf{x}\in \mathfrak{L}^{2}}$ of $\ell ^{2}(%
\mathfrak{L})\otimes \ell ^{2}(\mathfrak{L})$ defined by
\begin{equation*}
\mathbf{e}_{\mathbf{x}}:=\mathfrak{e}_{x^{(1)}}\otimes \mathfrak{e}%
_{x^{(2)}}\ ,\qquad \mathbf{x}:=(x^{(1)},x^{(2)})\in \mathfrak{L}^{2}\ .
\end{equation*}%
Recall that $\mathfrak{e}_{x}(y)\equiv \delta _{x,y}\in \ell ^{2}(\mathfrak{L%
})$. Then, the coefficient $\mathfrak{C}_{t+i\alpha }^{(\omega )}$ can be
seen as a kernel -- w.r.t. the canonical basis $\{\mathbf{e}_{\mathbf{x}}\}_{%
\mathbf{x}\in \mathfrak{L}^{2}}$ -- of an operator on $\ell ^{2}(\mathfrak{L}%
)\otimes \ell ^{2}(\mathfrak{L})$, again denoted by $\mathfrak{C}_{t+i\alpha
}^{(\omega )}$. Then, we observe from (\ref{equality a etudier}) that%
\begin{eqnarray}
\lefteqn{ \frac{1}{4\left\vert \Lambda _{l}\right\vert }\sum_{z_{1},z_{2}\in
\Lambda _{l}}\varrho ^{(\beta ,\omega ,\lambda )}\left( I_{\mathbf{y}%
+(z_{2},z_{2})}^{\mathrm{fl}}I_{\mathbf{x}+(z_{1},z_{1})}^{\mathrm{fl}%
}\right) } &&  \notag \\
&=&\frac{1}{4\left\vert \Lambda _{l}\right\vert }\sum_{z_{1},z_{2}\in
\Lambda _{l}}\left\langle \mathbf{e}_{\mathbf{x}+\left( z_{1},z_{1}\right) },%
\mathfrak{C}_{0}^{(\omega )}(\mathbf{e}_{\mathbf{y}+\left(
z_{2},z_{2}\right) })\right\rangle  \label{map cool3-1}
\end{eqnarray}%
for any $l,\beta \in \mathbb{R}^{+}$, $\omega \in \Omega $, $\lambda \in
\mathbb{R}_{0}^{+}$ and $\mathbf{x},\mathbf{y}\in \mathfrak{L}^{2}$.

By Lemma \ref{remark operator cool}, the operator $\mathfrak{C}_{t+i\alpha
}^{(\omega )}$ always satisfies $\Vert \mathfrak{C}_{t+i\alpha }^{(\omega
)}\Vert _{\mathrm{op}}\leq 4$ and hence,%
\begin{equation}
\left\vert \frac{1}{4\left\vert \Lambda _{l}\right\vert }\sum_{z_{1},z_{2}%
\in \Lambda _{l}}\left\langle \mathbf{e}_{\mathbf{x}+\left(
z_{1},z_{1}\right) },\mathfrak{C}_{0}^{(\omega )}\mathbf{e}_{\mathbf{y}%
+\left( z_{2},z_{2}\right) }\right\rangle \right\vert \leq 1
\label{map cool3}
\end{equation}%
for any $l,\beta \in \mathbb{R}^{+}$, $\omega \in \Omega $, $\lambda \in
\mathbb{R}_{0}^{+}$ and $\mathbf{x},\mathbf{y}\in \mathfrak{L}^{2}$. By (\ref%
{map cool3-1}), it follows that
\begin{equation*}
\left\vert \frac{1}{4\left\vert \Lambda _{l}\right\vert }\sum_{z_{1},z_{2}%
\in \Lambda _{l}}\varrho ^{(\beta ,\omega ,\lambda )}\left( I_{\mathbf{y}%
+(z_{2},z_{2})}^{\mathrm{fl}}I_{\mathbf{x}+(z_{1},z_{1})}^{\mathrm{fl}%
}\right) \right\vert \leq 1
\end{equation*}%
for any $l,\beta \in \mathbb{R}^{+}$, $\omega \in \Omega $, $\lambda \in
\mathbb{R}_{0}^{+}$ and all $\mathbf{x},\mathbf{y}\in \mathfrak{L}^{2}$.
\end{proof}

\subsection{Tree--Decay Bounds and Uniformity of Responses\label{Section
local joule-ohm}}

\subsubsection{Uniformity of Energy Increment Responses\label{Local
AC--Ohm's Law}}

For the reader's convenience we start by reminding a few definitions and
some standard mathematical facts used in our proofs. First of all, we recall
that in \cite[Section 5.2]{OhmI} we give an explicit expression of the
automorphisms $\tau _{t,s}^{(\omega ,\lambda ,\mathbf{A})}$ of $\mathcal{U}$
in terms of series involving \emph{multi--commutators}, see \cite[Eqs.
(3.14)-(3.15)]{OhmI}. Indeed, in \cite[Eq. (5.15)]{OhmI} we represent the
automorphisms $\tau _{t,s}^{(\omega ,\lambda ,\mathbf{A})}$ as the following
Dyson--Phillips series
\begin{eqnarray}
&&\tau _{t,s}^{(\omega ,\lambda ,\mathbf{A})}\left( B\right) -\tau
_{t-s}^{(\omega ,\lambda )}\left( B\right)  \label{Dyson tau 1} \\
&=&\sum\limits_{k\in {\mathbb{N}}}i^{k}\int_{s}^{t}\mathrm{d}s_{1}\cdots
\int_{s}^{s_{k-1}}\mathrm{d}s_{k}[W_{s_{k}-s,s_{k}}^{\mathbf{A}},\ldots
,W_{s_{1}-s,s_{1}}^{\mathbf{A}},\tau _{t-s}^{(\omega ,\lambda )}(B)]^{(k+1)}
\notag
\end{eqnarray}%
for any $B\in \mathcal{U}$ and $t\geq s$. Here, for any $t,s\in \mathbb{R}$,
\begin{equation}
W_{t,s}^{\mathbf{A}}:=\tau _{t}^{(\omega ,\lambda )}(W_{s}^{\mathbf{A}})\in
\mathcal{U}  \label{def LA}
\end{equation}%
is the time--evolution of the electromagnetic potential energy observable $%
W_{s}^{\mathbf{A}}$ defined by (\ref{eq def W}), that is,
\begin{eqnarray}
W_{s}^{\mathbf{A}}:= &&\sum\limits_{x,y\in \mathfrak{L}}\left[ \exp \left(
i\int\nolimits_{0}^{1}\left[ \mathbf{A}(s,\alpha y+(1-\alpha )x)\right] (y-x)%
\mathrm{d}\alpha \right) -1\right]  \notag \\
&&\qquad \ \ \ \times \langle \mathfrak{e}_{x},\Delta _{\mathrm{d}}\mathfrak{%
e}_{y}\rangle a_{x}^{\ast }a_{y}\ ,  \label{def:W}
\end{eqnarray}%
for any $\mathbf{A}\in \mathbf{C}_{0}^{\infty }$ and $s\in \mathbb{R}$.

The expression (\ref{Dyson tau 1}) is useful because we can apply \emph{%
tree--decay bounds} on multi--commutators. These bounds, derived in \cite[%
Section 4]{OhmI}, are useful to analyze multi--commutators of products of
annihilation and creation operators. Using them, we show for instance in
\cite[Lemma 5.10]{OhmI} that, for any $\mathbf{A}\in \mathbf{C}_{0}^{\infty
} $, there is $\eta _{0}\in \mathbb{R}^{+}$ such that, for $l,\varepsilon
\in \mathbb{R}^{+}$, there is a ball%
\begin{equation*}
B(0,R):=\{x\in \mathfrak{L}:|x|\leq R\}
\end{equation*}%
of radius $R\in \mathbb{R}^{+}$ centered at $0$ such that, for all $|\eta
|\in (0,\eta _{0}]$, $\beta \in \mathbb{R}^{+}$, $\omega \in \Omega $, $%
\lambda \in \mathbb{R}_{0}^{+}$, and $t_{0}\leq s_{1},\ldots ,s_{k}\leq t$,
\begin{multline*}
\sum\limits_{x\in \Lambda _{L}\backslash B_{R}}\sum\limits_{z\in \mathfrak{L}%
,|z|\leq 1}\sum\limits_{k\in {\mathbb{N}}}\frac{\left( t-t_{0}\right) ^{k}}{%
k!} \\
\left\vert \varrho ^{(\beta ,\omega ,\lambda )}\left(
[W_{s_{k}-t_{0},s_{k}}^{\eta \mathbf{A}_{l}},\ldots
,W_{s_{1}-t_{0},s_{1}}^{\eta \mathbf{A}_{l}},\tau _{t-t_{0}}^{(\omega
,\lambda )}(a_{x}^{\ast }a_{x+z})]^{(k+1)}\right) \right\vert \leq
\varepsilon \ .
\end{multline*}%
This property together with (\ref{entropic energy increment}) and (\ref%
{Dyson tau 1}) implies that, for all $|\eta |\in (0,\eta _{0}]$, $l,\beta
\in \mathbb{R}^{+}$, $\omega \in \Omega $, $\lambda \in \mathbb{R}_{0}^{+}$
and $t\geq t_{0}$,%
\begin{eqnarray}
\mathbf{S}^{(\omega ,\eta \mathbf{A}_{l})}\left( t\right)
&=&\sum\limits_{k\in {\mathbb{N}}}\sum\limits_{x,z\in \mathfrak{L},|z|\leq
1}\langle \mathfrak{e}_{x},\left( \Delta _{\mathrm{d}}+\lambda V_{\omega
}\right) \mathfrak{e}_{x+z}\rangle i^{k}\int_{t_{0}}^{t}\mathrm{d}%
s_{1}\cdots \int_{t_{0}}^{s_{k-1}}\mathrm{d}s_{k}  \notag \\
&&\varrho ^{(\beta ,\omega ,\lambda )}\left( [W_{s_{k}-t_{0},s_{k}}^{\eta
\mathbf{A}_{l}},\ldots ,W_{s_{1}-t_{0},s_{1}}^{\eta \mathbf{A}_{l}},\tau
_{t-t_{0}}^{(\omega ,\lambda )}(a_{x}^{\ast }a_{x+z})]^{(k+1)}\right) \ .
\notag \\
&&  \label{energy increment bis}
\end{eqnarray}%
See \cite[Section 5.5]{OhmI} for more details.

These assertions are important to get uniform bounds as explained in
Theorems \ref{thm Local Ohm's law} and \ref{Local Ohm's law thm copy(2)}.
Indeed, it is relatively straightforward to get the asymptotics of the
elements $W_{t}^{\eta \mathbf{A}_{l}}$ and $\partial _{t}W_{t}^{\mathbf{A}}$
when $(\eta ,l^{-1})\rightarrow (0,0)$ by using the integrated electric
field
\begin{equation}
\mathbf{E}_{t}^{\mathbf{A}}\left( \mathbf{x}\right) :=\int\nolimits_{0}^{1}%
\left[ E_{\mathbf{A}}(t,\alpha x^{(2)}+(1-\alpha )x^{(1)})\right]
(x^{(2)}-x^{(1)})\mathrm{d}\alpha  \label{V definition}
\end{equation}%
between $x^{(2)}\in \mathfrak{L}$ and $x^{(1)}\in \mathfrak{L}$ at time $%
t\in \mathbb{R}$ (cf. (\ref{V bar 0bis})) and the subset
\begin{equation}
\mathfrak{K}:=\left\{ \mathbf{x}:=(x^{(1)},x^{(2)})\in \mathfrak{L}^{2}\ :\
|x^{(1)}-x^{(2)}|=1\right\}  \label{proche voisins}
\end{equation}%
of bonds of nearest neighbors (cf. (\ref{proche voisins0})).

\begin{lemma}[Asymptotics of the potential energy observable]
\label{bound incr 1 Lemma copy(9)}\mbox{
}\newline
For any $\eta ,l\in \mathbb{R}^{+}$, $\mathbf{A}\in \mathbf{C}_{0}^{\infty }$
and $t\geq t_{0}$, there are complex numbers
\begin{equation*}
\left\{ \tilde{D}_{x,y}^{\eta \mathbf{A}_{l}}(t)\right\} _{x,y\in \mathfrak{L%
}}\subset \mathbb{C}
\end{equation*}%
and a $(\eta ,t)$--independent subset $\widetilde{\Lambda }_{l}\in \mathcal{P%
}_{f}(\mathfrak{L})$ of diameter of order $\mathcal{O}(l)$ such that%
\begin{eqnarray*}
W_{t}^{\eta \mathbf{A}_{l}} &=&-\frac{1}{2}\ \underset{\mathbf{x}\in
\mathfrak{K}}{\sum }\left\{ \eta \left( \int\nolimits_{t_{0}}^{t}\mathbf{E}%
_{s}^{\mathbf{A}_{l}}(\mathbf{x})\mathrm{d}s\right) I_{\mathbf{x}}+\frac{%
\eta ^{2}}{2}\left( \int\nolimits_{t_{0}}^{t}\mathbf{E}_{s}^{\mathbf{A}_{l}}(%
\mathbf{x})\mathrm{d}s\right) ^{2}P_{\mathbf{x}}\right\} \\
&&+\eta ^{3}\sum\limits_{x\in \widetilde{\Lambda }_{l}}\sum\limits_{z\in
\mathfrak{L},|z|\leq 1}\tilde{D}_{x,x+z}^{\eta \mathbf{A}_{l}}(t)a_{x}^{\ast
}a_{x+z}
\end{eqnarray*}%
with $\tilde{D}_{x,x+z}^{\eta \mathbf{A}_{l}}(t)$ and $\partial _{t}\tilde{D}%
_{x,x+z}^{\eta \mathbf{A}_{l}}(t)$ being uniformly bounded for all $\eta $
in compact sets, all $x,z\in \mathfrak{L}$ such that $|z|\leq 1$, and all $%
\omega \in \Omega $, $\lambda \in \mathbb{R}_{0}^{+}$ and $l\in \mathbb{R}%
^{+}$.
\end{lemma}

\begin{proof}
Note that (\ref{V definition}) yields%
\begin{equation*}
\mathbf{E}_{t}^{\mathbf{A}}(\mathbf{x})\equiv \mathbf{E}_{t}^{\mathbf{A}%
}(x^{(1)},x^{(2)})=-\mathbf{E}_{t}^{\mathbf{A}}(x^{(2)},x^{(1)})\ ,\quad
\mathbf{x}:=(x^{(1)},x^{(2)})\in \mathfrak{L}^{2}\ ,\ t\in \mathbb{R}\ .
\end{equation*}%
Therefore, the statement is a straightforward consequence of Equations (\ref%
{discrete laplacian}), (\ref{def:W}) and (\ref{V definition}) together with
\cite[Eqs. (5.37)--(5.39), (5.41)]{OhmI}.
\end{proof}

By combining this lemma with (\ref{energy increment bis}) one can obtain
Theorem \ref{Local Ohm's law thm copy(2)} (\textbf{S}). However, by using (%
\ref{work-dia}), it is easier to start with the paramagnetic and diamagnetic
energies $\mathfrak{J}_{\mathrm{p}}^{(\omega ,\mathbf{A})}$ and $\mathfrak{I}%
_{\mathrm{d}}^{(\omega ,\mathbf{A})}$ respectively defined by (\ref%
{lim_en_incr}) and (\ref{lim_en_incr dia}):

\begin{satz}[Microscopic paramagnetic and diamagnetic energies]
\label{Local Ohm's law thm}\mbox{
}\newline
For any $\mathbf{A}\in \mathbf{C}_{0}^{\infty }$, there is $\eta _{0}\in
\mathbb{R}^{+}$ such that, for all $|\eta |\in (0,\eta _{0}]$, $l,\beta \in
\mathbb{R}^{+}$, $\omega \in \Omega $, $\lambda \in \mathbb{R}_{0}^{+}$ and $%
t\geq t_{0}$, one has:\newline
\emph{(p)} Paramagnetic energy increment:%
\begin{equation*}
\mathfrak{I}_{\mathrm{p}}^{(\omega ,\eta \mathbf{A}_{l})}\left( t\right) =%
\frac{\eta ^{2}}{4}\int\nolimits_{t_{0}}^{t}\mathrm{d}s_{1}%
\int_{t_{0}}^{s_{1}}\mathrm{d}s_{2}\underset{\mathbf{x},\mathbf{y}\in
\mathfrak{K}}{\sum }\sigma _{\mathrm{p}}^{(\omega )}\left( \mathbf{x},%
\mathbf{y,}s_{1}-s_{2}\right) \mathbf{E}_{s_{2}}^{\mathbf{A}_{l}}(\mathbf{y})%
\mathbf{E}_{s_{1}}^{\mathbf{A}_{l}}(\mathbf{x})+\mathcal{O}(\eta ^{3}l^{d})\
.
\end{equation*}%
\emph{(d)} Diamagnetic energy:
\begin{eqnarray*}
\mathfrak{I}_{\mathrm{d}}^{(\omega ,\eta \mathbf{A}_{l})}\left( t\right) &=&-%
\frac{\eta }{2}\underset{\mathbf{x}\in \mathfrak{K}}{\sum }\varrho ^{(\beta
,\omega ,\lambda )}(I_{\mathbf{x}})\int\nolimits_{t_{0}}^{t}\mathbf{E}_{s}^{%
\mathbf{A}_{l}}(\mathbf{x})\mathrm{d}s \\
&&+\frac{\eta ^{2}}{2}\int\nolimits_{t_{0}}^{t}\mathrm{d}s_{1}%
\int_{t_{0}}^{s_{1}}\mathrm{d}s_{2}\underset{\mathbf{x}\in \mathfrak{K}}{%
\sum }\sigma _{\mathrm{d}}^{(\omega )}\left( \mathbf{x}\right) \mathbf{E}%
_{s_{2}}^{\mathbf{A}_{l}}(\mathbf{x})\mathbf{E}_{s_{1}}^{\mathbf{A}_{l}}(%
\mathbf{x})+\mathcal{O}(\eta ^{3}l^{d})\ .
\end{eqnarray*}%
The correction terms of order $\mathcal{O}(l^{d}\eta ^{3})$ in assertions
(p) and (d) are uniformly bounded in $\beta \in \mathbb{R}^{+}$, $\omega \in
\Omega $, $\lambda \in \mathbb{R}_{0}^{+}$ and $t\geq t_{0}$.
\end{satz}

\begin{proof}
(p) Using $W_{t}^{\mathbf{A}}=0$ for any $t\leq t_{0}$ and (\ref{stationary}%
) we note that, for any $t\geq t_{0}$,%
\begin{equation*}
\varrho ^{(\beta ,\omega ,\lambda )}(W_{t}^{\eta \mathbf{A}%
_{l}})=\int\nolimits_{t_{0}}^{t}\varrho ^{(\beta ,\omega ,\lambda )}\left(
\partial _{s}W_{s}^{\eta \mathbf{A}_{l}}\right) \mathrm{d}%
s=\int\nolimits_{t_{0}}^{t}\varrho ^{(\beta ,\omega ,\lambda )}\circ \tau
_{s-t_{0}}^{(\omega ,\lambda )}\left( \partial _{s}W_{s}^{\eta \mathbf{A}%
_{l}}\right) \mathrm{d}s\ .
\end{equation*}%
For all $s\in \mathbb{R}$,
\begin{equation*}
W_{s}^{\eta \mathbf{A}_{l}},\partial _{s}W_{s}^{\eta \mathbf{A}_{l}}\in
\mathcal{U}_{\widetilde{\Lambda }_{l}}
\end{equation*}%
for some finite subset $\widetilde{\Lambda }_{l}\in \mathcal{P}_{f}(%
\mathfrak{L})$ of diameter of order $\mathcal{O}(l)$, see, e.g., \cite[Eqs.
(5.41)]{OhmI}. As a consequence, by (\ref{lim_en_incr})--(\ref{work-dia}),
the paramagnetic energy increment equals%
\begin{equation}
\mathfrak{I}_{\mathrm{p}}^{(\omega ,\eta \mathbf{A}_{l})}\left( t\right)
=\int\nolimits_{t_{0}}^{t}\varrho ^{(\beta ,\omega ,\lambda )}\circ \left(
\tau _{s,t_{0}}^{(\omega ,\lambda ,\eta \mathbf{A}_{l})}-\tau
_{s-t_{0}}^{(\omega ,\lambda )}\right) \left( \partial _{s}W_{s}^{\eta
\mathbf{A}_{l}}\right) \mathrm{d}s  \label{local ohm2}
\end{equation}%
for any $\beta \in \mathbb{R}^{+}$, $\omega \in \Omega $, $\lambda \in
\mathbb{R}_{0}^{+}$, $\mathbf{A}\in \mathbf{C}_{0}^{\infty }$ and $t\geq
t_{0}$.

Similar to the proof of \cite[Lemma 5.10]{OhmI}, one uses Dyson--Phillips
expansions (\ref{Dyson tau 1}) and tree--decay bounds on multi--commutators
\cite[Corollary 4.3]{OhmI} to infer from Lemma \ref{bound incr 1 Lemma
copy(9)} and Equation (\ref{local ohm2}) that, for any $\mathbf{A}\in
\mathbf{C}_{0}^{\infty }$, there is $\eta _{0}\in \mathbb{R}^{+}$ such that,
for all $\left\vert \eta \right\vert \in (0,\eta _{0}]$, $l,\beta \in
\mathbb{R}^{+}$, $\omega \in \Omega $, $\lambda \in \mathbb{R}_{0}^{+}$ and $%
t\geq t_{0}$,%
\begin{eqnarray}
\mathfrak{I}_{\mathrm{p}}^{(\omega ,\eta \mathbf{A}_{l})}\left( t\right)
&=&\int\nolimits_{t_{0}}^{t}\mathrm{d}s_{1}\int_{t_{0}}^{s_{1}}\mathrm{d}%
s_{2}\ \varrho ^{(\beta ,\omega ,\lambda )}\left( i\left[ \tau
_{s_{2}-t_{0}}^{(\omega ,\lambda )}\left( W_{s_{2}}^{\eta \mathbf{A}%
_{l}}\right) ,\tau _{s_{1}-t_{0}}^{(\omega ,\lambda )}\left( \partial
_{s_{1}}W_{s_{1}}^{\eta \mathbf{A}_{l}}\right) \right] \right)  \notag \\
&&+\mathcal{O}(\eta ^{3}l^{d})\ .  \label{local ohm law1}
\end{eqnarray}%
This last correction term of order $\mathcal{O}(l^{d}\eta ^{3})$ is \emph{%
uniformly bounded} in $\beta \in \mathbb{R}^{+}$, $\omega \in \Omega $, $%
\lambda \in \mathbb{R}_{0}^{+}$ and $t\geq t_{0}$.

Note that (\ref{inequality idiote})--(\ref{stationary}) combined with the
group property of the family $\{\tau _{t}^{(\omega ,\lambda )}\}_{t\in {%
\mathbb{R}}}$ imply that%
\begin{equation*}
\varrho ^{(\beta ,\omega ,\lambda )}\left( [\tau _{s_{2}-t_{0}}^{(\omega
,\lambda )}(B_{2}),\tau _{s_{1}-t_{0}}^{(\omega ,\lambda )}\left(
B_{1}\right) ]\right) =\varrho ^{(\beta ,\omega ,\lambda )}\left( [\tau
_{s_{2}}^{(\omega ,\lambda )}\left( B_{2}\right) ,\tau _{s_{1}}^{(\omega
,\lambda )}(B_{1})]\right)
\end{equation*}%
for any $B_{1},B_{2}\in \mathcal{U}$ and all $s_{1},s_{1}\in \mathbb{R}$.
Therefore, we insert this equality and the asymptotics given by Lemma \ref%
{bound incr 1 Lemma copy(9)} in Equation (\ref{local ohm law1}) to arrive at
the equality%
\begin{eqnarray}
\mathfrak{I}_{\mathrm{p}}^{(\omega ,\eta \mathbf{A}_{l})}\left( t\right) &=&%
\frac{\eta ^{2}}{4}\underset{\mathbf{x},\mathbf{y}\in \mathfrak{K}}{\sum }%
\int\nolimits_{t_{0}}^{t}\mathrm{d}s_{1}\int_{t_{0}}^{s_{1}}\mathrm{d}%
s_{2}\int\nolimits_{t_{0}}^{s_{2}}\mathrm{d}s_{3}  \notag \\
&&\times \mathbf{E}_{s_{1}}^{\mathbf{A}_{l}}(\mathbf{x})\mathbf{E}_{s_{3}}^{%
\mathbf{A}_{l}}(\mathbf{y})\varrho ^{(\beta ,\omega ,\lambda )}\left( i[\tau
_{s_{2}}^{(\omega ,\lambda )}(I_{\mathbf{y}}),\tau _{s_{1}}^{(\omega
,\lambda )}\left( I_{\mathbf{x}}\right) ]\right)  \notag \\
&&+\mathcal{O}(\eta ^{3}l^{d})\ ,  \label{eq ohm local6}
\end{eqnarray}%
uniformly for $\beta \in \mathbb{R}^{+}$, $\omega \in \Omega $, $\lambda \in
\mathbb{R}_{0}^{+}$ and $t\geq t_{0}$.

For any $\beta \in \mathbb{R}^{+}$, $\omega \in \Omega $, $\lambda \in
\mathbb{R}_{0}^{+}$, $\mathbf{x},\mathbf{y}\in \mathfrak{L}^{2}$ and $%
s_{1},s_{2}\in \mathbb{R}$, let
\begin{equation}
\zeta _{\mathbf{x},\mathbf{y}}^{(\omega )}\left( s_{1},s_{2}\right)
:=\int\nolimits_{s_{1}}^{s_{2}}\varrho ^{(\beta ,\omega ,\lambda )}\left(
i[\tau _{s_{1}}^{(\omega ,\lambda )}\left( I_{\mathbf{y}}\right) ,\tau
_{s}^{(\omega ,\lambda )}(I_{\mathbf{x}})]\right) \mathrm{d}s\ .
\label{backwards -1}
\end{equation}%
Note that the function $\zeta _{\mathbf{x},\mathbf{y}}^{(\omega )}$ is a map
from $\mathbb{R}^{2}$ to $\mathbb{R}$. By combining (\ref{backwards -1})
with (\ref{inequality idiote})--(\ref{stationary}) and (\ref{backwards -1bis}%
), we observe that
\begin{equation}
\zeta _{\mathbf{x},\mathbf{y}}^{(\omega )}\left( s_{1},s_{2}\right) =\sigma
_{\mathrm{p}}^{(\omega )}\left( \mathbf{x},\mathbf{y,}s_{2}-s_{1}\right)
=\sigma _{\mathrm{p}}^{(\omega )}\left( \mathbf{y},\mathbf{x,}%
s_{1}-s_{2}\right)  \label{backwards -2}
\end{equation}%
for any $\beta \in \mathbb{R}^{+}$, $\omega \in \Omega $, $\lambda \in
\mathbb{R}_{0}^{+}$, $\mathbf{x},\mathbf{y}\in \mathfrak{L}^{2}$ and $%
s_{1},s_{2}\in \mathbb{R}$, while%
\begin{equation}
\partial _{s_{2}}\zeta _{\mathbf{y},\mathbf{x}}^{(\omega )}\left(
s_{1},s_{2}\right) =\varrho ^{(\beta ,\omega ,\lambda )}\left( i[\tau
_{s_{1}}^{(\omega ,\lambda )}\left( I_{\mathbf{x}}\right) ,\tau
_{s_{2}}^{(\omega ,\lambda )}(I_{\mathbf{y}})]\right) \ .  \label{ptt remark}
\end{equation}%
As a consequence, the assertion follows from (\ref{eq ohm local6}) and an
integration by parts.

\noindent (d) is a direct consequence of (\ref{backwards -1bispara}), (\ref%
{lim_en_incr dia}) and Lemma \ref{bound incr 1 Lemma copy(9)}.
\end{proof}

It remains to study the entropic energy increment $\mathbf{S}^{(\omega ,\eta
\mathbf{A}_{l})}$ and the electromagnetic energy $\mathbf{P}^{(\omega ,\eta
\mathbf{A}_{l})}$ defined by (\ref{entropic energy increment}) and (\ref%
{electro free energy}), respectively. To this end, it suffices to study the
\emph{potential energy difference}%
\begin{equation*}
\mathbf{P}^{(\omega ,\eta \mathbf{A}_{l})}\left( t\right) -\mathfrak{I}_{%
\mathrm{d}}^{(\omega ,\eta \mathbf{A}_{l})}\left( t\right) =\rho
_{t}^{(\beta ,\omega ,\lambda ,\eta \mathbf{A}_{l})}(W_{t}^{\eta \mathbf{A}%
_{l}})-\varrho ^{(\beta ,\omega ,\lambda )}(W_{t}^{\eta \mathbf{A}_{l}})
\end{equation*}%
for all times $t\geq t_{0}$. This is done in the following lemma:

\begin{lemma}[Potential energy difference]
\label{bound incr 1 Lemma copy(5)}\mbox{
}\newline
For any $\mathbf{A}\in \mathbf{C}_{0}^{\infty }$, there is $\eta _{0}\in
\mathbb{R}^{+}$ such that, for all $|\eta |\in (0,\eta _{0}]$ and $l\in
\mathbb{R}^{+}$,
\begin{eqnarray*}
&&\mathbf{P}^{(\omega ,\eta \mathbf{A}_{l})}\left( t\right) -\mathfrak{I}_{%
\mathrm{d}}^{(\omega ,\eta \mathbf{A}_{l})}\left( t\right) \\
&=&\frac{\eta ^{2}}{4}\underset{\mathbf{x},\mathbf{y}\in \mathfrak{K}}{\sum }%
\left( \int\nolimits_{t_{0}}^{t}\mathbf{E}_{s}^{\mathbf{A}_{l}}(\mathbf{x})%
\mathrm{d}s\right) \left( \int_{t_{0}}^{t}\sigma _{\mathrm{p}}^{(\omega
)}\left( \mathbf{x},\mathbf{y},t-s\right) \mathbf{E}_{s}^{\mathbf{A}_{l}}(%
\mathbf{y})\mathrm{d}s\right) \\
&&+\mathcal{O}(\eta ^{3}l^{d})\ ,
\end{eqnarray*}%
uniformly for $\beta \in \mathbb{R}^{+}$, $\omega \in \Omega $, $\lambda \in
\mathbb{R}_{0}^{+}$ and $t\geq t_{0}$.
\end{lemma}

\begin{proof}
The proof is very similar to the one of Theorem \ref{Local Ohm's law thm}.
In particular, to get the asymptotics, it suffices to observe that, for any $%
\mathbf{A}\in \mathbf{C}_{0}^{\infty }$, there is $\eta _{0}\in \mathbb{R}%
^{+}$ such that, for all $|\eta |\in (0,\eta _{0}]$, $l,\beta \in \mathbb{R}%
^{+}$, $\omega \in \Omega $, $\lambda \in \mathbb{R}_{0}^{+}$ and $t\geq
t_{0}$,
\begin{eqnarray}
&&\rho _{t}^{(\beta ,\omega ,\lambda ,\eta \mathbf{A}_{l})}(W_{t}^{\eta
\mathbf{A}_{l}})-\varrho ^{(\beta ,\omega ,\lambda )}(W_{t}^{\eta \mathbf{A}%
_{l}})  \notag \\
&=&\int_{t_{0}}^{t}\varrho ^{(\beta ,\omega ,\lambda )}\left( i[\tau
_{s}^{(\omega ,\lambda )}(W_{s}^{\eta \mathbf{A}_{l}}),\tau _{t}^{(\omega
,\lambda )}(W_{t}^{\eta \mathbf{A}_{l}})]\right) \mathrm{d}s+\mathcal{O}%
(\eta ^{3}l^{d})\ ,  \label{electric power1}
\end{eqnarray}%
by (\ref{inequality idiote})--(\ref{stationary}), the Dyson--Phillips
expansions (\ref{Dyson tau 1}), Lemma \ref{bound incr 1 Lemma copy(9)} and
tree--decay bounds on multi--commuta%
%TCIMACRO{\TeXButton{\-}{\-}}%
%BeginExpansion
\-%
%EndExpansion
tors \cite[Corollary 4.3]{OhmI}. Note that the correction term of order $%
\mathcal{O}(\eta ^{3}l^{d})$ in (\ref{electric power1}) is again \emph{%
uniformly bounded} in $\beta \in \mathbb{R}^{+}$, $\omega \in \Omega $, $%
\lambda \in \mathbb{R}_{0}^{+}$ and $t\geq t_{0}$.

Then, we use Lemma \ref{bound incr 1 Lemma copy(9)} in (\ref{electric power1}%
) to obtain%
\begin{eqnarray*}
&&\rho _{t}^{(\beta ,\omega ,\lambda ,\eta \mathbf{A}_{l})}(W_{t}^{\eta
\mathbf{A}_{l}})-\varrho ^{(\beta ,\omega ,\lambda )}(W_{t}^{\eta \mathbf{A}%
_{l}}) \\
&=&\frac{\eta ^{2}}{4}\underset{\mathbf{x},\mathbf{y}\in \mathfrak{K}}{\sum }%
\left( \int\nolimits_{t_{0}}^{t}\mathbf{E}_{s}^{\mathbf{A}_{l}}(\mathbf{x})%
\mathrm{d}s\right) \int_{t_{0}}^{t}\mathrm{d}s_{1}\left(
\int\nolimits_{t_{0}}^{s_{1}}\mathbf{E}_{s_{2}}^{\mathbf{A}_{l}}(\mathbf{y})%
\mathrm{d}s_{2}\right) \\
&&\times \varrho ^{(\beta ,\omega ,\lambda )}\left( i[\tau _{s_{1}}^{(\omega
,\lambda )}\left( I_{\mathbf{y}}\right) ,\tau _{t}^{(\omega ,\lambda
)}\left( I_{\mathbf{x}}\right) ]\right) +\mathcal{O}(\eta ^{3}l^{d})\ ,
\end{eqnarray*}%
uniformly for $\beta \in \mathbb{R}^{+}$, $\omega \in \Omega $, $\lambda \in
\mathbb{R}_{0}^{+}$ and $t\geq t_{0}$. We then obtain%
\begin{eqnarray}
&&\rho _{t}^{(\beta ,\omega ,\lambda ,\eta \mathbf{A}_{l})}(W_{t}^{\eta
\mathbf{A}_{l}})-\varrho ^{(\beta ,\omega ,\lambda )}(W_{t}^{\eta \mathbf{A}%
_{l}})  \label{electric power2} \\
&=&\frac{\eta ^{2}}{4}\underset{\mathbf{x},\mathbf{y}\in \mathfrak{K}}{\sum }%
\left( \int\nolimits_{t_{0}}^{t}\mathbf{E}_{s}^{\mathbf{A}_{l}}(\mathbf{x})%
\mathrm{d}s\right) \left( \int_{t_{0}}^{t}\zeta _{\mathbf{y},\mathbf{x}%
}^{(\omega )}\left( t,s\right) \mathbf{E}_{s}^{\mathbf{A}_{l}}(\mathbf{y})%
\mathrm{d}s\right) +\mathcal{O}(\eta ^{3}l^{d})\ ,  \notag
\end{eqnarray}%
by using (\ref{backwards -1}), (\ref{ptt remark}) and an integration by
parts. We now invoke Equation (\ref{backwards -2}) in (\ref{electric power2}%
) to arrive at the assertion.
\end{proof}

Therefore, Theorem \ref{Local Ohm's law thm copy(2)} (\textbf{Q}) and (%
\textbf{P}) are direct consequences of (\ref{entropic energy increment})--(%
\ref{electro free energy}), (\ref{lim_en_incr})--(\ref{lim_en_incr dia}),
Theorem \ref{Local Ohm's law thm} and Lemma \ref{bound incr 1 Lemma copy(5)}.

\subsubsection{Uniformity of Current Linear Response\label{Section Local
AC--Conductivity0}}

Following Section \ref{Sect Local Ohm law} we take $\vec{w}:=(w_{1},\ldots
,w_{d})\in \mathbb{R}^{d}$, $\mathcal{A}\in C_{0}^{\infty }\left( \mathbb{R};%
\mathbb{R}\right) $ and $\mathcal{E}_{t}:=-\partial _{t}\mathcal{A}_{t}$ for
any $t\in \mathbb{R}$, with $\mathcal{E}_{t}\vec{w}$ being the
space--homogeneous electric field. Then, $\mathbf{\bar{A}}\in \mathbf{C}%
_{0}^{\infty }$ is defined to be the electromagnetic potential such that the
value of the electric field equals $\mathcal{E}_{t}\vec{w}$ at time $t\in
\mathbb{R}$ for all $x\in \left[ -1,1\right] ^{d}$ and $(0,0,\ldots ,0)$ for
$t\in \mathbb{R}$ and $x\notin \left[ -1,1\right] ^{d}$. This choice yields
rescaled electromagnetic potentials $\eta \mathbf{\bar{A}}_{l}$ as defined
by (\ref{rescaled vector potential}) for $l\in \mathbb{R}^{+}$ and $\eta \in
\mathbb{R}$. Recall that $\mathcal{A}_{t}:=0$ for all $t\leq t_{0}$, where $%
t_{0}\in \mathbb{R}$ is any fixed starting time. We also recall that $%
\{e_{k}\}_{k=1}^{d}$ is the canonical orthonormal basis of the Euclidian
space $\mathbb{R}^{d}$.

In this case, the electromagnetic\emph{\ }potential energy observable
defined by (\ref{eq def W}) equals%
\begin{equation}
W_{t}^{\eta \mathbf{\bar{A}}_{l}}=-\sum\limits_{x\in \Lambda
_{l}}\sum\limits_{q\in \{1,\ldots ,d\}}2\mathrm{Re}\left[ \left( \exp \left(
-i\eta w_{q}\int_{t_{0}}^{t}\mathcal{E}_{s}\ \mathrm{d}s\right) -1\right)
a_{x}^{\ast }a_{x+e_{q}}\right] \in \mathcal{U}
\label{symmetric time--depending derivation2}
\end{equation}%
for any $l\in \mathbb{R}^{+}$, $\eta \in \mathbb{R}$, $\vec{w}%
:=(w_{1},\ldots ,w_{d})\in \mathbb{R}^{d}$, $\mathcal{A}\in C_{0}^{\infty
}\left( \mathbb{R};\mathbb{R}\right) $ and $t\in {\mathbb{R}}$.

The full current density is the sum of the paramagnetic and diamagnetic
currents $\mathbb{J}_{\mathrm{p}}^{(\omega ,\eta \mathbf{\bar{A}}_{l})}$\
and $\mathbb{J}_{\mathrm{d}}^{(\omega ,\eta \mathbf{\bar{A}}_{l})}$ that are
respectively defined by (\ref{finite volume current density}) and (\ref%
{finite volume current density2}). These currents are directly related to
the transport coefficients $\Xi _{\mathrm{p},l}^{(\omega )}$ and $\Xi _{%
\mathrm{d},l}^{(\omega )}$ (cf. (\ref{average microscopic AC--conductivity}%
)--(\ref{average microscopic AC--conductivity dia})). We show this in two
lemmata that yield Theorem \ref{thm Local Ohm's law}:

\begin{lemma}[Linear response of paramagnetic currents]
\label{Lemma LR para}\mbox{
}\newline
For any $\vec{w}:=(w_{1},\ldots ,w_{d})\in \mathbb{R}^{d}$\ and $\mathcal{A}%
\in C_{0}^{\infty }\left( \mathbb{R};\mathbb{R}\right) $, there is $\eta
_{0}\in \mathbb{R}^{+}$ such that, for $|\eta |\in \lbrack 0,\eta _{0}]$,
\begin{equation*}
\mathbb{J}_{\mathrm{p}}^{(\omega ,\eta \mathbf{\bar{A}}_{l})}\left( t\right)
=\eta \int_{t_{0}}^{t}\left( \Xi _{\mathrm{p},l}^{(\omega )}\left(
t-s\right) \vec{w}\right) \mathcal{E}_{s}\mathrm{d}s+\mathcal{O}\left( \eta
^{2}\right) \ ,
\end{equation*}%
uniformly for $l,\beta \in \mathbb{R}^{+}$, $\omega \in \Omega $, $\lambda
\in \mathbb{R}_{0}^{+}$ and $t\geq t_{0}$.
\end{lemma}

\begin{proof}
The first assertion is proven by essentially the same arguments as in
Section \ref{Local AC--Ohm's Law}. Indeed, one uses the stationarity (\ref%
{stationary}) of the $(\tau ^{(\omega ,\lambda )},\beta )$--KMS state $%
\varrho ^{(\beta ,\omega ,\lambda )}$, Dyson--Phillips expansions (\ref%
{Dyson tau 1}) for the non--autonomous dynamics, Lem%
%TCIMACRO{\TeXButton{\-}{\-}}%
%BeginExpansion
\-%
%EndExpansion
ma \ref{bound incr 1 Lemma copy(9)}, and tree--decay bounds on
multi--commutators \cite[Corollary 4.3]{OhmI} as in \cite[Lemma 5.10]{OhmI}
in order to deduce from (\ref{finite volume current density}) the existence
of $\eta _{0}\in \mathbb{R}^{+}$ such that, for $|\eta |\in \lbrack 0,\eta
_{0}]$,
\begin{equation*}
\left\{ \mathbb{J}_{\mathrm{p}}^{(\omega ,\eta \mathbf{\bar{A}}_{l})}\left(
t\right) \right\} _{k}=\frac{1}{\left\vert \Lambda _{l}\right\vert }%
\int_{t_{0}}^{t}\varrho ^{(\beta ,\omega ,\lambda )}\left( i[\tau
_{s}^{(\omega ,\lambda )}(W_{s}^{\eta \mathbf{\bar{A}}_{l}}),\tau
_{t}^{(\omega ,\lambda )}(\mathbb{I}_{k,l})]\right) \mathrm{d}s+\mathcal{O}%
\left( \eta ^{2}\right) \ ,
\end{equation*}%
uniformly for all $l,\beta \in \mathbb{R}^{+}$, $\omega \in \Omega $, $%
\lambda \in \mathbb{R}_{0}^{+}$, $k\in \{1,\ldots ,d\}$ and $t\in \mathbb{R}$%
. Then, for $|\eta |\in \lbrack 0,\eta _{0}]$, we employ (\ref{symmetric
time--depending derivation2}) and derive an assertion similar to Lemma \ref%
{bound incr 1 Lemma copy(9)} in order to get%
\begin{multline*}
\left\{ \mathbb{J}_{\mathrm{p}}^{(\omega ,\eta \mathbf{\bar{A}}_{l})}\left(
t\right) \right\} _{k} \\
=\frac{\eta }{\left\vert \Lambda _{l}\right\vert }\sum\limits_{q\in
\{1,\ldots ,d\}}\int_{t_{0}}^{t}\mathrm{d}s_{1}\int\nolimits_{t_{0}}^{s_{1}}%
\mathrm{d}s_{2}\ \mathcal{E}_{s_{2}}w_{q}\ \varrho ^{(\beta ,\omega ,\lambda
)}\left( i[\tau _{s_{1}}^{(\omega ,\lambda )}(\mathbb{I}_{q,l}),\tau
_{t}^{(\omega ,\lambda )}\left( \mathbb{I}_{k,l}\right) ]\right) \\
+\mathcal{O}\left( \eta ^{2}\right) \ ,
\end{multline*}%
uniformly for all $l,\beta \in \mathbb{R}^{+}$, $\omega \in \Omega $, $%
\lambda \in \mathbb{R}_{0}^{+}$, $k\in \{1,\ldots ,d\}$ and $t\in \mathbb{R}$%
. It follows from an integration by parts that%
\begin{multline*}
\left\{ \mathbb{J}_{\mathrm{p}}^{(\omega ,\eta \mathbf{\bar{A}}_{l})}\left(
t\right) \right\} _{k} \\
=\frac{\eta }{\left\vert \Lambda _{l}\right\vert }\int_{t_{0}}^{t}\sum%
\limits_{q\in \{1,\ldots ,d\}}\left( \int\nolimits_{t}^{s_{1}}\varrho
^{(\beta ,\omega ,\lambda )}\left( i[\mathbb{I}_{k,l},\tau
_{s_{2}-t}^{(\omega ,\lambda )}(\mathbb{I}_{q,l})]\right) \mathrm{d}%
s_{2}\right) w_{q}\mathcal{E}_{s_{1}}\ \mathrm{d}s_{1} \\
+\mathcal{O}\left( \eta ^{2}\right) \ ,
\end{multline*}%
which, combined with (\ref{average microscopic AC--conductivity}) and (\ref%
{average time reversal symetry}), yields the assertion.
\end{proof}

\begin{lemma}[Linear response of diamagnetic currents]
\label{Lemma LR dia}\mbox{
}\newline
For any $\vec{w}:=(w_{1},\ldots ,w_{d})\in \mathbb{R}^{d}$\ and $\mathcal{A}%
\in C_{0}^{\infty }\left( \mathbb{R};\mathbb{R}\right) $, there is $\eta
_{0}\in \mathbb{R}^{+}$ such that, for $|\eta |\in \lbrack 0,\eta _{0}]$,
\begin{equation*}
\mathbb{J}_{\mathrm{d}}^{(\omega ,\eta \mathbf{\bar{A}}_{l})}\left( t\right)
=\eta \left( \Xi _{\mathrm{d},l}^{(\omega )}\vec{w}\right) \int_{t_{0}}^{t}%
\mathcal{E}_{s}\mathrm{d}s+\mathcal{O}\left( \eta ^{2}\right) \ ,
\end{equation*}%
uniformly for $l,\beta \in \mathbb{R}^{+}$, $\omega \in \Omega $, $\lambda
\in \mathbb{R}_{0}^{+}$ and $t\geq t_{0}$.
\end{lemma}

\begin{proof}
By (\ref{stationary}), for any $k\in \{1,\ldots ,d\}$, note that%
\begin{eqnarray}
\left\{ \mathbb{J}_{\mathrm{d}}^{(\omega ,\eta \mathbf{\bar{A}}_{l})}\left(
t\right) \right\} _{k} &=&\frac{1}{\left\vert \Lambda _{l}\right\vert }%
\varrho ^{(\beta ,\omega ,\lambda )}\left( (\tau _{t,t_{0}}^{(\omega
,\lambda ,\eta \mathbf{\bar{A}}_{l})}-\tau _{t-t_{0}}^{(\omega ,\lambda )})(%
\mathbf{I}_{k,l}^{\eta \mathbf{A}_{l}})\right)  \notag \\
&&+\frac{1}{\left\vert \Lambda _{l}\right\vert }\varrho ^{(\beta ,\omega
,\lambda )}(\mathbf{I}_{k,l}^{\eta \mathbf{A}_{l}})\ ,
\label{dia current proof 1}
\end{eqnarray}%
while
\begin{equation}
\mathbf{I}_{k,l}^{\eta \mathbf{A}_{l}}=-\eta w_{k}\left( \int_{t_{0}}^{t}%
\mathcal{E}_{s}\mathrm{d}s\right) \underset{x\in \Lambda _{l}}{\sum }\left(
a_{x+e_{k}}^{\ast }a_{x}+a_{x}^{\ast }a_{x+e_{k}}\right) +\mathcal{O}(\eta
^{2}l^{d})\ ,  \label{dia current proof 2}
\end{equation}%
uniformly for all $\beta \in \mathbb{R}^{+}$, $\omega \in \Omega $, $\lambda
\in \mathbb{R}_{0}^{+}$, $k\in \{1,\ldots ,d\}$ and $t\in \mathbb{R}$.
Therefore, using again Dyson--Phillips expansions (\ref{Dyson tau 1}) for
the non--autonomous dynamics, Lem%
%TCIMACRO{\TeXButton{\-}{\-}}%
%BeginExpansion
\-%
%EndExpansion
ma \ref{bound incr 1 Lemma copy(9)}, and tree--decay bounds on
multi--commutators \cite[Corollary 4.3]{OhmI} one deduces the existence of $%
\eta _{0}\in \mathbb{R}^{+}$ such that, for $|\eta |\in \lbrack 0,\eta _{0}]$%
, the first term in the right hand side of (\ref{dia current proof 1}) is of
order $\mathcal{O}\left( \eta ^{2}\right) $, uniformly for $l,\beta \in
\mathbb{R}^{+}$, $\omega \in \Omega $, $\lambda \in \mathbb{R}_{0}^{+}$, $%
k\in \{1,\ldots ,d\}$ and $t\geq t_{0}$. Then the assertion follows by
combining this property with (\ref{average microscopic AC--conductivity dia}%
) and (\ref{dia current proof 1})--(\ref{dia current proof 2}).
\end{proof}

\appendix

\section{Duhamel Two--Point Functions\label{Section Duhamel Two--Point
Functions}}

\subsection{Duhamel Two--Point Function on the CAR Algebra}

The Duhamel two--point function $(\cdot ,\cdot )_{\sim }^{(\omega )}$ is
defined by (\ref{def bogo jetman}), that is,
\begin{equation}
(B_{1},B_{2})_{\sim }^{(\omega )}\equiv (B_{1},B_{2})_{\sim }^{(\beta
,\omega ,\lambda )}:=\int\nolimits_{0}^{\beta }\varrho ^{(\beta ,\omega
,\lambda )}\left( B_{1}^{\ast }\tau _{i\alpha }^{(\omega ,\lambda
)}(B_{2})\right) \mathrm{d}\alpha  \label{def bogo jetmanbis}
\end{equation}%
for any $B_{1},B_{2}\in \mathcal{U}$. Its name comes from the clear relation
to Duhamel's formula, see \cite[Section IV.4]{simon} for more details. This
sesquilinear form appears in different contexts. For instance, it has been
used by Bogoliubov \cite{bog1} for finite volume quantum systems in quantum
statistical mechanics. It is an useful tool in the first mathematical
justification -- by Ginibre \cite{Ginibre} in 1968 -- of the Bogoliubov
approximation for the Bose gas. This sesquilinear form is also used in the
context of linear response theory, see for instance \cite[Discussion after
Lemma 5.3.16 and Section 5.4]{BratteliRobinson}. In fact, it is also named
in the literature Bogoliubov or Kubo--Mori \emph{scalar product} as well as
the canonical correlation. A detailed analysis of this sesquilinear form for
KMS states has been performed by Naudts, Verbeure and Weder in the paper
\cite{Nau2}. Their aim was to extend to infinite systems some results of
linear response theory initiated by Kubo \cite{kubo} and Mori \cite{mori}.

Note that our definition of $(\cdot ,\cdot )_{\sim }$ is slightly different
from the usual one because of the missing normalization factor $\beta ^{-1}$
in front of the integral in (\ref{def bogo jetmanbis}). Discussions on
Duhamel two--point functions and examples of applications can also be found
in \cite{Mermin,Hohenberg,Falk,Nau3,roebstorff,Dyson}.

A first way to study this sesquilinear form is to use finite volume systems.
This is possible because, by using the Phragm{\'e}n--Lindel\"{o}f theorem \cite[Proposition 5.3.5]%
{BratteliRobinson} and \cite[Theorem A.3]{OhmI}, one checks that the formal
expression%
\begin{equation*}
\varrho ^{(\beta ,\omega ,\lambda )}\left( B^{\ast }\tau _{i\alpha
}^{(\omega ,\lambda )}(B)\right) =\varrho ^{(\beta ,\omega ,\lambda )}\left(
(\tau _{i\alpha /2}^{(\omega ,\lambda )}(B))^{\ast }\tau _{i\alpha
/2}^{(\omega ,\lambda )}(B)\right) \geq 0
\end{equation*}%
is correct for any $B\in \mathcal{U}$ and all $\alpha \in \lbrack 0,\beta ]$%
. So $(B_{1},B_{2})\mapsto (B_{1},B_{2})_{\sim }$ is a positive
semi--definite sesquilinear form on $\mathcal{U}$. It is however important
for the study of the conductivity measure to know that this form defines a
scalar product. To this end, we invoke the modular theory to have access to
functional calculus as it is done in the paper \cite{Nau2}.

\subsection{Duhamel Two--Point Functions on von Neumann Algebras \label%
{Section duhamal scalar}}

We consider in all the following subsections an arbitrary strongly
continuous one--parameter group $\tau :=\{\tau _{t}\}_{t\in {\mathbb{R}}}$
of automorphisms of a $C^{\ast }$--algebra $\mathcal{X}$ as well as an
arbitrary $(\tau ,\beta )$--KMS state $\varrho \in \mathcal{X}^{\ast }$ for
some $\beta >0$. Similar to (\ref{def bogo jetmanbis}), the Duhamel
two--point function $(\cdot ,\cdot )_{\sim }$ on the $C^{\ast }$--algebra $%
\mathcal{X}$ is defined by
\begin{equation}
(B_{1},B_{2})_{\sim }:=\int\nolimits_{0}^{\beta }\varrho \left( B_{1}^{\ast
}\tau _{i\alpha }(B_{2})\right) \mathrm{d}\alpha \ ,\qquad B_{1},B_{2}\in
\mathcal{X}\ .  \label{def bogo jetmanbisbis}
\end{equation}%
We have in mind the example $\mathcal{X}=\mathcal{U}$, $\tau =\tau ^{(\omega
,\lambda )}$ and $\varrho =\varrho ^{(\beta ,\omega ,\lambda )}$ for $\beta
\in \mathbb{R}^{+}$, $\omega \in \Omega $ and $\lambda \in \mathbb{R}%
_{0}^{+} $, of course.

The GNS representation of $\varrho $ is denoted by $(\mathcal{H},\pi ,\Psi )$%
. There is a unique normal state of the von Neumann algebra $\mathfrak{M}%
:=\pi \left( \mathcal{X}\right) ^{\prime \prime }$, also denoted by $\varrho
\in \mathfrak{M}^{\ast }$ to simplify notation, with $\rho =\rho \circ \pi $
on $\mathcal{X}$. By \cite[Corollary 5.3.4]{BratteliRobinson}, there is a
unique $\sigma $--weakly continuous $\ast $--automorphism group on $%
\mathfrak{M}$, which is again denoted by $\tau =\{\tau _{t}\}_{t\in {\mathbb{%
R}}}$, such that $\tau _{t}\circ \pi =\pi \circ \tau _{t}$, $t\in \mathbb{R}$%
, on $\mathcal{X}$. Moreover, the normal state $\varrho \in \mathfrak{M}%
^{\ast }$ is a $(\tau ,\beta )$--KMS state on $\mathfrak{M}$ and it thus
satisfies the KMS (or modular) condition, that is, for any $b_{1},b_{2}\in
\mathfrak{M}$, the map
\begin{equation*}
t\mapsto \mathfrak{m}_{b_{1},b_{2}}\left( t\right) :=\varrho (b_{1}\tau
_{t}(b_{2}))=\left\langle \Psi ,b_{1}\tau _{t}(b_{2})\Psi \right\rangle _{%
\mathcal{H}}
\end{equation*}%
from ${\mathbb{R}}$ to $\mathbb{C}$ extends uniquely to a continuous map $%
\mathfrak{m}_{b_{1},b_{2}}$ on ${\mathbb{R}}\times \lbrack 0,\beta ]\subset {%
\mathbb{C}}$ which is holomorphic on ${\mathbb{R}}\times (0,\beta )$ whereas
\begin{equation*}
\mathfrak{m}_{b_{1},b_{2}}\left( i\beta \right) =\varrho (b_{2}b_{1})\
,\qquad b_{1},b_{2}\in \mathfrak{M}\ .
\end{equation*}%
Here, $\left\langle \cdot ,\cdot \right\rangle _{\mathcal{H}}$ denotes the
scalar product of the Hilbert space $\mathcal{H}$. See, e.g., \cite[%
Proposition 5.3.7]{BratteliRobinson}.

Because $\varrho $ is invariant with respect to $\tau $, the $\ast $%
--automorphism group $\tau $ has a unique representation by conjugation with
unitaries $\{U_{t}\}_{t\in \mathbb{R}}\subset \mathfrak{M}$, i.e.,
\begin{equation*}
\tau _{t}\left( b\right) =U_{t}bU_{t}^{\ast }\ ,\qquad t\in \mathbb{R}\ ,\
b\in \mathfrak{M}\ ,
\end{equation*}%
such that $U_{t}\Psi =\Psi $. As $t\mapsto \tau _{t}$ is $\sigma $--weakly
continuous, the map $t\mapsto U_{t}$ is strongly continuous. Therefore, the
unitary group $\{U_{t}\}_{t\in \mathbb{R}}$ has an anti--self--adjoint
operator $i\mathcal{L}$ as generator, i.e., $U_{t}=\mathrm{e}^{it\mathcal{L}%
} $. In particular, $\Psi \in \mathrm{Dom}(\mathcal{L})$ and $\mathcal{L}$
annihilates $\Psi $, i.e., $\mathcal{L}\Psi =0$. The operator $\mathcal{L}$
is known in the literature as the \emph{standard Liouvillean} of $\tau $
associated with $\varrho $. The spectral theorem applied to the
self--adjoint operator $\mathcal{L}$ ensures the existence of a
projection--valued measure $E$ on the real line $\mathbb{R}$ such that
\begin{equation*}
\mathcal{L}=\int\nolimits_{\mathbb{R}}\nu \ \mathrm{d}E(\nu )\ .
\end{equation*}

We now use the (Tomita--Takesaki) modular objects $\Delta $, $J$ of the pair
$\left( \mathfrak{M},\Psi \right) $. In particular,%
\begin{equation}
J\mathbf{\Delta }^{1/2}\left( b\Psi \right) =b^{\ast }\Psi \ ,\qquad b\in
\mathfrak{M}\ .  \label{modular0}
\end{equation}%
By \cite[Proposition 5.11]{AttalJoyePillet2006a}, the modular operator $%
\mathbf{\Delta }$ is equal to
\begin{equation}
\mathbf{\Delta }=\exp \left( -\beta \mathcal{L}\right) =\int\nolimits_{%
\mathbb{R}}\mathrm{e}^{-\beta \nu }\mathrm{d}E(\nu )  \label{delta}
\end{equation}%
and $U_{t}=\mathbf{\Delta }^{-it\beta ^{-1}}$.

Now, let the (unbounded) positive operator $\mathfrak{T}$ acting on $%
\mathcal{H}$ be defined by%
\begin{equation}
\mathfrak{T}:=\beta ^{1/2}\int\nolimits_{\mathbb{R}}\left( \frac{1-\mathrm{e}%
^{-\beta \nu }}{\beta \nu }\right) ^{1/2}\mathrm{d}E(\nu )\ .
\label{operator bogoliubov}
\end{equation}%
Here,
\begin{equation*}
\frac{1-\mathrm{e}^{-\beta \cdot 0}}{\beta \cdot 0}:=1\ .
\end{equation*}%
The Duhamel two--point function $(\cdot ,\cdot )_{\sim }$ is directly
related to this operator:

\begin{satz}[Duhamel two--point function in the GNS representation]
\label{toto fluctbis copy(2)}\mbox{
}\newline
For any $B_{1},B_{2}\in \mathcal{X}$,
\begin{equation*}
(B_{1},B_{2})_{\sim }=\left\langle \mathfrak{T}\pi \left( B_{1}\right) \Psi ,%
\mathfrak{T}\pi \left( B_{2}\right) \Psi \right\rangle _{\mathcal{H}}\ .
\end{equation*}%
In particular, $(B_{1},B_{1})_{\sim }\geq 0$.
\end{satz}

\begin{proof}
The proof can be found in \cite[Theorem II.4]{Nau2}. Since it is short, we
give it here for completeness. Note first that, for any $b_{1},b_{2}\in
\mathfrak{M}$,
\begin{eqnarray*}
\left\langle \Psi ,b_{1}\mathbf{\Delta }^{1/2}b_{2}\Psi \right\rangle _{%
\mathcal{H}} &=&\left\langle \mathbf{\Delta }^{1/2}b_{1}^{\ast }\Psi
,b_{2}\Psi \right\rangle _{\mathcal{H}}=\left\langle Jb_{2}\Psi ,b_{1}\Psi
\right\rangle _{\mathcal{H}} \\
&=&\left\langle \mathbf{\Delta }^{1/2}J\mathbf{\Delta }^{1/2}b_{2}\Psi
,b_{1}\Psi \right\rangle _{\mathcal{H}}=\left\langle \Psi ,b_{2}\mathbf{%
\Delta }^{1/2}b_{1}\Psi \right\rangle _{\mathcal{H}}\ ,
\end{eqnarray*}%
where we have used $\mathbf{\Delta =\Delta }^{\ast }$, the anti--unitarity
of $J$, $J^{2}=\mathbf{1}$, and $J\mathbf{\Delta }^{1/2}J=\mathbf{\Delta }%
^{-1/2}$. Using this fact and properties of the map $\mathfrak{m}%
_{b_{1},b_{2}}$ from ${\mathbb{R}}\times \lbrack 0,\beta ]\subset {\mathbb{C}%
}$ to ${\mathbb{C}}$ together with the Phragm{\'e}n--Lindel\"{o}f theorem
\cite[Proposition 5.3.5]{BratteliRobinson} one shows that, for any $%
b_{1},b_{2}\in \mathfrak{M}$,
\begin{equation*}
\mathfrak{m}_{b_{1},b_{2}}\left( i\beta \alpha \right) =\left\{
\begin{array}{lll}
\left\langle \Psi ,b_{1}\mathbf{\Delta }^{\alpha }b_{2}\Psi \right\rangle _{%
\mathcal{H}} & , & \qquad \alpha \in \left[ 0,1/2\right] \ , \\
\left\langle \Psi ,b_{2}\mathbf{\Delta }^{1-\alpha }b_{1}\Psi \right\rangle
_{\mathcal{H}} & , & \qquad \alpha \in \left[ 1/2,1\right] \ .%
\end{array}%
\right.
\end{equation*}%
By (\ref{def bogo jetmanbisbis}) and (\ref{modular0}), it follows that
\begin{eqnarray}
(B_{1},B_{2})_{\sim } &=&\beta \int\nolimits_{0}^{1/2}\left\langle \pi
\left( B_{1}\right) \Psi ,\mathbf{\Delta }^{\alpha }\pi \left( B_{2}\right)
\Psi \right\rangle _{\mathcal{H}}\mathrm{d}\alpha  \label{duhamel1} \\
&&+\beta \int\nolimits_{0}^{1/2}\left\langle J\mathbf{\Delta }^{1/2}\pi
\left( B_{2}\right) \Psi ,\mathbf{\Delta }^{\alpha }J\mathbf{\Delta }%
^{1/2}\pi \left( B_{1}\right) \Psi \right\rangle _{\mathcal{H}}\mathrm{d}%
\alpha \ .  \notag
\end{eqnarray}%
Because $J^{2}=\mathbf{1}$, $J\mathbf{\Delta }^{\alpha }J=\mathbf{\Delta }%
^{-\alpha }$ and $J$ is anti--unitary, note that
\begin{eqnarray*}
&&\left\langle J\mathbf{\Delta }^{1/2}\pi \left( B_{2}\right) \Psi ,\mathbf{%
\Delta }^{\alpha }J\mathbf{\Delta }^{1/2}\pi \left( B_{1}\right) \Psi
\right\rangle _{\mathcal{H}} \\
&=&\left\langle J\mathbf{\Delta }^{\alpha }J\mathbf{\Delta }^{1/2}\pi \left(
B_{1}\right) \Psi ,\mathbf{\Delta }^{1/2}\pi \left( B_{2}\right) \Psi
\right\rangle _{\mathcal{H}} \\
&=&\left\langle \mathbf{\Delta }^{-\alpha }\mathbf{\Delta }^{1/2}\pi \left(
B_{1}\right) \Psi ,\mathbf{\Delta }^{1/2}\pi \left( B_{2}\right) \Psi
\right\rangle _{\mathcal{H}}
\end{eqnarray*}%
for all $\alpha \in \left[ 0,1/2\right] $. Therefore, we deduce from (\ref%
{operator bogoliubov}) and (\ref{duhamel1}) that
\begin{equation*}
(B_{1},B_{2})_{\sim }=\beta \left\langle \pi \left( B_{1}\right) \Psi ,\frac{%
\mathbf{\Delta }-\mathbf{1}}{\ln \mathbf{\Delta }}\pi \left( B_{2}\right)
\Psi \right\rangle _{\mathcal{H}}=\left\langle \mathfrak{T}\pi \left(
B_{1}\right) \Psi ,\mathfrak{T}\pi \left( B_{2}\right) \Psi \right\rangle _{%
\mathcal{H}}\ ,
\end{equation*}%
using that%
\begin{equation*}
\int\nolimits_{0}^{1/2}\mathbf{\Delta }^{\alpha }b\Psi \ \mathrm{d}\alpha =%
\frac{\mathbf{\Delta }^{1/2}-\mathbf{1}}{\ln \mathbf{\Delta }}b\Psi \text{%
\qquad and\qquad }\int\nolimits_{0}^{1/2}\mathbf{\Delta }^{-\alpha }b\Psi \
\mathrm{d}\alpha =\frac{\mathbf{1}-\mathbf{\Delta }^{-1/2}}{\ln \mathbf{%
\Delta }}b\Psi
\end{equation*}%
for any $b\in \mathfrak{M}$.
\end{proof}

By (\ref{operator bogoliubov}), one checks that $\mathrm{Dom}(\mathbf{\Delta
}^{1/2})\subset \mathrm{Dom}(\mathfrak{T})$ and thus, $\mathfrak{M}\Psi
\subset \mathrm{Dom}(\mathfrak{T})$. It is therefore natural to define the
Duhamel two--point function, again denoted by $(\cdot ,\cdot )_{\sim }$, on
the von Neumann algebra $\mathfrak{M}:=\pi \left( \mathcal{X}\right)
^{\prime \prime }$ by
\begin{equation}
(b_{1},b_{2})_{\sim }:=\left\langle \mathfrak{T}b_{1}\Psi ,\mathfrak{T}%
b_{2}\Psi \right\rangle _{\mathcal{H}}\ ,\qquad b_{1},b_{2}\in \mathfrak{M}\
.  \label{def bogo jetmanbis neunmann}
\end{equation}%
This sesquilinear form is a scalar product:

\begin{satz}[Duhamel two--point function as a scalar product]
\label{toto fluctbis copy(3)}\mbox{
}\newline
The sesquilinear form $(\cdot ,\cdot )_{\sim }$ is a scalar product of the
pre--Hilbert space $\mathfrak{M}$.
\end{satz}

\begin{proof}
The positivity of the sesquilinear form $(\cdot ,\cdot )_{\sim }$ is clear.
Therefore, it only remains to verify that it is non--degenerated. This is
proven in \cite[Lemma II.2.]{Nau2} as follows: First note that $0$ is not an
eigenvalue of $\mathfrak{T}$. This follows from (\ref{operator bogoliubov}).
Indeed, for all $\nu \in \mathbb{R}$,
\begin{equation*}
\left( \frac{1-\mathrm{e}^{-\beta \nu }}{\beta \nu }\right) ^{1/2}>0\ .
\end{equation*}%
Since $\varrho $ is a $(\tau ,\beta )$--KMS state, the cyclic vector $\Psi $
is also separating for $\mathfrak{M}$, by \cite[Corollary 5.3.9.]%
{BratteliRobinson}. Therefore, $(b,b)_{\sim }=0$ yields $\mathfrak{T}b\Psi
=0 $ which in turn implies that $b\Psi =0$ and $b=0$.
\end{proof}

Note that the kernel of $\pi $ is a closed two--sided ideal. If the $C^{\ast
}$--algebra $\mathcal{X}$ is simple (like $\mathcal{U}$), i.e., when $\{0\}$
and $\mathcal{X}$ are the only closed two--sided ideals, it then follows
that
\begin{equation*}
\ker \left( \pi \right) =\{0\}.
\end{equation*}
Using this and Theorem \ref{toto fluctbis copy(3)} we deduce that the
Duhamel two--point function (\ref{def bogo jetmanbisbis}) for $%
B_{1}=B_{2}\in \mathcal{X}\backslash \{0\}$ is never zero:

\begin{satz}[Duhamel two--point function -- Strict positivity]
\label{Thm stric positivityI copy(1)}\mbox{
}\newline
If the $C^{\ast }$--algebra $\mathcal{X}$ is simple then $(B,B)_{\sim }>0$
for all non--zero $B\in \mathcal{X}\backslash \{0\}$.
\end{satz}

Finally, we observe that it is a priori not clear that the scalar products $%
(\cdot ,\cdot )_{\sim }$ and $\langle \cdot ,\cdot \rangle _{\mathcal{H}}$
are related to each other via some upper or lower bounds. In fact, a
combination of Roepstorff's results \cite[Eq. (10)]{roebstorff} for finite
dimensional systems with those of Naudts and Verbeure on von Neumann
Algebras yields the so--called \emph{auto--correlation upper bounds} \cite[%
Theorem III.1]{Nau3}, also called Roepstorff's inequality. For self--adjoint
observables, these upper bounds read:

\begin{satz}[Auto--correlation upper bounds for observables]
\label{thm auto--correlation upper bounds}\mbox{
}\newline
For any self--adjoint element $b=b^{\ast }\in \mathfrak{M}$, $(b,b)_{\sim
}\leq \left\langle b\Psi ,b\Psi \right\rangle _{\mathcal{H}}$. In
particular, for all $B=B^{\ast }\in \mathcal{X}$,
\begin{equation*}
(B,B)_{\sim }\leq \varrho (B^{2})\leq \left\Vert B\right\Vert _{\mathcal{X}%
}^{2}\ .
\end{equation*}
\end{satz}

\begin{proof}
This theorem is a particular case of \cite[Theorem III.1]{Nau3}, by
observing in its proof that $(u-v)\log (u/v)$ should be replaced by $u$ when
$u=v$. See also \cite[Theorem 5.3.17]{BratteliRobinson}.
\end{proof}

Note that the authors derive in \cite{roebstorff,Nau3} further upper and
lower bounds related the scalar products $(\cdot ,\cdot )_{\sim }$ and $%
\langle \cdot ,\cdot \rangle _{\mathcal{H}}$. These are however not used in
the sequel. For more details, we refer to \cite{Nau3} or \cite[Section 5.3.1]%
{BratteliRobinson}. We only conclude this subsection by an important
equality for the Duhamel two--point function $(\cdot ,\cdot )_{\sim }$ which
was widely used for finite volume systems. See, e.g., \cite[Eq.\ (2.4)]%
{Ginibre}.

This equality does not seem to be proven before for general KMS states. It
is a straighforward consequence of Theorem \ref{toto fluctbis copy(2)}. To
this end, denote by $\delta $ the generator of the strongly continuous
one--parameter group $\tau :=\{\tau _{t}\}_{t\in {\mathbb{R}}}$ of
automorphisms of the $C^{\ast }$--algebra $\mathcal{X}$.

\begin{satz}[Commutators and Duhamel two--point function]
\label{thm auto--correlation upper bounds copy(1)}\mbox{
}\newline
For all $B_{1}\in \mathcal{X}$ and $B_{2}\in \mathrm{Dom}(\delta )$,
\begin{equation*}
-i(B_{1},\delta \left( B_{2}\right) )_{\sim }=\varrho \left( \left[
B_{1}^{\ast },B_{2}\right] \right) \ .
\end{equation*}
\end{satz}

\begin{proof}
It is a direct consequence of (\ref{modular0})--(\ref{operator bogoliubov})
and (\ref{def bogo jetmanbis neunmann}): For any $B_{1}\in \mathcal{X}$ and $%
B_{2}\in \mathrm{Dom}(\delta )$,%
\begin{eqnarray*}
-i(B_{1},\delta \left( B_{2}\right) )_{\sim } &=&\left\langle \mathfrak{T}%
\pi \left( B_{1}\right) \Psi ,\mathfrak{T}\pi \left( \delta \left(
B_{2}\right) \right) \Psi \right\rangle _{\mathcal{H}} \\
&=&\left\langle \pi \left( B_{1}\right) \Psi ,\pi \left( B_{2}\right) \Psi
\right\rangle _{\mathcal{H}}-\left\langle \mathbf{\Delta }^{1/2}\pi \left(
B_{1}\right) \Psi ,\mathbf{\Delta }^{1/2}\pi \left( B_{2}\right) \Psi
\right\rangle _{\mathcal{H}} \\
&=&\left\langle \pi \left( B_{1}\right) \Psi ,\pi \left( B_{2}\right) \Psi
\right\rangle _{\mathcal{H}}-\left\langle \pi \left( B_{2}^{\ast }\right)
\Psi ,\pi \left( B_{1}^{\ast }\right) \Psi \right\rangle _{\mathcal{H}} \\
&=&\varrho \left( \left[ B_{1}^{\ast },B_{2}\right] \right) \ .
\end{eqnarray*}%
See also Theorem \ref{toto fluctbis copy(2)}.
\end{proof}

\begin{koro}[Duhamel two--point function and generator of dynamics]
\label{Corollary Stationarity copy(2)}\mbox{ }\newline
For any self--adjoint element $B=B^{\ast }\in \mathrm{Dom}(\delta )\subset
\mathcal{X}$,
\begin{equation*}
(B,\delta \left( B\right) )_{\sim }=0\qquad \text{and}\qquad -i\varrho
\left( \left[ \delta \left( B\right) ,B\right] \right) =(\delta \left(
B\right) ,\delta \left( B\right) )_{\sim }\geq 0\ .
\end{equation*}
\end{koro}

\subsection{Duhamel GNS Representation}

In view of Theorem \ref{toto fluctbis copy(3)}, we denote by $\mathcal{%
\tilde{H}}$\ the completion of $\mathfrak{M}$ w.r.t. the scalar product $%
(\cdot ,\cdot )_{\sim }$. This Hilbert space is related to the GNS Hilbert
space of $\varrho $ by a unitary transformation:

\begin{satz}[Unitary equivalence of $\mathcal{H}$ and $\mathcal{\tilde{H}}$]

\label{toto fluctbis copy(4)}\mbox{
}\newline
$U_{\sim }\mathcal{\tilde{H}}=\mathcal{H}$ with $U_{\sim }$ being the
unitary operator defined by $U_{\sim }b=\mathfrak{T}b\Psi $ for $b\in
\mathfrak{M}$.
\end{satz}

\begin{proof}
Since $\Vert U_{\sim }b\Vert _{\mathcal{H}}=\Vert b\Vert _{\sim }$, the
operator $U_{\sim }$ defined by $U_{\sim }b=\mathfrak{T}b\Psi $ for $b\in
\mathfrak{M}$ has a continuous isometric extension on $\mathcal{\tilde{H}}$.
Then, one checks that the range of $\mathfrak{T}$ is dense in $\mathcal{H}$
and is included in the range of $U_{\sim }$. For more details, see \cite[%
Theorem II.3.]{Nau2}.
\end{proof}

A simple consequence of Theorem \ref{toto fluctbis copy(4)} is a cyclic
representation based on the Duhamel two--point function:

\begin{definition}[Duhamel GNS representation]
\label{GNS Duhamel}\mbox{ }\newline
The Duhamel GNS representation of the $(\tau ,\beta )$--KMS state $\varrho
\in \mathcal{X}^{\ast }$ is defined by the triplet $(\mathcal{\tilde{H}},%
\tilde{\pi},\tilde{\Psi})$ where
\begin{equation*}
\tilde{\Psi}:=U_{\sim }^{\ast}\Psi =U_{\sim }^{\ast}\mathfrak{T}\Psi \in
\mathcal{\tilde{H}}\text{\qquad and\qquad }\tilde{\pi}\left( B\right)
=U_{\sim }^{\ast}\pi \left( B\right) U_{\sim }\ ,\text{\quad }B\in \mathcal{X%
}\ .
\end{equation*}%
If $\mathcal{X}$ has an identity $\mathbf{1}$, then $\tilde{\Psi}=\pi (%
\mathbf{1})\in \mathfrak{M}\subset \mathcal{\tilde{H}}$.
\end{definition}

\noindent This cyclic representation of KMS states does not seem -- at least
to our knowledge -- to have been previously used, even if it is a direct
consequence of \cite[Theorem II.3.]{Nau2}. In particular, the name \emph{%
Duhamel GNS\ representation} is not standard and it could also be called
\emph{Bogoliubov} or \emph{Kubo--Mori} \emph{GNS} representation in
reference to the scalar product $(\cdot ,\cdot )_{\sim }$.

As explained in Section \ref{Section duhamal scalar}, there is a unique $%
\sigma $--weakly continuous $\ast $--automor\-phism group $\tilde{\tau}=\{%
\tilde{\tau}_{t}\}_{t\in \mathbb{R}}$ on the von Neumann algebra $\mathfrak{%
\tilde{M}}:=\tilde{\pi}\left( \mathcal{X}\right) ^{\prime \prime }$, such
that $\tau_t = \tilde{\tau}_t \circ \pi$, $t \in \mathbb{R}$. It has a
representation by conjugation with unitaries
\begin{equation*}
\{\mathrm{e}^{it\mathcal{\tilde{L}}}\}_{t\in \mathbb{R}}\subset \mathfrak{M},
\end{equation*}%
the self--adjoint operator $\mathcal{\tilde{L}}$ being equal to%
\begin{equation}
\mathcal{\tilde{L}}=U_{\sim }^{\ast}\mathcal{L}U_{\sim }\ .
\label{toto oublie}
\end{equation}%
Clearly, $\tilde{\Psi}\in \mathrm{Dom}(\mathcal{\tilde{L}})$ and $\mathcal{%
\tilde{L}}\tilde{\Psi}=0$. The normal state $\tilde{\varrho}\in \mathfrak{%
\tilde{M}}^{\ast }$ is a $(\tilde{\tau},\beta )$--KMS state.

At the end of the previous subsection we explain that if the $C^{\ast }$%
--algebra $\mathcal{X}$ is simple, like the CAR algebra $\mathcal{U}$, then $%
\pi :\mathcal{X}\rightarrow \mathfrak{M}$ is injective and one can see the $%
C^{\ast }$--algebra $\mathcal{X}$ as a \emph{subspace} of $\mathcal{\tilde{H}%
}$. In particular, if $\mathcal{X}$ has an identity $\mathbf{1}$, then
\begin{equation*}
\tilde{\Psi}=\mathbf{1}\in \mathcal{X}\subset \mathfrak{M}\subset \mathcal{%
\tilde{H}}\ .
\end{equation*}%
Note additionally that, in this case, for any element $B\in \mathcal{X}$ and
time $t\in \mathbb{R}$, one has $\tau _{t}(B)\in \mathcal{X}\subset \mathcal{%
\tilde{H}}$ and it is straightforward to check (cf. \cite[Section III]{Nau2}%
) that $i\mathcal{\tilde{L}}$ is the generator of a unitary group extending $%
\tau $ to the whole Hilbert space $\mathcal{\tilde{H}}$:

\begin{satz}[Duhamel GNS representation and dynamics]
\label{Thm important equality asymptotics}\mbox{ }\newline
Assume $\mathcal{X}$ is simple. Then, for $B\in \mathcal{X}\subset \mathcal{%
\tilde{H}}$ and $t\in \mathbb{R}$, $\tau _{t}(B)=\mathrm{e}^{it\mathcal{%
\tilde{L}}}B$ with $(B,\mathcal{\tilde{L}}B)_{\sim }=0$ if $B\in \mathrm{Dom}%
(\mathcal{\tilde{L}})$.
\end{satz}

\begin{proof}
See \cite[Section III]{Nau2}: By Theorem\ \ref{toto fluctbis copy(4)}, for
any $B\in \mathcal{X}\subset \mathfrak{M}\subset \mathcal{\tilde{H}}$ and $%
t\in \mathbb{R}$,
\begin{eqnarray*}
\tau _{t}(B) &=&U_{\sim }^{\ast }\mathfrak{T}\pi \left( \tau _{t}(B)\right)
\Psi =U_{\sim }^{\ast }\mathfrak{T}\mathrm{e}^{it\mathcal{L}}\pi \left(
B\right) \Psi \\
&=&U_{\sim }^{\ast }\mathrm{e}^{it\mathcal{L}}\mathfrak{T}\pi \left(
B\right) \Psi =U_{\sim }^{\ast }\mathrm{e}^{it\mathcal{L}}U_{\sim }B=\mathrm{%
e}^{itU_{\sim }^{\ast }\mathcal{L}U_{\sim }}B\ .
\end{eqnarray*}%
Recall that $(\mathcal{H},\pi ,\Psi )$ is the GNS representation of the $%
(\tau ,\beta )$--KMS state $\varrho $ and $\mathcal{L}$ is the associated
standard Liouvillean. See also (\ref{toto oublie}). The equality $(B,%
\mathcal{\tilde{L}}B)_{\sim }=0$ results from Corollary \ref{Corollary
Stationarity copy(2)}.
\end{proof}

Note that Theorem \ref{Thm important equality asymptotics} directly yields
the invariance of the norm of $B\in \mathcal{X}\subset \mathcal{\tilde{H}}$
w.r.t. to the group $\tau $ acting on the subspace $\mathcal{X}\subset
\mathcal{\tilde{H}}$.

\begin{koro}[Stationarity of the Duhamel norm]
\label{Corollary Stationarity}\mbox{ }\newline
Assume $\mathcal{X}$ is simple. Then, for $B\in \mathcal{X}\subset \mathcal{%
\tilde{H}}$ and $t\in \mathbb{R}$, $\Vert \tau _{t}(B)\Vert _{\sim }=\Vert
B\Vert _{\sim }$ with $\Vert \cdot \Vert _{\sim }$ denoting the (Duhamel)
norm of $\mathcal{\tilde{H}}$ associated with the scalar product $(\cdot
,\cdot )_{\sim }$.
\end{koro}

Therefore, by Theorem \ref{Thm important equality asymptotics}, we can
invoke the spectral theorem in order to analyze the dynamics in relation
with the scalar product $(\cdot ,\cdot )_{\sim }$. This is exploited for
instance in Theorem \ref{lemma sigma pos type copy(3)} to extract the
conductivity measure from a spectral measure.

\begin{bemerkung}[$\mathcal{U}$ as a pre--Hilbert space]
\label{Reminder}\mbox{ }\newline
We identify in all the paper the Duhamel two--point function $(\cdot ,\cdot
)_{\sim }$ defined by (\ref{def bogo jetmanbis}) on the CAR $C^{\ast }$%
--algebra $\mathcal{U}$ with the scalar product $(\cdot ,\cdot )_{\sim }$
defined by (\ref{def bogo jetmanbis neunmann}) for $\varrho =\varrho
^{(\beta ,\omega ,\lambda )}$ and $\tau =\tau ^{(\omega ,\lambda )}$ on $%
\mathfrak{M}:=\pi \left( \mathcal{U}\right) ^{\prime \prime }\subset
\mathcal{\tilde{H}}$. Note that $\mathcal{U} \equiv \pi(\mathcal{U}) \subset
\mathfrak{M}$ is a pre--Hilbert space w.r.t. $(\cdot ,\cdot )_{\sim }$.
\end{bemerkung}

\subsection{Duhamel Two--Point Function and Time--Reversal Symmetry}

Let $\mathcal{X}$ be a $C^{\ast }$--algebra with unity $\mathbf{1}$ and
assume the existence of a map $\Theta :\mathcal{X}\rightarrow \mathcal{X}$
with the following properties:

\begin{itemize}
\item $\Theta $ is antilinear and continuous.

\item $\Theta \left( \mathbf{1}\right) =\mathbf{1}$ and $\Theta \circ \Theta
=\mathrm{Id}_{\mathcal{X}}$.

\item $\Theta \left( B_{1}B_{2}\right) =\Theta \left( B_{1}\right) \Theta
\left( B_{2}\right) $ for all $B_{1},B_{2}\in \mathcal{X}$.

\item $\Theta \left( B^{\ast }\right) =\Theta \left( B\right) ^{\ast }$ for
all $B\in \mathcal{X}$.
\end{itemize}

\noindent Such a map is called a \emph{time--reversal} operation of the $%
C^{\ast }$--algebra $\mathcal{X}$.

Observe that, for any strongly continuous one--parameter group $\tau
:=\{\tau _{t}\}_{t\in {\mathbb{R}}}$ of automorphisms of $\mathcal{X}$, the
family $\tau ^{\Theta }:=\{\tau _{t}^{\Theta }\}_{t\in {\mathbb{R}}}$
defined by
\begin{equation*}
\tau _{t}^{\Theta }:=\Theta \circ \tau _{t}\circ \Theta \ ,\qquad t\in {%
\mathbb{R}}\ ,
\end{equation*}%
is again a strongly continuous one--parameter group of automorphisms.
Similarly, for any state $\rho \in \mathcal{X}^{\ast }$, the linear
functional $\rho ^{\Theta }$ defined by%
\begin{equation*}
\rho ^{\Theta }\left( B\right) =\overline{\rho \circ \Theta \left( B\right) }%
\ ,\qquad B\in \mathcal{X}\ ,
\end{equation*}%
is again a state. We say that $\tau $\ and $\rho $ are \emph{time--reversal
invariant} if they satisfy $\tau _{t}^{\Theta }=\tau _{-t}$ for all $t\in
\mathbb{R}$ and $\rho ^{\Theta }=\rho $.

If $\tau $\ is time--reversal invariant then, for all $\beta >0$, there is
at least one time--reversal invariant $(\tau ,\beta )$--KMS state $\varrho
\in \mathcal{X}^{\ast }$, provided the set of $(\tau ,\beta )$--KMS states
is not empty. This follows from the convexity of the set of KMS states:

\begin{lemma}[Existence of time--reversal invariant $(\protect\tau ,\protect%
\beta )$--KMS states]
\label{Lemma LR dia copy(1)}\mbox{
}\newline
Assume that $\tau $\ is time--reversal invariant and $\varrho $ is a $(\tau
,\beta )$--KMS state. Then, $\rho ^{\Theta }$ is a $(\tau ,\beta )$--KMS
state. In particular, $\frac{1}{2}\rho +\frac{1}{2}\rho ^{\Theta }$ is a
time--reversal invariant $(\tau ,\beta )$--KMS state.
\end{lemma}

\begin{proof}
For any $t\in {\mathbb{R}}$ and $B_{1},B_{2}\in \mathcal{X}$,
\begin{equation*}
\rho ^{\Theta }\left( B_{1}\tau _{t}\left( B_{2}\right) \right) =\overline{%
\rho \left( \Theta \left( B_{1}\right) \tau _{-t}\left( \Theta \left(
B_{2}\right) \right) \right) }=\rho \left( \Theta \left( B_{2}^{\ast
}\right) \tau _{t}\left( \Theta \left( B_{1}^{\ast }\right) \right) \right)
\ ,
\end{equation*}%
using the stationarity of KMS--states and hermiticity of states. Since $\rho
$ is by assumption a $(\tau ,\beta )$--KMS state, the continuous function
\begin{equation*}
t\mapsto \mathfrak{m}_{B_{1},B_{2}}\left( t\right) :=\rho \left( \Theta
\left( B_{2}^{\ast }\right) \tau _{t}\left( \Theta \left( B_{1}^{\ast
}\right) \right) \right)
\end{equation*}%
from ${\mathbb{R}}$ to $\mathbb{C}$ extends uniquely to a continuous map $%
\mathfrak{m}_{B_{1},B_{2}}$ on ${\mathbb{R}}\times \lbrack 0,\beta ]\subset {%
\mathbb{C}}$ which is holomorphic on ${\mathbb{R}}\times (0,\beta )$ while,
again by stationarity and hermiticity of $\rho $,
\begin{eqnarray*}
\mathfrak{m}_{B_{1},B_{2}}\left( t+i\beta \right) &=&\rho \left( \tau
_{t}\left( \Theta \left( B_{1}^{\ast }\right) \right) \Theta \left(
B_{2}^{\ast }\right) \right) \\
&=&\rho \left( \Theta \left( B_{1}^{\ast }\right) \Theta \left( \tau
_{t}\left( B_{2}^{\ast }\right) \right) \right) =\rho ^{\Theta }\left( \tau
_{t}\left( B_{2}\right) B_{1}\right)
\end{eqnarray*}%
for any $t\in {\mathbb{R}}$ and $B_{1},B_{2}\in \mathcal{X}$. As a
consequence, $\rho ^{\Theta }$ is a $(\tau ,\beta )$--KMS state, see \cite[%
Proposition 5.3.7]{BratteliRobinson}.
\end{proof}

\noindent This lemma implies that, if $\varrho $ is the unique $(\tau ,\beta
)$--KMS state with $\tau $\ being time--reversal invariant, then $\varrho $
is time--reversal invariant.

Let
\begin{equation*}
\mathcal{X}_{+}:=\left\{ B=B^{\ast }\in \mathcal{X}:\Theta \left( B\right)
=B\right\} \ ,\quad \mathcal{X}_{-}:=\left\{ B=B^{\ast }\in \mathcal{X}%
:\Theta \left( B\right) =-B\right\} \ .
\end{equation*}%
These spaces are closed real subspaces of $\mathcal{X}$. Furthermore, they
are \emph{real} pre--Hilbert spaces w.r.t. the Duhamel two--point function $%
(\cdot ,\cdot )_{\sim }$ defined by (\ref{def bogo jetmanbisbis}).

\begin{lemma}[$\mathcal{X}_{\pm }$ as real pre--Hilbert spaces]
\label{Lemma LR dia copy(2)}\mbox{
}\newline
Assume that $\tau $\ is time--reversal invariant and $\varrho $ is a
time--reversal invariant $(\tau ,\beta )$--KMS state defining the Duhamel
two--point function $(\cdot ,\cdot )_{\sim }$. Then, for all $B_{1},B_{2}\in
\mathcal{X}_{-}$ and all $B_{3},B_{4}\in \mathcal{X}_{+}$,
\begin{equation*}
(B_{1},B_{2})_{\sim }=(B_{2},B_{1})_{\sim }\in \mathbb{R\quad }\text{and}%
\mathbb{\quad }(B_{3},B_{4})_{\sim }=(B_{4},B_{3})_{\sim }\in \mathbb{R}\ .
\end{equation*}
\end{lemma}

\begin{proof}
For any $B_{1},B_{2}\in \mathcal{X}_{-}$, one clearly has
\begin{equation*}
(B_{1},B_{2})_{\sim }=(\Theta \left( B_{1}\right) ,\Theta \left(
B_{2}\right) )_{\sim }\ .
\end{equation*}%
Thus, we have to prove that
\begin{equation*}
(\Theta \left( B_{1}\right) ,\Theta \left( B_{2}\right) )_{\sim
}=(B_{2},B_{1})_{\sim }\ ,\qquad B_{1},B_{2}\in \mathcal{X}_{-}\ .
\end{equation*}%
By the Phragm{\'e}n--Lindel\"{o}f theorem \cite[Proposition 5.3.5]%
{BratteliRobinson}, the stationarity of KMS states and Definition (\ref{def
bogo jetmanbisbis}), it suffices to show that
\begin{equation*}
\varrho \left( \Theta \left( B_{1}\right) \tau _{t}(\Theta \left(
B_{2}\right) )\right) =\varrho \left( B_{2}\tau _{t}(B_{1})\right)
\end{equation*}%
for all $t\in {\mathbb{R}}$ and every $B_{1},B_{2}\in \mathcal{X}_{-}$. In
fact, by the time--reversal invariance of $\varrho $, the stationarity of
KMS states and the hermiticity of states,%
\begin{equation*}
\varrho \left( \Theta \left( B_{1}\right) \tau _{t}(\Theta \left(
B_{2}\right) )\right) =\overline{\varrho \left( B_{1}\tau
_{-t}(B_{2})\right) }=\overline{\varrho \left( \tau _{t}\left( B_{1}\right)
B_{2}\right) }=\varrho \left( B_{2}\tau _{t}\left( B_{1}\right) \right) \ .
\end{equation*}%
As\ $(\cdot ,\cdot )_{\sim }$ is a sesquilinear form, we thus have
\begin{equation*}
(B_{1},B_{2})_{\sim }=\overline{(B_{2},B_{1})_{\sim }}=(B_{2},B_{1})_{\sim
}\in \mathbb{R}\ ,\qquad B_{1},B_{2}\in \mathcal{X}_{-}\ .
\end{equation*}%
The assertion for $\mathcal{X}_{+}$ is proven in the same way.
\end{proof}

This lemma can be generalized for time--dependent Duhamel correlation
functions. To this end, we show the following assertions:

\begin{lemma}[Commutators and Duhamel correlation functions]
\label{lemma conductivty4 copy(1)}\mbox{
}\newline
Let $\varrho $ be a $(\tau ,\beta )$--KMS state defining the Duhamel
two--point function $(\cdot ,\cdot )_{\sim }$. Then, for any $B_{1},B_{2}\in
\mathcal{X}$ and all $t\in \mathbb{R}$,%
\begin{equation*}
\int\nolimits_{0}^{t}\varrho \left( i[B_{1},\tau _{s}(B_{2})]\right) \mathrm{%
d}s=(B_{1},\tau _{t}(B_{2}))_{\sim }-(B_{1},B_{2})_{\sim }\ .
\end{equation*}
\end{lemma}

\begin{proof}
It is an obvious consequence of Theorem \ref{thm auto--correlation upper
bounds copy(1)}. The assertion can also be deduced from \cite[Theorem II.5]%
{Nau2}. We give here another proof because some of its arguments are used
elsewhere in the paper.

By assumption, for any $B_{1},B_{2}\in \mathcal{X}$, the map from $\mathbb{R}
$ to $\mathbb{C}$ defined by%
\begin{equation*}
t\mapsto \varrho \left( B_{1}\tau _{t}(B_{2})\right)
\end{equation*}%
uniquely extends to a continuous map
\begin{equation*}
z\mapsto \varrho \left( B_{1}\tau _{z}(B_{2})\right)
\end{equation*}%
on the strip $\mathbb{R+}i[0,\beta ]$, which is holomorphic on $\mathbb{R+}%
i(0,\beta )$. The KMS property of $\varrho $, that is,
\begin{equation}
\varrho (B_{1}\tau _{t+i\beta }(B_{2}))=\varrho (\tau _{t}(B_{2})B_{1})\
,\qquad B_{1},B_{2}\in \mathcal{X}\ ,\ t\in \mathbb{R}\ ,
\label{KMS property}
\end{equation}%
implies that, for any $B_{1},B_{2}\in \mathcal{X}$ and $t\in \mathbb{R}$,
\begin{equation*}
\varrho \left( \lbrack B_{1},\tau _{t}(B_{2})]\right) =\varrho \left(
B_{1}\tau _{t}(B_{2})\right) -\varrho \left( B_{1}\tau _{t+i\beta
}(B_{2})\right) \ .
\end{equation*}%
As a consequence, by the Cauchy theorem for analytic functions, we obtain
that%
\begin{equation*}
\int\nolimits_{0}^{t}\varrho \left( i[B_{1},\tau _{s}(B_{2})]\right) \mathrm{%
d}s=\int\nolimits_{0}^{\beta }\varrho \left( B_{1}\tau _{t+i\alpha
}(B_{2})\right) \mathrm{d}\alpha -(B_{1},B_{2})_{\sim }
\end{equation*}%
for any $B_{1},B_{2}\in \mathcal{X}$ and $t\in \mathbb{R}$. The group
property of $\tau $ obviously yields
\begin{equation}
\varrho \left( B_{1}\tau _{t+z}(B_{2})\right) =\varrho \left( B_{1}\tau
_{z}(\tau _{t}(B_{2}))\right)  \label{eq Duhamel jetman}
\end{equation}%
for all $z,t\in \mathbb{R}$. On the other hand, the KMS property (\ref{KMS
property}) of $\varrho $ leads to Equation (\ref{eq Duhamel jetman}) for all
$z\in \mathbb{R}+i\beta $. Therefore, we infer from the Phragm{\'{e}}%
n--Lindel\"{o}f theorem \cite[Proposition 5.3.5]{BratteliRobinson} that, for
any $B_{1},B_{2}\in \mathcal{X}$, (\ref{eq Duhamel jetman}) holds true for
all $z\in \mathbb{R}+i[0,\beta ]$. In particular,
\begin{equation}
\int\nolimits_{0}^{\beta }\varrho \left( B_{1}\tau _{t+i\alpha
}(B_{2})\right) \mathrm{d}\alpha =(B_{1},\tau _{t}(B_{2}))_{\sim }\text{ }.
\label{eq caclu a la con}
\end{equation}
\end{proof}

\begin{lemma}[Time--reversal symmetry of commutators]
\label{lemma conductivty4}\mbox{
}\newline
Assume that $\tau $\ is time--reversal invariant and $\varrho $ is a
time--reversal invariant state. Then, for any $B_{1},B_{2}\in \mathcal{X}%
_{-} $ (or $\mathcal{X}_{+}$) and all $t\in \mathbb{R}$,%
\begin{equation*}
\int\nolimits_{0}^{t}\varrho \left( i[B_{1},\tau _{s}(B_{2})]\right) \mathrm{%
d}s=\int\nolimits_{0}^{-t}\varrho \left( i[B_{1},\tau _{s}(B_{2})]\right)
\mathrm{d}s=\int\nolimits_{0}^{t}\varrho \left( i[B_{2},\tau
_{s}(B_{1})]\right) \mathrm{d}s\ .
\end{equation*}
\end{lemma}

\begin{proof}
The first equality follows from the following assertions: For any $%
B_{1},B_{2}\in \mathcal{X}_{-}$ (or $\mathcal{X}_{+}$) and $t\in \mathbb{R}$,%
\begin{eqnarray*}
\int\nolimits_{0}^{-t}\varrho \left( i[B_{1},\tau _{s}(B_{2})]\right)
\mathrm{d}s &=&\int\nolimits_{0}^{-t}\overline{\varrho \circ \Theta \left(
i[B_{1},\tau _{s}(B_{2})]\right) }\mathrm{d}s \\
&=&-\int\nolimits_{0}^{-t}\varrho \left( i[B_{1},\tau _{-s}(B_{2})]\right)
\mathrm{d}s \\
&=&\int\nolimits_{0}^{t}\varrho \left( i[B_{1},\tau _{s}(B_{2})]\right)
\mathrm{d}s\ .
\end{eqnarray*}%
Furthermore, by stationarity of KMS states,
\begin{equation*}
\int\nolimits_{0}^{t}\varrho \left( i[B_{2},\tau _{s}(B_{1})]\right) \mathrm{%
d}s=-\int\nolimits_{0}^{t}\varrho \left( i[B_{1},\tau _{-s}(B_{2})]\right)
\mathrm{d}s=\int\nolimits_{0}^{-t}\varrho \left( i[B_{1},\tau
_{s}(B_{2})]\right) \mathrm{d}s
\end{equation*}%
for any $B_{1},B_{2}\in \mathcal{X}_{-}$ (or $\mathcal{X}_{+}$) and $t\in
\mathbb{R}$.
\end{proof}

\noindent We are now in position to prove a generalization of Lemma \ref%
{Lemma LR dia copy(2)}:

\begin{satz}[Symmetries of Duhamel correlation functions]
\label{Thm important equality asymptotics copy(1)}\mbox{ }\newline
Assume that $\tau $\ is time--reversal invariant and $\varrho $ is a
time--reversal invariant $(\tau ,\beta )$--KMS state defining the Duhamel
two--point function $(\cdot ,\cdot )_{\sim }$. Then, for all $B_{1},B_{2}\in
\mathcal{X}_{-}$ (or $\mathcal{X}_{+}$) and $t\in \mathbb{R}$,
\begin{equation*}
\left( B_{1},\tau _{t}\left( B_{2}\right) \right) _{\sim }=\left( B_{1},\tau
_{-t}\left( B_{2}\right) \right) _{\sim }=\left( B_{2},\tau _{t}\left(
B_{1}\right) \right) _{\sim }\in \mathbb{R\ }.
\end{equation*}
\end{satz}

\begin{proof}
By Lemma \ref{lemma conductivty4 copy(1)},
\begin{equation*}
(B_{1},\tau _{t}(B_{2}))_{\sim }=\int\nolimits_{0}^{t}\varrho \left(
i[B_{1},\tau _{s}(B_{2})]\right) \mathrm{d}s+(B_{1},B_{2})_{\sim }
\end{equation*}%
for all $B_{1},B_{2}\in \mathcal{X}_{-}$ (or $\mathcal{X}_{+}$) and $t\in
\mathbb{R}$. Observe that
\begin{equation*}
\varrho \left( i[B_{1},\tau _{s}(B_{2})]\right) \in \mathbb{R\ },
\end{equation*}%
for all $B_{1},B_{2}\in \mathcal{X}_{-}$ (or $\mathcal{X}_{+}$) and $s\in
\mathbb{R}$, because $B_{1},B_{2}$ are self--adjoint elements of $\mathcal{X}
$. From Lemma \ref{Lemma LR dia copy(2)}, it follows that, for any $%
B_{1},B_{2}\in \mathcal{X}_{-}$ (or $\mathcal{X}_{+}$) and $t\in \mathbb{R}$,%
\begin{equation*}
(B_{1},\tau _{t}(B_{2}))_{\sim }\in \mathbb{R\ }\ .
\end{equation*}%
Moreover, by Lemmata \ref{Lemma LR dia copy(2)} and \ref{lemma conductivty4}%
,
\begin{equation*}
(B_{1},\tau _{t}(B_{2}))_{\sim }=(B_{1},\tau _{-t}(B_{2}))_{\sim
}=(B_{2},\tau _{t}(B_{1}))_{\sim }
\end{equation*}%
for any $B_{1},B_{2}\in \mathcal{X}_{-}$ (or $\mathcal{X}_{+}$) and $t\in
\mathbb{R}$.
\end{proof}

\bigskip

\noindent \textit{Acknowledgments:} We would like to thank Volker Bach,
Horia Cornean, Abel Klein and Peter M\"{u}ller for relevant references and
interesting discussions as well as important hints. JBB and WdSP thank Mr.
and Mrs. Bru for their hospitality, support and interesting discussions. JBB
and WdSP are also very grateful to the organizers of the Hausdorff Trimester
Program entitled \textquotedblleft \textit{Mathematical challenges of
materials science and condensed matter physics}\textquotedblright\ for the
opportunity to work together on this project at the Hausdorff Research
Institute for Mathematics in Bonn. This work has also been supported by the
grants MTM2010-16843, MTM2014-53850 (MINECO), the FAPESP grant
2013/13215--5, the grant IT641-13 (Basque Government), the BERC 2014-2017
program (Basque Government) as well as the BCAM Severo Ochoa accreditation
SEV-2013-0323 (MINECO). Finally, we thank very much the referee for
her/his work and interest in the improvement of the paper.

\end{document}